\documentclass[atmp]{ipart_v1}

\Vol{19}
\Issue{3}
\Year{2015}
\firstpage{701}

\usepackage{t1enc}
\usepackage[latin1]{inputenc}
\usepackage[english]{babel}

\usepackage{amsthm}
\usepackage{yfonts}

\usepackage{bbm}
\usepackage{bm}
\usepackage{mathrsfs}

\usepackage{fullpage}
\usepackage{amsmath,amsfonts,amssymb,amsthm,amscd}
\usepackage{mathabx}
\usepackage{spectralsequences}
\newcommand{\A}{{\mathcal A}}
\newcommand{\B}{{\mathcal B}}
\newcommand{\C}{{\mathcal C}}
\newcommand{\E}{{\mathcal E}}
\newcommand{\F}{{\mathcal F}}
\newcommand{\G}{{\mathcal G}}

\newcommand{\M}{{\mathcal M}}
\newcommand{\MM}{{\mathcal M}{\mathcal M}}
\newcommand{\N}{{\mathcal N}}
\newcommand{\NN}{{\mathcal N}{\mathcal N}}

\newcommand{\T}{{\mathcal T}}

\newcommand{\K}{{\mathcal K}}

\newcommand{\HC}{{\mathcal H}{\mathcal C}}

\newcommand{\KR}{{\mathbb R}}

\newcommand{\KZ}{{\mathbb Z}}

\newcommand{\KT}{{\mathbb T}}

\newcommand{\Simp}{{\mbox{\bf Simp}}}
\newcommand{\Geom}{{\mbox{\bf Geom}}}
\newcommand{\Core}{{\mbox{\bf Core}}}
\newcommand{\Grpd}{{\mbox{\bf Grpd}}}

\newcommand{\id}{{\mbox{id}}}
\newcommand{\Aut}{\operatorname{Aut}}
\newcommand{\im}{\operatorname{im}}
\newcommand{\be}[0]{\begin{equation}}
\newcommand{\ee}[0]{\end{equation}}
\numberwithin{equation}{section}

\theoremstyle{plain}
\newtheorem{theorem}{Theorem}[section]
\newtheorem{lemma}[theorem]{Lemma}

\newtheorem{corollary}{Corollary}[section]
\newtheorem{definition}{Def}[section]

\AddToHook{begindocument/before}{\RequirePackage{nameref}}

\begin{document}

\title[Higher Topological T-duality]{On Higher Topological T-duality Functors}

\author[Ashwin S. Pande]{Ashwin S. Pande}

\begin{abstract}
We use String Field Theory (SFT) to 
construct a higher analogue of Bunke-Schick's 
functor $P: \mathbf{Top}^{op} \to \mathbf{Set}$ 
\cite{BunkeS1} by geometrizing $P.$
We use the projection of SFT onto its massless modes 
\cite{SFTDiffeo} to construct the category $\C$ 
whose objects are pairs (which we identify with
SFT backgrounds) and whose maps are morphisms of pairs
(which are gauge transformations). 
Using $\C$ and categorical equivalence, for any
$CW-$complex $X$ we define the moduli space $G(X)$ 
of SFT backgrounds which are pairs over $X$
up to gauge equivalence. 
We use the homotopy theory of the moduli space $G(X)$
to define functors on the category of $CW-$complexes
$P_k:\mathbf{CW}^{op} \to \mathbf{Grpd}$
such that $P_0 \simeq P,$ $P_1$ is nontrivial
and $P_k(X)$ are always trivial for $k \geq 2.$ 
Arrows in $P_1(X)$ are shown to be
isotopy classes of maps in the mapping class group of $X$
acting on (isomorphism classes of) pairs over $X.$
We discuss applications to Topological T-duality
for triples and to modelling doubled geometries
and T-folds \cite{HullT}.
\end{abstract}
\maketitle
\section{Introduction \label{SSecTTDP}}
In this paper we define a higher version of
Bunke-Schick's functor $P(X)$ (where
$P:\mathbf{Top}^{op} \to \mathbf{Set}$ and $X$
is a topological space) using the
homotopy theory of the moduli space of
String Field Theory backgrounds over $X$ 
which we construct
from the category of all pairs over $X.$

We begin this Section with a review of
the T-duality in Type II String Theory. We give
a brief outline of the Topological T-duality
formalisms of Mathai and Rosenberg \cite{MR1}
and Bunke and Schick \cite{BunkeS1}. 

We then give an abbreviated summary of
the main results of this paper and their
physical relevance.

After this, we present a detailed discussion of the outline
of this paper consisting of 
a SubSection by SubSection summary of the work in
this paper. We end this Section with a brief
introduction to closed bosonic String Field Theory
since we will use ideas from SFT in later Sections.

\subsubsection{T-duality and Topological T-duality}
Topological T-duality was defined by Mathai and Rosenberg in Ref.\  \cite{MR1} as a topological model for the T-duality symmetry of Type II String Theory. The theory of Mathai and Rosenberg makes a mathematical model of the T-duality symmetry of Type II String Theory backgrounds $E$ with free circle action with $H-$flux using methods from noncommutative topology. The background $E$ is then automatically a principal circle bundle
$p:E \to X$ with $X \simeq E/S^1$ the base space of the bundle and the $H-$flux an integral cohomology class $H$ in $H^3(E,\KZ).$

It is well known in String Theory that Type II String Theory on a background $E$ with
a free circle action with $H-$flux $H$ is physically equivalent to Type II String Theory on a `T-dual' background $E^{\#}$ 
with a {\em different} free circle action and a {\em different} $H-$flux $H^{\#}.$ 
Mathai and Rosenberg in Ref.\ \cite{MR1} argued that T-duality could be studied 
using ideas from the theory of crossed products of continuous-trace  $C^{\ast}-$algebras 
naturally associated to such backgrounds. 

The theory of Mathai and Rosenberg (\cite{MR1}) 
associates a $C^{\ast}-$dynamical system 
$(\A, \KR, \alpha)$ to a given spacetime background
$E \to X$ where $\A$ is a continuous-trace 
$C^{\ast}-$algebra with spectrum $\hat{A} \simeq E$ and 
$\alpha$ is an $\KR-$action on $\A$ inducing the natural
circle action on $\hat{\A}.$
The Dixmier-Douady invariant of $\A$ is chosen to be $H.$
It is well known (see Ref.\ \cite{JMRCBMS} and references therein) that the crossed product of the continuous-trace $C^{\ast}-$algebra $\A$ by the
$\KR-$action $\alpha$ in the $C^{\ast}-$dynamical system above is 
another $C^{\ast}-$dynamical system of the form 
$(\A^{\#}, \KR^{\#}, \alpha^{\#}).$ 
Bouwknegt, Evslin and Mathai in Ref.\ \cite{BEM-TTD}
and Mathai and Rosenbeg in Ref.\ \cite{MR1} 
identified the spectrum of the 
$C^{\ast}-$algebra $\A^{\#}$ with the Topological T-dual
spacetime $E^{\#}.$ In addition they identified the 
Dixmier-Douady invariant of $\A^{\#}$ with the dual
$H-$flux $H^{\#}$ on the T-dual spacetime. 

It is strange that the topological class of the bundle $p:E^{\#} \to X$ calculated using the crossed product construction (as described in the previous paragraph) always agrees with the topological class
of the physical T-dual spacetime obtained from String Theory calculations for
all known examples. Also it turns out that in all these physical examples  $H^{\#}$ is the field strength of the $B-$field on 
the physical T-dual spacetime.

Hence under the action of Topological T-duality a principal circle 
bundle $p:E \to X$ together with a $H-$flux $H \in H^3(E,\KZ)$ is 
transformed into a dual principal circle bundle $p^{\#}:E^{\#} \to X$ 
together with a dual $H-$flux $H^{\#} \in H^3(E^{\#},\KZ)$ on it. This transformation is termed the Topological T-duality Transformation or
Topological T-duality in short.

Motivated by the work of Mathai and Rosenberg 
(Ref.\ \cite{MR1}), Bunke and Schick in Ref.\ \cite{BunkeS1}
derived the Topological T-duality transformation using
methods from Algebraic Topology. In Ref.\ \cite{BunkeS1}
the authors define 'pairs' over a topological space --- where
a pair over a topological space $X$ consists of
a principal circle bundle $p:E \to X$ together with a class
$H\in H^3(E,\KZ).$

The authors define a functor  
$P: X \to$ \{ {\bf Equivalence classes of Pairs over $X$} \} 
on the category $\mathbf{Top}$ with objects topological
spaces and arrows all continuous maps which associates to
each topological space $X$ the set of all 
pairs over $X.$ They show that $P$ is a 
$\mathbf{Set}-$valued presheaf on $\mathbf{Top}$ i.e., 
a functor 
$P:\mathbf{Top}^{op} \to \mathbf{Set}.$

The authors show that $P$ is a representable functor and 
in particular it has a classifying space. 
They further show that there is a natural homotopy
automorphism of the classifying space of this functor. 
The authors show that this automorphism and the 
self-transformation it induces of the functor 
$P(X)$ agrees with the Topological T-duality transformation
of isomorphism classes of pairs consisting of a 
principal circle bundle over $X$ together with a $H-$flux
studied in Refs.\ \cite{BEM-TTD, MR1}. 

\subsubsection{Main Results and Physical Relevance}
The aim of this paper is to define a higher version of
Bunke-Schick's functor $P(X)$ (where
$P:\mathbf{Top}^{op} \to \mathbf{Set}$ and $X$
is any topological space) using the
homotopy theory of the moduli space of
String Field Theory backgrounds over $X.$
We construct this moduli space by geometrizing
$P$ using the category of elements construction.

We now outline the main results of this paper and 
explain their connection to String Theory.

To construct the higher version of Bunke-Schick's 
functor we first define the category whose objects
are pairs with morphisms of pairs as arrows.
In Sec.\ (\ref{SecIntro}) we define the category
of pairs $\C$ based on the work of Ref.\ \cite{BunkeS1}.
We also note that there is a natural forgetful functor
$F:\C \to \mathbf{Top}$ which sends a pair to its base space.

We argue using String Theory that we may
identify pairs with SFT backgrounds and morphisms
between pairs as SFT gauge transformations which don't
distort the background too much. We identify $\C$
with a topological approximation to the low-energy 
effective theory obtained from SFT as in Ref.\ \cite{SFTDiffeo}.

%
%
%
In Sec.\ (\ref{SecGXHot}) we begin by
geometrizing Bunke-Schick's functor $P.$
We then use this to define the moduli space $G(X)$
which we will use to construct the functors $P_0, P_1.$

We first define the category $\HC$ whose objects
consist of homotopy equivalence classes of 
pairs and whose morphisms are induced from
morphisms in $\C$ and show that there is a natural
forgetful functor $\pi: \C \to \HC.$ 
We demonstrate that $F \simeq E \circ \pi$
and that $E:\HC \to \mathbf{Top}$
is the category of elements of 
$P:\mathbf{Top}^{op} \to \mathbf{Set}$
hence the functor $P$
has a geometric interpretation as the fibration 
$E:\HC \to \mathbf{Top}.$

We use the idea of categorical equivalence of fibrations of 
categories to construct for any $CW-$complex $X$ 
the moduli space $G(X)$ of SFT backgrounds which
are pairs over $X.$ The points in this moduli space 
consist of homotopy equivalence
classes of pairs over $X$ and hence to SFT gauge
equivalence classes of SFT backgrounds with $H-$flux 
which are principal circle bundles over $X.$ Paths in 
$G(X)$ correspond to those SFT gauge transformations
between pairs corresponding to morphisms of pairs 
covering a self-homeomorphism of
the base $X.$ These are exactly the SFT gauge
transformations which don't distort the background too
much. 

In Sec.\ (\ref{SecP_i}) we prove that the homotopy
theory of the moduli space $G(X)$ then gives us invariants
$P_k(X)$ of $X$ such that $P_0(X)$ is Bunke-Schick's
functor $P(X)$ restricted to the category of $CW-$complexes
and cellular maps, $P_1:\mathbf{CW}^{op} \to \mathbf{Grpd}$ 
is a higher version of Bunke-Schick's
functor and the $P_k(X)$ are trivial if 
$k \geq 2.$ We geometrize $P_1$ by proving that
elements of $P_1(X)$ have a geometric interpretation
as lifts of the action of the mapping class group
of $X$ denoted $MCG(X)$ on isomorphism classes
of pairs over $X.$
We examine $P_1(X)$ when $X$ is a surface and $X$ 
is a three-manifold in particular, a knot complement.
Thus, $P_0$ takes isomorphic values on homotopy
equivalent $CW-$complexes. However, $P_1$ takes
values which are equivalent groupoids on 
isotopy equivalent $CW-$complexes.

We conclude this paper by discussing three possible
extensions of the ideas in this paper in Secs.\ 
(\ref{SecPropP0P1}, \ref{SecP32P1}, \ref{SecTFold})
respectively.

\subsubsection{Outline of Paper}
We now present a detailed outline of this paper 
with section and subsection
numbers in the paragraphs following this paragraph
up to the end of the section. 
We also give a brief introduction to closed bosonic
String Field Theory (SFT) at the end of this section.
We will use arguments from SFT throughout the paper.

{\flushleft{\bf{Outline of Sec.\ (\ref{SecIntro}):}}}
In Sec.\ (\ref{SecIntro}) of 
this paper we use Ref.\ \cite{BunkeS1} to define the 
category which we denote by $\C$ whose objects
are pairs over an 
arbitrary space $X$ in ${\bf Top}$ consisting of a principal
circle bundle $p:E \to X$ and a $H-$flux $H \in H^3(E,\KZ).$
The morphisms between two pairs are equivariant maps between the
underlying bundles of those two pairs
which induce maps on the cohomology
of the total spaces of the bundles which
map the $H-$flux on one pair to the $H-$flux on the other.

We show in SubSec.\ (\ref{SSecDefCCat}) below 
that every morphism of pairs is
a composite of some number of
isomorphisms of a pair over the same base with
pullbacks of a pair over a space along arbitrary continuous
maps between two topological spaces.
We show that we can define the category $\C$ above 
using either definition of morphism of pairs.
We show that there is a natural forgetful
functor $F: \C \to \mathbf{Top}$ which
sends a pair $(E,H)$ over $X$ to the base space $X.$
We also study morphisms of pairs over two different
bases and over the same base in detail in 
SubSec.\ (\ref{SSecMorC}).

Physically, we view the objects in the 
category $\C$ as consisting of all possible Type II
flux backgrounds with sourceless $H-$flux.
All these backgrounds possess a free circle action. 
We also assume Type II String Theory propagates on 
these backgrounds.

We describe this situation using ideas from
String Field Theory (SFT), in particular closed bosonic
SFT (see Ref.\ \cite{Horowitz}
for a general introduction).
The massless modes of SFT consist of the graviton
and the $B-$field. In Sec.\ (2.1) of Ref.\ \cite{SFTDiffeo}, 
a consistent reduced 
description for closed bosonic SFT was derived in which
the effective action for the
massless modes of SFT was calculated by
projecting the full String Field Theory onto its massless
modes by integrating out the massive modes. 

We argue that
{\em objects}  in the category $\C$ are exactly the valid backgrounds
for the projected SFT (see SubSec.\ (\ref{SSecSFT-TTD})
for details) since each pair keeps only the
topological information (i.e., the bundle characteristic class and
value of $H-$flux) about the massless modes of 
closed bosonic SFT.

The SFT action has infinitesimal
gauge symmetries which form a $L_{\infty}-$algebra.
In the reduced description in Ref.\ \cite{SFTDiffeo}
a reduced set of gauge symmetries are present.
In this paper we consider finite (or exponentiated)
gauge symmetries of the reduced SFT. 
We argue that {\em morphisms} in $\C$ may be identified with
background preserving
and background changing exponentiated
gauge transformations of the reduced SFT 
(see SubSec.\ (\ref{SSecSFT-TTD})
for details).

{\flushleft{\bf{Outline of Sec.\ (\ref{SecGXHot}):}}}
In Sec.\ (\ref{SecGXHot}) we first geometrize Bunke-Schick's
functor $P$ by showing it is equivalent to a categorical
fibration. Then we construct a categorically equivalent
presheaf of groupoids on $\mathbf{Top}$ from a 
generalization of this fibration.

We first construct a category whose objects are 
homotopy equivalence classes of 
pairs denoted $\HC$ and whose morphisms are induced
from $\C.$ We show that there is a natural map $\pi: \C \to \HC.$
We geometrize Bunke-Schick's functor
$P$ by demonstrating that there is a categorical
fibration $E: \HC \to \mathbf{Top}$ which has
$P$ as its category of elements.
 
Since the Bunke-Schick
functor $P$ is a presheaf of sets on $\mathbf{Top},$
it naturally defines a functor
$E: \HC \to \mathbf{Top}$ as its category of
elements (see Ref.\ \cite{nLab-CatEl}) and we show
that $F \simeq E \circ \pi.$ The fiber of the functor
$E$ over a topological space $X$ is, by definition of 
the category of elements, a small category whose
set of objects is the set $P(X)$ and whose only
arrow is the identity.
Due to this it is natural to view objects in $\HC$ as defining
a generalization of the idea of a topological
space.

Hence, for any topological space $X,$
the subcategory of $\HC$ mapping to $X$ under
$E$ should be viewed as equivalent
to $X$ (see discussion after Thm.\ (\ref{ThmHCTEqv}) below).
This subcategory is the subcategory $\pi(\C_X)$ of $\HC$
where $\C_X$ is the subcategory of $\C$ consisting of all
pairs in $\C$ over $X.$

We cannot use this category in what follows for technical reasons
(see SubSubSec.\ (\ref{SSecNonTrivP_k}) below) so
we use a subcategory of this category (termed $\pi(\F_X)$)
which has the same objects but with only iso-arrows as follows:
We identify the objects of $P(X)$ with
SFT vacua which are pairs up to 
SFT gauge transformations
induced by bundle gauge transformations
of the underlying circle bundle of the pair.
Since the category $P(X)$ is a {\em set} these
vacua have only {\em identity} gauge transformations 
as morphisms. It is natural to define a subcategory
of $\HC$ (denoted $\pi(\F_X)$ below)
which has the same objects as $P(X)$
but has arrows between objects which
are images under $\pi$ of isomorphisms
of pairs in $\C$ covering a self-homeomorphism
$f:X \to X.$ This is the subcategory $\pi(\F_X)$ above.

Physically, the objects of $\pi(\F_X)$ correspond to 
SFT vacua which are homotopy equivalence
classes of pairs over $X$ with equivalence classes
of SFT gauge transformations between these
vacua which are isomorphisms of pairs covering
a nontrivial self-homeomorphism $f:X \to X$ as arrows.
The arrows of $\pi(\F_X)$ then must
correspond to the action of SFT gauge transformations on
these backgrounds which preserve the fiber structure
and do not the underlying background too much.
We show in this Section that this class of
gauge transformations and backgrounds emerges
naturally from SFT. 

The category $\F_X$ used in constructing
the extension $\pi(\F_X)$ of $P(X)$ in the previous
paragraph has a very
natural mathematical definition as the 
essential fiber category of the functor 
$F:\C \to \mathbf{Top}.$ 
We prove in Lem.\ (\ref{LemPiFXGrpd})
and Thm.\ (\ref{ThmPiFXDef}) that the essential
fiber of $E$ is the small subcategory $\pi(\F_X)$
of $\HC.$ We may view $\pi(\F_X)$ as a categorification
of Bunke-Schick's functor $P(X)$ which turns the
{\em set} $P(X)$ into the {\em small groupoid} $\pi(\F_X).$

In this paper we principally study the assignment
$X \mapsto \pi(\F_X).$ We show in 
SubSec.\ (\ref{SSecNMFE})
that we may use this assignment to define a higher version of 
Bunke-Schick's functor $P.$ To do this we extend
the assignment $X \mapsto \pi(\F_X)$ above
to a functor $\M: \mathbf{Top}^{op} \to \mathbf{Grpd}$
and prove that $\M$ is categorically equivalent to
the functor $E$ above. In Sec.\ (\ref{SecP_i})
we argue that $\M$ is a generalization of $P.$

In SubSec.\ (\ref{SSecNMFE}) we discuss the construction
of the functor $\M$ from $E$ in detail using the
analogue of the Grothendieck Construction for Street Fibrations.
The value of $\M$ on a $CW-$complex $X$ is a 
small groupoid $\M(X)$ which is equivalent to $\pi(\F_X).$
Recall that the Bunke-Schick functor $P$ is
a $\mathbf{Set}-$valued presheaf on $\mathbf{Top}$
and the underlying set 
of objects of the groupoid $\M(X)$ may be naturally
identified with the set $P(X).$ Thus$\M(X)$ is a
`categorification' of Bunke-Schick's functor $P(X).$ 

{\flushleft{\bf{Outline of Sec.\ (\ref{SecP_i}):}}}
In this Section we assume that $X$ is a
$CW-$complex. This is because we study
the homotopy theory of a classifying space
constructed from $\M(X)$ and for this it
is essential for $\M(X)$ to be a small category.
For this, we need to assume that $X$ is a 
$CW-$complex or we cannot bound the
cardinality of the space of pairs over $X.$

In this Section we prove that the homotopy 
theory of the moduli space $G(X)$ obtained from
$\M(X)$ by simplicial realization (see below)
gives us invariants $P_k(X)$ of $X$ such
that $P_0(X)$ is Bunke-Schick's functor 
$P(X),$ $P_1(X)$ is a higher version of 
Bunke-Schick's functor and the $P_k(X)$
are trivial if $k \geq 2.$

We had noted above that $\M(X)$ should be
viewed as a categorified version of $P(X).$ 
However, since $P(X)$ is a set and $\M(X)$ is
a groupoid it is not clear how the two functors
are related to each other or how to understand
and to interpret physically any
information contained in $\M(X).$

We do this in Sec.\ (\ref{SecP_i}) by using the functor
$\M$ and the groupoid $\M(X)$ above
to construct a `higher' Topological T-duality functor by
passing to the simplicial category $\Simp(\M(X))$ 
associated to the groupoid $\M(X).$ 
We show that this gives us a natural assignment 
$X \mapsto \M(X) \mapsto |\Simp(\M(X))|$
where $|.|$ denotes the geometric realization functor of
a simplicial category.
We show that this defines a new functor 
$G: \mathbf{CW}^{op} \to \mathbf{Top}$
such that $G(X)$ is the geometric realization of the
simplicial classifying space of the small category
$\M(X)$ for every $X$ in ${\mathbf{CW}}.$ Here, $X$
must be a $CW-$complex for technical reasons.

We show in Sec.\ (\ref{SecP_i})
that Bunke-Schick's functor $P(X)$ surjects 
onto $\pi_0(G(X))$ by a map $\phi: P(X) \to \pi_0(G(X)).$
We argue in Sec.\ (\ref{SecP_i}) that is is possible
to naturally define $P_i(X), i >0$ as
functors obtained from the higher homotopy
groupoids $\Pi_i(G(X)), i >0$ (see Ref.\ \cite{GraVit}) by
$2-$pullback \cite{nLab-2Pbck} along $\phi.$
We prove that we can obtain exactly
two nontrivial functors $P_0$ and $P_1$ using this 
method, and all other functors $P_k, k \geq 2$ are
trivial. We then identify $P_1(X)$ with a 'higher' 
Topological T-duality functor.

We had geometrized Bunke-Schick's functor $P_0(X)$ (see discussion after
Thm.\ (\ref{ThmHCTEqv})) by identifying $E:\HC \to \mathbf{Top}$
with the category of elements of $P_0(X).$
We geometrize $P_1(X)$ by proving that elements of $P_1(X)$
correspond to actions of the mapping class group of $X$
(denoted $MCG(X)$) on isomorphism classes of pairs over $X.$
We prove that elements of $P_1(X)$ may be identified 
with isomorphism classes of pairs over $X$ which 
remain the same when pulled back along an
{\em isotopy class} of self-homeomorphisms of $X.$ 
Hence, $MCG(X)$ naturally acts on isomorphism
classes of pairs over $X.$ We give examples of 
pairs over $\KT^2$ and pairs over a knot complement.

We note that $P_0(X)$ is a homotopy invariant
functor on the category ${\mathbf{CW}}$ i.e., it
takes the same value on homotopy equivalent spaces. 
However, due to the above argument, 
$P_1(X)$ is an isotopy invariant functor on $\mathbf{CW},$
that is, its values on spaces which differ by an isotopy equivalence
are naturally equivalent groupoids.

In SubSec.\ (\ref{SecPropP0P1}) we show that
unlike $P(X)$ and $\M(X)$ we can compare 
$P_1(X)$ with $P_1(A)$ for any $A \subseteq X$
due to an exact sequence joining $P_1(X)$ and $P_1(A).$ 

\subsubsection{Introduction to SFT}
Throughout the paper, we also point out an interesting
connection of the above view of Topological T-duality
with closed bosonic SFT 
(see SubSec.\ (\ref{SSecSFT-TTD}) for more details). 

We now discuss this connection in detail.
The category $\C$ described above has as objects pairs 
and has as arrows morphisms of pairs. 
Thus the category $\C$ has as objects topological
approximations to String Theory backgrounds which are
principal circle bundles over a base with $H-$flux. Hence to
describe $\C$ physically we need a way to describe all such
backgrounds. 

Closed bosonic SFT provides a way to describe all
these String backgrounds at once and all transformations
between them since it is a background independent way 
of describing String Theory (see Ref.\ \cite{KZw, Horowitz}).

Closed bosonic SFT is written in terms
of a String Field $\Psi$ and an action for $\Psi$
(see SubSec.\ (\ref{SSecSFT-TTD}) below).
The theory has a large number of infinitesimal symmetries,
which form an $L_{\infty}$ Algebra (see 
Ref.\ \cite{nLab-LInf}). The $L_{\infty}-$algebra structure 
defines a commutator bracket of two symmetries and
also higher commutator brackets. The defining
relations of the $L_{\infty}-$algebra consist
of Jacobi-like identities between lower commutator
bracktes of the infinitesimal symmetry generators
such that the right-hand side of the Jacobi-like
identity do not equal zero but equal a higher 
commutator bracket.

Thus infinitesimal gauge transformations in
SFT do not form a gauge group
when exponentiated unlike a gauge field theory
since all commutator brackets in the theory
only commute up to higher brackets.

In this paper we consider closed bosonic String
Field Theory backgrounds and exponentiated
gauge transformations between them. These
gauge transformations have not been studied
much, but we do not study the entire set of
transformations but only a restricted one.

Mazel and coworkers in Ref.\ \cite{SFTDiffeo} Sec.\ (2.1) 
derived a consistent reduced description of closed
bosonic SFT projected
onto its massless modes for spacetimes with
small curvature and weak field strengths. The
massless modes for closed bosonic SFT consist of
the graviton and the $B-$field.

Hence in this paper we identify
the objects of the category $\C$
with these backgrounds and the 
morphisms in $\C$ with
certain types of gauge transformations of the reduced SFT.

We study exponentiated transformations of classical
closed bosonic SFT which do not
deform the background too much. It can be shown
(see Ref.\ (\cite{SFTDiffeo})) that these gauge
transformations correspond to self-diffeomorphisms
of the spacetime background which move points, 
to maps between two different spacetime backgrounds 
of differing characteristic class 
(see Ref.\ (\cite{KZw}) for the
case of a toroidal background) and also to gauge
transformations of the $H-$flux on the background.

In SubSec.\ (\ref{SSecSFT-TTD})
we identify objects in $\C$ with the subset of string
backgrounds in SFT which
are principal circle bundles over a base 
(the base then corresponds to the 
uncompactified directions of the string background) 
as in Refs.\ \cite{D-Hull, HullT}. We also 
argue that topological approximations to some 
background-changing gauge transformations of 
SFT give maps between
backgrounds and correspond to morphisms in $\C.$

We cannot study all the gauge transformations
of the reduced SFT between two given spacetime
backgrounds in this paper. In this paper we mainly
study gauge transformations of the reduced SFT
backgrounds corresponding to isomorphisms 
of pairs over a fixed base $X$ covering the identity
and isomorphisms of pairs over $X$ covering a fixed
self-homeomorphism of $X.$ We study these types of
morphisms of pairs in detail in Sec.\ (\ref{SSecMorC}).
In that Section and in Sec.\ (\ref{SecPropP0P1})
we also remark on morphisms of pairs covering an arbitrary
map $f:X \to X.$

We argue in Sec.\ (\ref{SecP_i})
that both Bunke-Schick's functor $P$ and 
the higher Topological T-duality functor defined in this paper
correspond to a gauge equivalence class of data obtained
from SFT on the background $E$ under 
{\em background-preserving} gauge transformations 
(see Ref.\ \cite{KZw, Horowitz}).

We conclude this paper by discussing three possible
extensions of the ideas in this paper in Secs.\ 
(\ref{SecPropP0P1}, \ref{SecP32P1}, \ref{SecTFold})
respectively. In Sec.\ (\ref{SecPropP0P1}) we discuss a
natural extension of the formalism in this paper which
lets us define more nontrivial invariants of a space. 
In this extension the invariants $P_k, k \geq 2$ need not 
be zero automatically. We also propose a second way of
extending the functors $P_k$ using Thomason
Cohomology so that functors similar to the $P_k$
may be defined for spaces with extra data (i.e., spaces
with not just background sourcless $H-$flux but other
similar data like background sourceless Ramond-Ramond 
flux).

In Sec.\ (\ref{SecP32P1}) we relate the functors
$P_0, P_1$ defined in this paper to the functor
$P_{3,2}$ defined in Ref.\ (\cite{Pan2}). 

It is well known from the work of Hull and Zwiebach in
Ref.\ \cite{HullZ} that SFT on toroidal
backgrounds must be described using the
original coordinates and the `dual' coordinates
simultaneously. This space is termed a doubled geometry in
Refs.\ \cite{HullZ, KZw}. 

This doubled geometry may be identified with
the correspondence space in the `Diamond Diagram' 
of Topological T-duality (see Refs.\ \cite{BEM,MR1,BunkeS1}
for details about the Diamond Diagram and the
correspondence space). In Sec.\ (\ref{SecTFold}) 
we argue that the correspondence space of
Refs.\ (\cite{MR1,BunkeS1}) may be directly constructed
using the formalism described in this paper. 
We further argue that the T-fold
geometries of Dabholkar and Hull (\cite{D-Hull, HullT}) 
may also be naturally constructed by a generalization of
the ideas in this paper.

\section{ Morphisms in the Category of Pairs \label{SecIntro}}
In this section we will construct the category of pairs $\C$ using
the pairs and morphisms of pairs 
described Ref.\ \cite{BunkeS1}. As we had discussed in 
SubSec.\ (\ref{SSecTTDP}) above we view a pair $(E,H)$
over $X$ as a SFT
background $E$ (with background $H-$flux $H$)
which is circle compactified over a spacetime $X.$ 

More precisely in SubSec.\ (\ref{SSecSFT-TTD})
we will argue that the objects of $\C$ i.e., the data of a 
pair $(E,H)$ label SFT backgrounds with 
a free circle isometry whose quotient by the circle isometry
is the `base' spacetime $X.$
The morphisms in the category $\C$ then correspond to
certain background-preserving or background-changing
gauge transformations of the SFT. These
morphisms cannot correspond to all such gauge
transformations, since such gauge tranformations need
not respect the circle isometry, while all the morphisms in 
$\C$ do respect this.

In SubSec. (\ref{SSecDefCCat}) below, we discuss the construction
of the category of pairs $\C.$ 
In SubSec. (\ref{SSecMorPf1}) below, we discuss some properties arbitrary
morphisms of pairs.
In SubSec. (\ref{SSecMorC}) below, we compute all morphisms of pairs 
covering all possible maps $f:X \to Y.$
We will use the results in these three subsections when
we construct higher Topological T-duality functors in 
Sec. (\ref{SecGXHot}).
\subsection{Definition of the Category of Pairs $\C$}
\label{SSecDefCCat}
In Ref.\ \cite{BunkeS1} Bunke and Schick defined a pair over a topological space---a pair over a space $X$ consists of a principal circle bundle $E \to X$ with a class $H \in H^3(E, \KZ).$
The authors of Ref.\ \cite{BunkeS1} used this definition of pairs over a
topological space  to construct a functor $P$ on the category of
topological spaces and maps such that $P(X) =${\{\bf Equivalence classes of pairs over $X$\}} and proved that this functor is representable. 
They further showed that Topological T-duality for circle bundles may be defined purely geometrically as a homotopy automorphism of the classifying space of this functor and the induced automorphism of 
$P$ satisfies all the properties (including induced isomorphisms in $K-$theory) that are satisfied by the formalism of Topological T-duality of Mathai and Rosenberg.

We would like to extend $P$ to higher Topological T-duality functors.
To do this precisely we now define the category $\C$ 
whose objects are the Topological T-duality pairs of Bunke and Schick 
(as in Ref.\ \cite{BunkeS1}) and whose morphisms 
are the morphisms 
of these pairs in Ref.\ \cite{BunkeS1}. 
We define $\C$ formally in 
Def.\ (\ref{DefCCat2}) below. 
We then define a new category
$\C'$ with the same objects as $\C$ but with a
different definition of morphism of pairs in 
Def.\ (\ref{DefCCat1}) below.
We then show that $\C$ and $\C'$ are equivalent
categories in the two 
paragraphs after Def.\ (\ref{DefCCat1}). 
Thus the category $\C$ has
two equivalent definitions due to two different but equivalent 
definitions of morphism. We will use this characterization
of $\C$ in later sections of this paper.

We have a forgetful functor $F:\C \to {\mathbf {Top}}$ which sends a 
pair to its underlying topological space that is given a pair
$(E \to X, H)$ over $X,$ the functor $F$ acts as
$F:(E \to X, H) \mapsto X.$ 
We also have the Topological T-duality functor 
$P:{\mathbf{Top}^{op}} \to {\mathbf{Set}}$ of Ref.\ \cite{BunkeS1}. 
As shown in Sec. (2.4) of Ref.\ \cite{BunkeS1} the
Topological T-duality transformation 
sends an isomorphism class of a pair over $X$ to 
the isomorphism class of a dual pair over $X$ and
is an automorphism of the functor $P.$ We say that 
$P$ possesses the T-duality symmetry (see 
Ref.\ \cite{BunkeS1} for details).

In this paper we use $\C$ and $P$ to construct  
a second contravariant functor which takes values
in groupoids
$P_1: {\mathbf{Top}^{op}} \to {\mathbf{Grpd}}$
possessing the T-duality symmetry.
We characterize $P_1(X)$ in terms
of automorphisms of pairs over $X$ and show that
$P_1$ is also invariant under the action of 
Topological T-Duality. 

For every $X$ in ${\mathbf{Top}}$ we define a subcategory 
$\C_X$ of $\C$ as follows: The objects of $\C_X$ are 
pairs over $X$ and the arrows of $\C_X$ 
are morphisms of pairs over $X$ which cover an arbitrary
map $f:X \to X$ in ${\mathbf{Top}}.$
We use $\C$ and $\C_X$ throughout this paper.

We now define the category $\C$ (also see 
Ref.\ \cite{Pan2}). The objects of $\C$ are pairs over a
topological space and a morphism
of pairs in $\C$ is an equivariant morphism of 
principal circle bundles 
compatible with the $H-$flux as in 
Refs.\ \cite{UMDThesis, Pan2}. 

\begin{definition}
We define the objects of the category $\C$ as
ordered pairs $(E\to X, H)$ consisting of a principal circle bundle
$E \to X$ over a topological space $X$ together with a class termed the 
$H-$flux $H \in H^3(E,\KZ).$ We denote pairs by $(E,H)$ or $(E \to X, H).$

Let $f:Y \to X$ be any continuous map of topological spaces.
Let $(E' \to Y, H')$ be a pair over $Y$ and 
$(E \to X, H)$ be another pair over $X.$ 
We define a morphism of pairs from $(E \to X, H)$ to $(E' \to Y, H')$
 in $\C$ to be a commutative diagram of the form
\begin{equation}
\begin{CD}
E' @>\tilde{f}>> E \\
@VV{p'}V   @VV{p}V \\
Y @>>f> X \label{MorP} \\
\end{CD}
\end{equation}
with $\tilde{f}^{\ast}(H) = H'$ and
where the map $\tilde{f}$ need not be a pullback but need only be
$S^1-$equivariant. 
\label{DefMorPair2}
\label{DefCCat2}

It is clear that with the above definitions $\C$ is a category.
\end{definition}

We now follow 
Ref.\ \cite{BunkeS1, UMDThesis, Pan2} and define a 
morphism of pairs in $\C$ as a composite of two
types of arrows between objects of $\C.$ 

After this we define another category $\C'$ with the same
objects as $\C$ but with a different definition of arrows in
Def.\ (\ref{DefCCat1}) below.  
We show that these two definitions give equivalent categories after Def.\ (\ref{DefCCat1}) below.
\begin{definition}
We define the category of pairs --- denoted $\C'$ --- as follows:
\begin{enumerate}
\item[] The category of pairs $\C'$ has as objects ordered pairs 
$(E\to X, H)$ consisting of a principal circle bundle
$E \to X$ over a topological space $X$ together with a class termed the 
$H-$flux $H \in H^3(E,\KZ).$ We denote the above pair either as
$(E,H)$ or as $(E \to X, H).$

\item[] We define a morphism of pairs in $\C'$ to be generated by
composing isomorphisms of pairs over a fixed base together with pullback
of pairs along an arbitrary continuous map $f:Y \to X.$ 
These two morphisms of pairs are defined as follows:
\begin{enumerate}
\item{{\bf Isomorphisms of pairs over a fixed base:}}
As in Ref.\ \cite{BunkeS1} given two pairs $(E,H)$ and $(E',H')$ over
$X$ we say $(E',H')$ is isomorphic to $(E, H)$ if there is an isomorphism of principal circle bundles $h:E' \to E$ over $X$ of the form 
\begin{equation}
\begin{CD}
E' @>h>> E \\
@VV{p'}V   @VV{p}V \\
X @>>id> X \label{PairIso} \\
\end{CD}
\end{equation}
such that $h^{\ast}(H) = H'.$ 
\label{Mor1Pt1}

\item{{\bf Pullback of pairs along an arbitrary continuous map $f:Y \to X$ }}
In addition as in Ref.\ \cite{BunkeS1} we define a morphism
from the pair $(F \to X, H_F)$ to the pair $(E = f^{\ast}F \to Y, H_E)$
as a pullback square of principal circle bundles along the map
$f:Y \to X$
\begin{equation}
\begin{CD}
(f^{\ast}F, H_E = f^{\ast}(H_F)) @>\tilde{f}>> (F, H_F) \\
@VV{\phi}V   @VV{\pi}V \\
Y @>>f> X \label{PairPbck} \\
\end{CD}
\end{equation}
such that $f^{\ast}(H_F) = H_E.$
\label{Mor1Pt2}
\end{enumerate}
It is clear with these definitions of objects and arrows, 
$\C'$ forms a category.
\end{enumerate}
\label{DefCCat1}
\end{definition}

Note the following for Part (\ref{Mor1Pt1}) of 
Def.\ (\ref{DefCCat1}) above:
\begin{itemize} 
\item From Ref.\ \cite{BunkeS1} Lem.\ (2.1) isomorphic
pairs over the same base $X$ as in Part (\ref{Mor1Pt1}) of
Def.\ (\ref{DefCCat1}) above
are homotopy equivalent by a `homotopy of pairs'
defined in Ref.\ \cite{BunkeS1}.
\item Also from Ref.\ \cite{BunkeS1}
the isomorphism $h$ in Part (\ref{Mor1Pt1}) of 
Def.\ (\ref{DefCCat1}) above
need not be a pullback map it need
only be a circle bundle morphism. 
\item Note that the bundle morphism $h$ in Part (\ref{Mor1Pt1}) of
Def.\ (\ref{DefCCat1}) above is also a gauge transformation
of principal circle bundles underlying the above pairs since the
two bundles must be isomorphic by commutativity of
the Diagram Eq.\ (\ref{PairIso}).
\end{itemize}

Note that the definition of morphism of
pairs in Eq.\ (\ref{MorP}) is equivalent to pullback
morphisms between the homotopy equivalence classes of
pairs defined in Ref.\ \cite{BunkeS1} by the following
argument: Suppose we were given a morphism between
pairs covering a map $f: Y \to X$ as in Eq.\ (\ref{MorP}).
We construct the pullback square along the map $f: Y \to X$
associated to the spaces $E, X,Y$ 
with the pullback $f^{\ast}E$ replacing the circle bundle
$E'$ in Eq.\ (\ref{MorP}). Using the universal property of 
pullback bundles there is a unique natural
map $\nu:E' \to f^{\ast}(E)$ which causes $\tilde{f}$ to factor as
$\tilde{f} = \nu \circ f^{\ast}.$ This shows that Eq.\ (\ref{MorP})
factors uniquely into the composition of two commutative diagrams, 
one of which is a unique pullback square of the bundle 
$E \to X$ over $f:Y \to X$ and the other a unique isomorphism of a pair
over $Y.$ Hence it is easy to show that any morphism of pairs in the sense
mentioned in Eq.\ (\ref{MorP}) above is exactly equivalent to a
compositions of a unique pullback square over the map $f:Y \to X$ 
and a unique isomorphism of pairs over a fixed base. 
The converse is trivial and hence for the
rest of this paper
we consider $\C$ and $\C'$ to be {\em equivalent categories}.  

Thus we may alternately define $\C$ as consisting of objects which
are pairs over topological spaces
with arrows which are morphisms of pairs in the sense of  
Def.\ (\ref{DefMorPair2}).  It is clear from the above
that we may use either Def.\ (\ref{DefCCat1}) or Def.\ (\ref{DefCCat2})
to define morphisms in $\C.$

Note that if we have an isomorphism of pairs covering an arbitrary
map of spaces $f:Y \to X,$ then, $f$ has to be a homeomorphism of spaces.
In addition, when we factor this isomorphism as a  composite of
a unique isomorphism of pairs over $Y,$ and a unique pullback map
of pairs covering the given map of spaces $f:Y \to X,$
(as in the discussion after Def.\ (\ref{DefCCat2}) above)
this pullback map has to be an isomorphism of pairs covering
the homeomorphism $f:Y \to X$ as well.

We now describe two useful functors connected to the category
$\C$ above which we will use throughout this paper:
There is a natural forgetful functor $F:\C \to {\mathbf{Top}}$ which
sends a pair $(E \to X, H)$ to the base space $X$ of the associated 
principal bundle and sends any morphism of pairs 
from $(E \to X, H)$ to $(E' \to Y, H')$ namely
\begin{equation}
\begin{CD}
E' @>\tilde{f}>> E \\
@VV{p'}V   @VV{p}V \\
Y @>>f> X \nonumber \\
\end{CD}
\end{equation}
(with $\tilde{f}^{\ast}(H) = H'$ )
to the map $f:Y \to X.$ We will examine this functor in detail
in SubSec.\ (\ref{SSecPCLift}) below.

The functor $P$ of Bunke-Schick (see Sec.\ (2.1.8) of Ref.\ \cite{BunkeS1}) 
associates to any topological space $X$ the set of Bunke-Schick 
isomorphism classes of pairs (see Def.\ (\ref{DefCCat1}))
over $X$ and to any map of topological spaces $f:Y \to X$ a map of pairs 
$P(f): P(X) \to P(Y)$. The functor is contravariant since pairs
over $X$ pullback to pairs over $Y.$ In SubSec.\ (\ref{SSecPCLift})
we will discuss the lift of $P$ to the category $\C$ above.

\subsection{Some Properties of Morphisms of Pairs}
\label{SSecMorPf1}
Let $f:Y \to X$ a continuous map and let $\tilde{f}$ be a morphism
of pairs from pairs over $Y$ to pairs over $X$ covering $f$
as in Eq.\ (\ref{MorP}) above.
In this subsection we prove some results about the behavior of 
pairs over $X$ under pullback by $f.$ 

First, we show using the results in Ref.\ \cite{BunkeS1} that 
pullbacks of isomorphic pairs over $X$ along a map 
$f:Y \to X$ are isomorphic (note that the map $f$ can be 
nontrivial as a map from $Y$ to $X$):
\begin{lemma}
Suppose $f:Y \to X$ is a map and let 
$(E_0, H_0)$ and $(E_1, H_1)$ be isomorphic
pairs over $X.$ Then the pullback pairs 
$(f^{\ast}E_0, f^{\ast}H_0)$ and 
$(f^{\ast}E_1, f^{\ast}H_1)$ are isomorphic
over $Y.$ \label{LemPairIso}
\end{lemma}
\begin{proof}
By Ref.\ \cite{BunkeS1}, the pairs $(E_0,H_0)$ and $(E_1, H_1)$
are homotopy equivalent over $X$ by a homotopy of pairs
defined in Ref.\ \cite{BunkeS1} which gives a canonical pair
$(E \to (X \times I), H).$ The pair $(E \to (X\times I), H)$
is pulled back from the universal pair over $R$ by a
classifying map $\phi: X \times I  \to R$ where $R$ is the 
Bunke-Schick classifying space in Ref.\ \cite{BunkeS1}.
The composite map $\phi \circ (f \times id)$ 
classifies a pair $(f^{\ast} \times 1)E$
over $Y \times I$ and the restriction of this pair
to $Y \times {i}, i = 0,1$ are isomorphic to the pullbacks
of $(E_i, H_i)$ with $i = 0$ and $1$ respectively.
Thus by Ref.\ \cite{BunkeS1} the two pullback
pairs are homotopy equivalent over $Y$ as well and hence
isomorphic over $Y.$
\end{proof}
Thus if $\phi:(E_0, H_0) \to (E_1, H_1)$ is an isomorphism of pairs
over $X$ then the pullback pairs $(f^\ast E_0, f^{\ast} H_0 )$
and $(f^\ast E_1, f^\ast H_1)$ are isomorphic over $Y$ by
an isomorphism of pairs obtained from $\phi$ and $f.$
Thus given $f:Y \to X,$
pullback extends to a nontrivial map between isomorphism classes 
of pairs over $X$to isomorphism classes of pairs over $Y.$ 
This map is the  image of $f$ under the functor $P$ 
and is the map of sets $P(f)$ defined in Ref.\ \cite{BunkeS1}.

Next we show that pulling back a pair over $X$ via two homotopic
maps $f:Y \to X$ and $g:Y \to X$ gives isomorphic pullback
pairs over $Y:$
\begin{lemma}
Suppose $f,g:Y \to X$ are maps and $f \sim g$ via a homotopy 
$H: I \times Y \to X.$ Then, the pullback of a pair $(E,[a])$ over $X$ with 
$[a] \in H^3(E,\KZ)$  by $f$ is isomorphic to the
pullback of the same pair over $X$ by $g.$ In addition
the pullback pairs are homotopic by a homotopy of
pairs in the sense of Ref.\ \cite{BunkeS1}.
\label{LemHotIso}
\end{lemma}
\begin{proof}
Let $w: X \to R$ be the classifying map associated to
$(E,[\alpha]).$ Since $f \sim g$ it is clear the pullbacks
of $(E, [\alpha])$ along $f$ and along $g$ are isomorphic since
$w \circ f \sim w \circ g$
 
By the proof of Lem.\ (2.2) of Ref.\ \cite{BunkeS1},
isomorphism of pairs over the same base is equivalent to 
homotopy of pairs and the homotopy may be constructed naturally from the isomorphism.  
\end{proof}
Thus deforming a map $f:Y \to X$ in its homotopy equivalence
class doesn't change the associated map of sets $P(f): P(X) \to P(Y).$

It follows from Lemma (\ref{LemHotIso})
that if $X$ is homotopy equivalent to
$Y$ by maps $f:Y \to X, g: X \to Y$ with $f\circ g \sim \id_X,$
$g \circ f \sim \id_Y$
then the homotopy equivalence induces a bijection of sets
$P(f): P(X) \to P(Y).$ 

We also note the following fact which we will use below: Isomorphism classes of the pullback diagram
\begin{equation}
\begin{CD}
(E, \tilde{f}^{\ast}(H)) @>\tilde{h}_{\alpha}>> (E, H)\\
@VV{\pi}V   @VV{\pi}V \\
X @>>h_{\alpha}> X \nonumber \\
\end{CD}
\end{equation}

are in one-to-one correspondence with isomorphism classes of pairs 
over the mapping torus of $X$ by $h_{\alpha}.$
To prove this we show that every pullback of a pair over $X$ by a 
map $f:X \to X$  naturally gives rise to a pair on
the mapping torus $T_f X.$
\begin{lemma}
Let $f:X \to X$ be an arbitrary map, then every morphism of pairs
on $X$ of the form

\begin{equation}
\begin{CD}
(E, \tilde{f}^{\ast}(H)) @>\tilde{f}>> (E, H)\\
@VV{f^{\ast}\pi}V   @VV{\pi}V \\
X @>>f> X \nonumber \\
\end{CD}
\end{equation}

gives rise to a unique isomorphism class of a pair on the mapping torus $T_f X$ of $X$ by $f.$ The
isomorphism class of this pair does not change if we replace $f$ by a homotopic map.
\label{LemHomTf}
\end{lemma}
\begin{proof}
By Lem.\ (\ref{LemHotIso}), the isomorphism class of the pullback of
a pair does not change if $f$ is replaced by a map $g$ with $g \sim \id.$ 

Also, given a pullback diagram of pairs as above, we may form the mapping
torus $T_f X.$ It is well known that changing $f$ by a homotopy 
only changes $T_f X$ by a homeomorphism.

Thus, the nontrivial isomorphisms classes of the pullback diagram above 
are in one-to-one
correspondence with pairs over mapping tori each of the form $T_f X.$
\end{proof}
\subsection{Morphisms in $\C$}
\label{SSecMorC}
Let $f$ be an arbitrary map $f:Y \to X$ in ${\mathbf{Top}}$ 
and consider a morphism covering $f$ in $\C$ 
from a pair over $X$ to a pair
over $Y$ as in Def.\ (\ref{DefCCat2}). 
The Bunke-Schick functor $P:\mathbf{Top}^{op} \to \mathbf{Sets}$
of Ref.\ \cite{BunkeS1} associates a map of sets $P(f)$ to $f.$

In this section we show how to calculate $P(f)$ for every possible
morphism of pairs covering a morphism of spaces 
$f: X \to X$ as described in Def.\ (\ref{DefCCat2}) 
and as a result $P(f)$ is determined for every map 
$f:X \to X$ by composition (recall Def.\ (\ref{DefCCat1}
and Def.\ (\ref{DefCCat2}) are {\em equivalent} definitions of the
category $\C$).

In Sec.\ (\ref{SSecMorPf1}) we had proved some properties of the
map of sets $P(f)$ (see Eq.\ (\ref{MorP})) when $f:Y \to X$ is
any continuous map. In particular we had shown that the pullback of an 
isomorphism of pairs over $X$ gives rise to an isomorphism of pairs
over $Y.$ Also, pulling back the same pair over $X$ along two homotopic
maps $f,g:Y \to X$ results in an isomorphism of pullback pairs
over $Y.$ Using these results and the results 
in this subsection it is possible to determine the map
$P(f)$ fully for every possible map $f:X \to X.$

We will examine, in sequence, isomorphisms of
pairs when  $X \simeq Y$ and
$f$ is the identity map (as in the first part of Def.\ (\ref{DefCCat1})
and when $f:X \to X$ is a homeomorphism.
After that we will examine arbitrary morphisms (as in Def.\ (\ref{DefCCat2}))
of pairs when $f: Y \to X$ is an arbitrary map which is not a
homeomorphism.

In (\ref{SSecMorPXXid}) we determine all isomorphisms of pairs in $\C$
between two pairs both over the same base space $X$ which
cover the identity map $\id:X \to X.$ We also calculate the 
image of the $H-$flux $\tilde{f}^\ast(H)$ in equation Eq.\ (\ref{MorP}) 
under an isomorphism of pairs of $X.$ We will use the
results in this subsubsection in SubSec.\ (\ref{SSecBXFX}) below
to define the category $\B_X.$

In (\ref{SSecMorPXXf}) we determine all isomorphisms of pairs over the same base space covering an arbitrary homeomorphism $f:X \to X$
by determining the pullback pair by $f$ when $f:X \to X$
is a homeomorphism not homotopy equivalent to the identity. 
We will use the results in this subsubsection
in SubSec.\ (\ref{SSecBXFX}) below to define the category $\F_X.$

In (\ref{SSecMorPXY}) we determine
arbitrary morphisms of pairs covering an arbitrary map
$f:Y \to X$ which is not the identity or a homeomorphism
by describing a method to
calculate (see Def.\ (\ref{DefMorPair2}) and Eq.\ (\ref{MorP}))
the map of sets $P(f)$ when $f:Y \to X$ is any continuous map which
is not a homeomorphism. We will use the arguments in this
subsubsection to motivate the definition of the category $\G$
in SubSec.\ (\ref{SSecGXCX}) below.

\subsubsection{Isomorphisms of Pairs covering the identity}
\label{SSecMorPXXid} 
We now determine all isomorphisms between two pairs over the
same base $X.$ We show in SubSec.\ (\ref{SSecBXFX}) below that
we may form a category $\B_X$ from this data. The results
in this subsubsection completely characterize the category $\B_X$
and will be used in SubSec.\ (\ref{SSecBXFX}).

Consider an isomorphism between two pairs
$(E' \to X, H')$ to $(E \to Y, H)$ over the same base $X$ covering 
a map $f:X \to X.$ By Eq. (\ref{MorP}) we have a commutative diagram
\begin{equation}
\begin{CD}
E' @>\tilde{f}>> E \\
@VV{p'}V   @VV{p}V \\
X @>>f> X\\
\end{CD}
\end{equation}
with $\tilde{f}^{\ast}(H) = H'$ and
where the map $\tilde{f}$ need only be a $S^1-$equivariant map.
Since this diagram is an isomorphism in $\C$ the map $f$ must be a homeomorphism in ${\mathbf{Top}}$ and the map $\tilde{f}$ must be a bundle isomorphism. 

By the argument after Def.\ (\ref{DefMorPair2}) 
we may factorize the above morphism of a pairs into a pullback
of pairs over a homeomorphism $f:X \to X$ composed with a 
isomorphism of pairs (in the sense of Eq. (\ref{PairIso})) over $X.$

By the definition of morphism of pairs we just used and
the equivalent definition of morphism of pairs in Def.\ (\ref{DefMorPair2})
it is clear that we only need to consider two types of mappings of pairs --- 
isomorphisms of pairs over $X$ covering the identity map $\id_X:X \to X$
in the sense of Eq.\ (\ref{PairIso}) and pullback of pairs over $X$ by a nontrivial homeomorphism $f:X \to X$ in the sense of Eq.\ (\ref{PairPbck}). 
Any isomorphism of pairs over $X$ can be obtained by composing these two.

We begin by characterizing isomorphisms of pairs over
a base $X$ covering the identity map $\id_X:X \to X:$
Consider an isomorphism between a pair $(E', H')$ and $(E, H)$ over $X$ 
given by a circle bundle map $h:E' \to E$ as in Eq.\ (\ref{PairIso}) above.


In the above it is clear that $h$ is an isomorphism we have that
\begin{equation}
\begin{CD}
(E', h^{\ast}(H)) @>h>> (E,H) \\
@VV{\pi}V @VV{\pi}V \\
X @>>\id> X \label{PairIso2}\\
\end{CD}
\end{equation}
and $E' \simeq E.$

Thus every morphism of pairs $h$ from $(E, h^{\ast}(H))$ to $(E,H)$ over $X$ must be a pair of maps
$(\phi, \lambda)$ consisting of a bundle morphism $\phi: E \to E$ and a group homomorphism 
$\lambda: H^3(E, \KZ) \to H^3(E,\KZ).$

The bundle morphism $\phi$ corresponds to the group of 
circle-valued functions on $X$ acting by multiplication on $E.$  
Two such automorphisms are homotopic if the corresponding functions are homotopic. 
Thus the set of all bundle morphisms  $\phi:E \to E$ covering the identity morphism
$id:X \to X$ is a group isomorphic to a quotient of
$H^1(X, \KZ) \simeq [X, S^1] \simeq [X, K(\KZ,1)].$
Clearly, if the bundle morphism is multiplication by a constant $a, |a| = 1,$ 
it is homotopic to multiplication by the identity.

The group homomorphisms $\lambda: H^3(E,\KZ) \to H^3(E, \KZ)$ can be calculated
as follows following an argument in Ref.\ \cite{BunkeS1}, proof of Thm.\ (2.1.6).
\begin{lemma}
Suppose a nontrivial bundle morphism $\tilde{f}$ of $E \to E$ 
was induced by a multiplication map on $E$ by a function $a: X \to S^1$ 
such that the following diagram was commutative
\begin{equation}
\begin{CD}
(E, \tilde{f}^{\ast}(H)) @>\tilde{f}>> (E, H)\\
@VV{\pi}V   @VV{\pi}V \\
X @>>\id> X \nonumber \\
\end{CD}
\end{equation}
then
\begin{enumerate}
\item The map $\tilde{f}^{\ast}$ is given by
\begin{eqnarray}
\tilde{f}^{\ast}(H) &=& H - \pi^{\ast}(c \cup a) \nonumber \\
 &=& H + \pi^{\ast}(\pi_!(H) \cup a) \label{EqOrbitH1} 
\end{eqnarray}
where $a$ is the characteristic class in $H^1(X,\KZ)$ of the map $a:X \to S^1$ such that
the action of $\tilde{f}$ on $E$ is multiplication by $a$ and $\pi_!(H) = -c.$

\item The orbit of $(E,H)$ under a nontrivial bundle morphism $\tilde{f}$
depends on the T-dual pair $(E^{\#}, H^{\#}).$ 
\end{enumerate}
\label{LemPairHom}
\end{lemma}
\begin{proof}
\leavevmode
\begin{enumerate}
\item As in proof of Thm.\ (2.1.6) of Ref.\ \cite{BunkeS1}, we may let $a$ act on
$E$ by multiplication so we have a factorization of $\tilde{f}$ as:
\begin{gather}
E \stackrel{\pi \times id}{\to} B \times E \stackrel{a \times id}{\to} S^1 \times E 
\stackrel{m}{\to} E \label{DiamES1}
\end{gather}

where $m: S^1 \times E \to E$ is the multiplication map.

The result now follows from that Theorem.
\item Note that in the above expression,
$\pi_!(H)$ is the characteristic class of the T-dual bundle, say $[\pi^{\#}].$
Thus, Eq.\ (\ref{EqOrbitH1}) may also be written as:
\begin{equation}
\tilde{f}^{\ast}(H) = H + \pi^{\ast}( [\pi^{\#}] \cup a) \label{EqOrbitH2}
\end{equation}
\end{enumerate}
\end{proof}

In the above we have characterized isomorphisms of pairs over$X$ covering
the identity map $\id:X \to X.$ 
\subsubsection{Isomorphisms of pairs covering a nontrivial homeomorphism}
\label{SSecMorPXXf}
Consider a nontrivial isomorphism of pairs in the sense of 
Def.\ (\ref{DefCCat2}) between two pairs over $X$ covering
a nontrivial map $f:X \to X.$ Since this is an isomorphism of pairs,
it induces isomorphisms of the underlying bundles
and a self-homeomorphism $f:X \to X.$ 
By the arguments after 
Def.\ (\ref{DefCCat2}) we may factorize this morphism as an
isomorphism of pairs over $X$ composed with a pullback.
We have classified all isomorphisms of pairs over $X$ in
SubSubsec.\ (\ref{SSecMorPXXid}) above. 
Thus, to characterize such a morphism of pairs over $X$ it
 is enough to determine the pullback of a pair over $X$ by a 
map $f:X \to X$ which is a homeomorphism which is not the identity.
We will determine this in this subsubsection.

We show in SubSec.\ (\ref{SSecBXFX}) below that
we may form a category $\F_X$ whose objects are pairs over $X$ 
and whose morphisms are isomorphisms of pairs over $X$
covering a nontrivial homeomorphism $f:X \to X.$ The results
in this subsubsection and the previous one
will be used to define the category $\F_X$
in SubSec.\ (\ref{SSecBXFX}).

Note that pulling back a pair over $X$ via two homotopic
maps $f:X \to X$ and $g:X \to X$ gives isomorphic pullback
pairs over $X,$ see Lemma (\ref{LemHotIso}) below.
Thus given a pair over $X$ by Lem.\ (\ref{LemHotIso}) we can pick representative elements $h_{\alpha}$ 
from the group $Homeo(X)$ of self-homeomorphisms of $X$ 
and replace $f$ by a representative $h_{\alpha}$ homotopy equivalent to $f$
without changing the isomorphism class of the pullback pair. 
We pick the identity map $\id:X \to X$ as one of the representative maps and require the
remaining maps $h_{\alpha}$ to be nontrivial homeomorphisms from $X$ to itself
not homotopy equivalent to the identity or to each other.

Now $f \circ h_{\alpha}^{-1}$ is homotopic to the identity map and hence may be replaced by the identity map without disturbing the isomorphism
class of the pullback pair by Lem.\ (\ref{LemHotIso}) and 
Lem.\ (\ref{LemPairHom}).  
We write $f$ as $(f \circ h_{\alpha}^{-1}) \circ h_{\alpha}$
and the pullback of $(E,H)$ by $f$ can now be calculated 
by the previous lemma and the pullback square induced 
by $h_{\alpha}.$ Hence, 
$P(f) = P( (f \circ h^{-1}_{\alpha}) \circ h_{\alpha}) = P(h_{\alpha}).$
Thus it is enough to determine $P(h_{\alpha})$ for all $\alpha$
to determine $P(f)$ in the case of isomorphisms of pairs covering
a self homeomorphism $f:X \to X.$

Note that isomorphisms of pairs 
covering a nontrivial homeomorphism $f:X \to X$ can `connect' bundles
with different characteristic classes. That is if we are
given an isomorphism of pairs
(in the sense of Def.\ \ref{DefCCat2})
covering a nontrivial homeomorphism $h_{\alpha}:X \to X$ the
two principal bundles on either side of the morphism can be of
{\em different} characteristic classes in $H^2(X,\KZ).$

Consider the map induced on $H^2(X,\KZ)$ by
the map $h_{\alpha}.$ Since $h_{\alpha}$ is a nontrivial
homeomorphism $h_{\alpha}:X \to X,$ the induced map
must be an automorphism of the group $H^2(X,\KZ).$
If this map is a nontrivial automorphism
then the topological type of the bundle $h^{\ast}_{\alpha}(E)$
could be different from $E.$ In this case we would get a
morphism of pairs covering $h_{\alpha}:X \to X$
(as in Def.\ (\ref{DefCCat2})) in which the two
bundles on either side of the morphism possessed two 
different Chern classes.

For example if $X \simeq S^2$ and $h_{\alpha}$ was the antipodal
map we could consider the morphism of pairs over $X$ covering
$h_{\alpha}$ induced by pullback along $h_{\alpha}.$
Clearly the induced morphism on $H^2(S^2, \KZ) \simeq \KZ$ would
be $a \mapsto (-a)$ and a bundle with characteristic class $[a]$
would get mapped to a bundle with class $-[a]$ under pullback.

In the presence of torsion, for example for $X \simeq L(p,q)$ a lens space,
there could be many such morphisms given by Frobenius-like
automorphisms of $H^2(X,\KZ).$ We will use this fact 
further in the later sections of this paper.

\subsubsection{Arbitrary Morphisms of Pairs}
\label{SSecMorPXY}
Suppose we had an arbitrary morphism of pairs 
in the sense of Def.\ (\ref{DefCCat2}) covering an arbitrary
map of spaces $f:Y \to X$ which was not the identity and 
not a homeomorphism.  By the argument after Def.\ (\ref{DefCCat2}),
we may factorize the morphism as a composition of a pullback
map along $f:Y \to X$ together with an isomorphism of pairs over
$X.$ Since we have determined all isomorphisms of pairs over $X$
for every $X$ in SubSubSec.\ (\ref{SSecMorPXXid}) it is enough
to determine the pullback of a pair over $X$ along a map
$f: Y  \to X$ which is not the identity and not a homeomorphism.
We will use the arguments in this
subsubsection to motivate the definition of the category $\G$
in SubSec.\ (\ref{SSecGXCX}) below.

We now determine the pullback of a pair over $X$ along
a map $f:Y \to X$ where $f$ is an arbitrary continuous map
which is not the identity map and
not a homeomorphism. It is well-known that we may factor $f$ 
as a composite of a cofibration followed by a homotopy equivalence
(see Ref.\ \cite{May}). The pullback map $\tilde{f}$ in the pullback 
square induced by $f$ is then completely determined by the cofibration 
$h:X \to M_f$ as follows:

We may factor a map $f:Y \to X,$
possibly after homotopy, as a composite $f = g \circ h$ 
as $$Y \overset{h}{\to} M_f \overset{g}{\simeq} X.$$ 
Here $M_f$ is the mapping cylinder of $f,$ namely
$M_f \simeq X \cup_f (Y \times I),$
$h:Y \to M_f$ is a cofibration which is the inclusion
$h:Y \hookrightarrow M_f$ with $h(y) = (x,1)$
and $g:M_f \simeq X$ is a homotopy equivalence 
(see May Ref.\ \cite{May} Ch.\ (3) Sec.\ (4) --- our map $h$ is his map $j$
there).

We have the following commutative diagram of base spaces:
\begin{equation}
\begin{CD}
Y @>f>>X \\
@VVhV   @VV{id}V \\
M_f @>g>{\simeq}> X \label{fFactorMf}\\
\end{CD}
\end{equation}
Thus by the results in Sec.\ (\ref{SSecMorPf1}) 
isomorphism classes of pairs on $X$ pullback via the homotopy
equivalence $g$ to unique isomorphism classes of pairs on $M_f.$
It is clear that isomorphism classes of pairs on $M_f$ pull back via
restriction along the inclusion $h$ to isomorphism classes of pairs on $Y.$
Hence the pullback map $\tilde{f}$ is determined by the cofibration $h$
which happens to be the inclusion $h(x) = (x,1) \in M_f.$

In Eq.\ (\ref{fFactorMf}) the map $g$ is a homotopy equivalence 
and hence the natural functor between the category of pairs over 
$X$ and $M_f$ induced by pulling back pairs along $g$ (see Sec.\ (\ref{SecGXHot})) is an equivalence of categories $\C_X  \simeq \C_{M_f}.$ 
Hence a pair on $X$ induces a unique pair on $M_f$ and
every pair on $M_f$ arises in this way. 

When $f:Y \to X$ is not the identity and not a
homeomorphism the natural functor
$\tilde{h}:\C_{M_f} \to \C_Y$ induced by pulling pairs over $M_f$ along
$h$  (since $g:M_f \to X$ is a homotopy equivalence)
determines the behaviour of pairs on $X$ under pullback by $f$ 
(i.e. the pullback functor $f:\C_Y \to \C_X$ ) and 
also determines the $S^1-$equivariant map $\tilde{f}$ in the
pullback square induced by $f$ (see Def.\ (\ref{DefCCat1}) part (b))
as we now explain.

By the commutativity of the Diagram (Eq.\ (\ref{fFactorMf})) above
the pullback of a pair $(E, H)$ over $X$ by $f$ is isomorphic to the
pullback of the unique pair over $M_f$ induced by $(E,H)$ to a pair
over $Y$ by pulling back the pair along $Y$ along the map $h.$ 
Hence we can pullback pairs along $f$ by pulling back the corresponding
pairs on $M_f$ along $h$ and isomorphic pairs
over $X$ and $M_f$ pullback to isomorphic pairs by
Lem.\ (\ref{LemHotIso}) above. In addition we can lift Eq.\ (\ref{fFactorMf})
above to equivariant maps of $S^1-$spaces and $H-$flux  and we must have
$\tilde{f} \simeq \tilde{g} \circ \tilde{h}$ and the equivariant map
$\tilde{f}$ covering $f$ is determined by the equivariant map 
$\tilde{h}$ covering $h$ since the map 
$\tilde{g}$ covering $g$ is an isomorphism.

In the paper of Bunke and Schick (see Ref.\ \cite{BunkeS1}) 
the map $f$ induces a map on isomorphism classes of pairs
which is the map of sets $P(f)$ where $P$ is the
Bunke-Schick functor.  By
the argument in the previous paragraph for $f:Y \to X$
an arbitrary map which is not the identity or a nontrivial
homeomorphism $P(f)$ is determined by $P(h).$ We show
in Sec.\ (3) of this paper that
the map $h$ defines a natural simplicial map 
$G(h):G(M_f) \to G(X)$ and
the argument in the previous paragraph shows that 
$G(h)$ determines $G(f)$ and hence also determines $P(f).$

\section{The Space $G(X)$ and its Homotopy Properties \label{SecGXHot}}

In the Introduction to this Section, we argue that the 
objects of $\C$ are SFT vacua which are
circle-fibered over a `base' spacetime
and which possess a fixed point free circle isometry i.e
they have the structure of a principal circle bundle with
$H-$flux over that spacetime. We explain this analogy in
some detail in this Introduction and show how the 
analogy relates to various subsections of this section. 
We discuss this proposed connection between
SFT and Topological T-duality in detail in 
SubSec.\ (\ref{SSecSFT-TTD}).

In SubSec.\ (\ref{SSecPCLift}) we
construct the category $\HC$ whose objects are
isomorphism classes of pairs in $\C$ under Bunke-Schick
isomorphism and whose morphisms are induced from $\C$
by the Bunke-Schick construction. We had noted
in Sec.\ (\ref{SecIntro}) above that $P$ was naturally
a presheaf of sets on $\mathbf{Top}.$ Using the
{\em category of elements} of this presheaf
(see \ref{SSecPCLift}) we prove that
Bunke-Schick's functor $P$ induces a natural
fibration $E: \HC \to \mathbf{Top}$ whose fibers
over a topological space $X$ are the set $P(X).$
We identify this fibration with the category
$\HC$ above, and show that $F \simeq E \circ \pi.$
We show that $E$ is a Grothendieck fibration and
hence also a Street fibration. We also show that we
may lift the functor $P$ to $\C$ such that the 
lifted functor (also denoted $P$) when evaluated
on a pair will give us the set $P(X).$ 

Physically, it is clear that the elements of the set 
$P(X)$ must label String Field 
Theory backgrounds with a free circle 
isometry which are circle compactified over $X$ i.e., 
isomorphism classes of objects in the subcategory 
$\C_X$ of  $\C$ --- see discussion at the beginning of 
Sec.\ (\ref{SecIntro}).

To construct the moduli space $G(X),$ we 
could use the subcategory
$\C_X$ consisting of all pairs over $X$ and
all morphisms of pairs but it is not 
possible to handle this category easily. We
restrict ourselves to considering the
subcategory $\pi(\F_X)$ of $\C_X$ defined
below. (Note that if we pass to the associated
classifying space of these categories,
$B \pi(\F_X)$ is the $1-$skeleton of
$\C_X$ and it is possible to show that the 
zeroth and first homotopy groups of
$B \pi(\F_X)$ are isomorphic to the zeroth and
first homotopy groups of $\C_X$ by the natural
inclusion $B \pi(\F_X) \hookrightarrow B\C_X.$ Thus,
$\pi(\F_X)$ should be viewed as a reasonable approximation
of $\C_X$ --- see SubSubSec. \ref{SSecNonTrivP_k} below.)

In SubSec.\ (\ref{SSecBXFX}) below we discuss some
properties of the subcategory $\pi(\F_X)$ of the category 
of pairs $\C$ whose objects we 
identified above with SFT
backgrounds with a free circle isometry which were 
circle compactified over $X$ and with 
gauge transformations induced from isomorphisms
of pairs as arrows. We note that $\pi(\F_X)$ is
a small groupoid.
Then we naturally construct a functor 
$\G: \mathbf{Top}^{op} \to \mathtt{Cat}$ which assigns
to any topological space $X$ the geometric
realization of the simplicial space
of the category $\pi(\F_X).$  We identify
$G(X)$ as the gauge moduli space of
SFT backgrounds which are principal
circle bundles over $X$ as in the
discussion at the end of SubSec.\ (\ref{SSecSFT-TTD}).
This functor is constructed  using $\C$ and the 
constructions in SubSec.\ (\ref{SSecPCLift}). 

Physically, it is clear that the functor $\G(X)$ 
assigns to any topological space $X$
the gauge moduli space of the SFT vacua in
the category $\pi(\F_X)$
which has as objects 
all the SFT backgrounds
which are circle-fibered over $X$ 
(with the circle fiber an isometry
direction) and has as morphisms 
SFT gauge transformations ---both background-preserving and
background-changing--- induced by isomorphisms
of pairs between these backgrounds. 
To any map $f:Y \to X$ this assignment assigns the pullbacks
of all such backgrounds along the map $f$ and corresponds
to the pullback functor $\G(f)$ discussed in 
SubSec.\ (\ref{SSecDefCX}). 
Physically this pullback functor should be interpreted 
as a gauge transformation of a 
string theory background corresponding to a principal 
circle bundle with $H-$flux over
$X$ to another such string theory background (possibly
with a different characteristic class) over $Y.$

In SubSec.\ (\ref{SSecBXFX}) we naturally construct two
subcategories of $\C$ which we term
$\B_X, \F_X$ and a subcategory of $\HC$ namely
$\pi(\F_X)$ which we will use throughout the paper.

The subcategory $\B_X$ of $\C$ has as objects pairs 
over $X$ and as morphisms Bunke-Schick 
isomorphisms between pairs over $X.$
By the discussion in SubSubSec.\ (\ref{SSecMorPXXid})
above, we showed that these may be identified 
gauge transformations of the underlying principal
circle bundle of a given pair $(E\to X, H)$ over $X.$

Physically, these gauge transformations of the principal circle
bundle underlying a pair must then be identified with certain
background-preserving gauge transformations 
(see SubSec.\ (\ref{SSecSFT-TTD}) for details) of the 
SFT spacetime $E$ associated with
the pair $(E \to X, H)$ which induce the identity map on the 
`base' spacetime $X$ of the SFT vacua in $\B_X.$ 

The subcategory $\F_X$ has as objects all pairs over $X$ and
as arrows all isomorphisms of pairs which cover a homeomorphism $f:X \to X.$ 

Physically, the objects
in $\F_X$ are still SFT 
spacetimes $E$ associated to
pairs $(E \to X, H)$ over $X.$ The morphisms in $\F_X$
correspond to
both background-preserving gauge
transformations (corresponding to morphisms of pairs
covering $\id:X \to X$) and background-changing gauge
transformations (corresponding to morphisms of pairs 
$f:X \to X$ covering a 
nontrivial self-homeomorphism of $X$) of any pair in $\C_X.$
As discussed in Sec.\ (\ref{SecIntro}) above, these
backgrounds and gauge transformations are an
approximation to SFT when the 
SFT gauge transformations do not 
distort the background too much (see 
Refs. \ \cite{SFTDiffeo, KZw}). In this
paper we study the category $\F_X$ by using
the subcategory $\pi(\F_X)$ of $\HC.$

The subcategory $\pi(\F_X)$ of $\HC$ is
the essential fiber of the functor $E:\HC \to \mathbf{Top}$
which was defined in SubSec.\ (\ref{SSecPCLift}) above.
It has as objects homotopy
classes of pairs over $X$ with morphisms between
them. We show that the arrows in $\pi(\F_X)$ are
images of pullback arrows in $\C$ under $\pi:\C \to \HC.$
This category is very important for the rest of this
paper.

In SubSec.\ (\ref{SSecNMFE}) we define a natural
presheaf of groupoids $\M$ on $\mathbf{Top}$
from the functor $E: \HC \to \C$ defined in
Thm.\ (\ref{ThmHCTEqv}) above using
categorical equivalence for {\em Street} fibrations.
We show that
the value of the functor $\M$ on a topological
space $X,$ namely $\M(X),$ 
is the groupoid $\pi(\F_X)$ above. 
We argue in SubSec.\ (\ref{SecP_i}) below that 
$\M(X)$ is a `categorification' of Bunke-Schick's
functor $P(X).$

In SubSec.\ (\ref{SSecGXCX}) we study the functor
$\M(X)$ by studying the homotopy theory of the
geometric realization of the groupoid $\M(X)$ as
$X$ varies. We naturally construct a functor 
$G: \mathbf{Top}^{op} \to \mathbf{Simp}$
from the functor $\M$ defined in SubSec.\ (\ref{SSecNMFE})
which sends any topological space $X$ to
the simplicial realization (see Refs.\ \cite{Segal1, BRichter})
of the small groupoid $\M(X).$ 
We also show that the set of $0-$simplices of $G(X)$ is
$P(X).$ This supports our claim that $\M(X)$ is
a categorification of Bunke-Schick's funtor $P(X).$

In SubSec.\ (\ref{SSecHotG}) we discuss properties 
of the homotopy groups of $G(X)$ for any topological
space $X.$ We discuss the relation of these 
homotopy groups with $P(X).$ We show that we may 
define a generalization of $P(X)$ from these homotopy
groups.

In SubSec.\ (\ref{SSecSFT-TTD})
we relate the above
(especially SubSec.\ (\ref{SSecGXCX}) and SubSec.\ (\ref{SSecHotG}))
to closed SFT and
justify the construction of the space
$G(X)$ from $X$ based on arguments from 
SFT.

\subsection{The Lift of $P$ to $\C$ \label{SSecPCLift}}
Bunke-Schick's functor $P$ assigns to any topological space
$X$ the set of isomorphism classes of pairs  over $X$ and assigns
to a map of topological spaces $f:Y \to X$ a map of sets from the set of 
isomorphism classes of pairs over $X$ to the set of
isomorphism classes of pairs over $Y$ \cite{BunkeS1}. 
We would like to lift this functor to the category $\C.$

We had noted above that there was a natural forgetful functor
$F: \C \to \mathbf{Top}$ which sends a pair $(E \to X, H)$ to 
$X$ and sends morphisms of pairs (in the sense of Def.\ (\ref{DefCCat2}))
to maps between spaceds. Thus, it should be possible to lift
$P$ to $\C$ by forming the composite functor $P \circ F.$

However it is not clear that $F$ is the only such forgetful functor, and
hence that this lift is unique. For example, sending the target of 
the assignment $F:(E \to X, H) \mapsto X$ to a homotopy equivalent space (say $(E \to X, H) \mapsto X \times \KR$) would result in a seemingly
equivalent `forgetful' functor which, when composed with $P$ would result in the same functor $\C \to \mathbf{Set}$ (since $P(X)$ only depends on
the homotopy type of $X$ --- see Ref.\ \cite{BunkeS1} or Ref.\ \cite{Pan2}).

We now prove that the functor $F$ above may be expressed as 
a composite of two natural maps one from $\C$ to the 
category of homotopy equivalence classes of pairs 
$\HC$ (defined below) which we write as
$\pi:\C \to \HC$ and another a
natural functor $\rho:\HC \to \mathbf{Top}$ which we construct
below and as a result $F$ must be of the form $F:(E \to X, H) \mapsto X.$

This gives us a unique lift the Bunke-Schick
functor $P:\mathbf{Top}^{op} \to \mathbf{Set}$ to $\C$ namely 
$P \circ F: \C \to \mathbf{Set}.$
The lift $P \circ F$ acts on a pair $(E \to X, H)$ 
and sends it to the set of all isomorphism classes of pairs on $X.$
We will use this lift throughout this paper.

We now define a new category $\HC$ which is the quotient of
$\C$ by the equivalence relation of Bunke-Schick Equivalence.
\begin{theorem}
\label{ThmCSimCat}
\leavevmode
\begin{enumerate}
\item The class $\HC$ has homotopy equivalence classes of 
pairs as objects. Every morphism between two pairs in 
$\C$ naturally induces a morphism between the corresponding
objects in $\HC.$ This gives $\HC$ the structure of a category. 
There is a natural functor $\pi: \C \to \HC.$
\item Let $f: e \to e'$ be any morphism of pairs between two
pairs $e,e'$ in $\C.$ Then $\pi(f) = \pi(p)$ for some pullback
map $p: e \to e'.$
\item Assume we are given any morphism $\eta$ between two equivalence
classes of pairs $\pi(a), \pi(b)$ in $\HC$ with $\eta: \pi(a) \to \pi(b)$ in
$\HC$ where $a = (E_X \to X, H_X)$ and $b = (E_Y \to Y, H_Y)$ are
pairs in $\C.$ Then, we may change the pairs $a,b \in \C$
in their respective homotopy equivalence
classes in $\C$ to pairs $a', b' \in \C$ 
such that there is a morphism of pairs
$f:a' \to b'$ in $\C$ covering $\eta$ i.e., $\eta = \pi(f),$
and $f$ is always a pullback morphism in $\C.$ 
\end{enumerate}
\end{theorem}
\begin{proof}
\leavevmode
\begin{enumerate}
\item We only need to check that each morphism in $\C$
induces a morphism between objects in $\HC.$ 
This is clear since if we have a morphism between two pairs
in $\C$ in the sense of Def.\ (\ref{DefCCat2}) then,
precomposing and postcomposing the morphism by
isomorphisms of the corresponding pairs over $X$ will give
another morphism of pairs in $\C$ by Def.\ (\ref{DefCCat1})
(since the composition of an isomorphism of a pair with an arbitrary
morphism of a pair gives a valid morphism of a pair by
Def.\ (\ref{DefCCat1}) and the two definitions of a morphism of a pair
in $\C$ are equivalent). Thus each morphism
of pairs in $\C$ gives rise to a well-defined morphism of objects
in $\HC.$

The identity morphism in $\C$ gives the identity morphism
on homotopy equivalence classes of pairs in $\HC$ and clearly
$\HC$ is a category.

From the above it is clear that the assignment which
sends a pair to its homotopy equivalence class and a morphism
of pairs to the induced morphisms on the homotopy equivalence
classes of the domain and range pairs (see above) is a functor from
$\pi: \C \to \HC.$

\item Any map $f:e \to e'$ in $\C$ may be written
as a product of a unique Bunke-Schick isomorphism
$\iota:e \to e$ and a unique pullback map $p:e \to e'$
by the discussion after Thm.\ (\ref{DefCCat2}) above.
Hence, the image of $f$ by 
the functor $\pi$ is just $\pi(p),$ i.e. $\pi(f) = \pi(p).$

\item Assume that we are given two objects $\pi(a)$ and $\pi(b)$ in $\HC$
together with a map $\eta: \pi(a) \to \pi(b)$ in $\HC.$
Let $a = (E_X \to X, H_X)$ and $b = (E_Y \to Y, H_Y)$ be 
pairs in $\C$  covering $\pi(a), \pi(b)$ respectively.
Then we may lift $\eta$ to a map of pairs $h:a \to b$ in $\C$ since,
by construction, the functor $\pi$ is full.

Now consider the map of pairs $h:a \to b.$ We suppose this map of pairs
covered a map $f:Y \to X.$ By the discussion after Def.\ (\ref{DefCCat2}) 
$h$ may be written as a composition of a unique 
pullback $\tilde{f}$ along $f$ together with
a unique isomorphism of pairs $\iota$ over $Y$.
The image of $h$ under the map $\pi$ is clearly $\pi(h) = \eta$ 
and $\pi(h) = \pi(\iota) \circ \pi(\tilde{f})$ 
but due to the definition of $\pi,$ and $\HC,$ we have
$\pi(\iota) = \id$ and hence $\pi(h) = \pi(\tilde{f}).$

Thus shifting the given pair over $X$ and using 
Lem.\ (\ref{LemPairIso}) shows that the lift of $\eta$ to $\C$ may
always to be a pullback map up to isomorphism of 
pairs over a fixed base --- or equivalently up to homotopy of pairs in the sense of Ref.\ \cite{BunkeS1}. 
\end{enumerate}
\end{proof}

Recall that there is the natural functor 
$F: \C \to \mathbf{Top}$ which
sends a pair $(E \to X, H)$ to $X.$ Also, the
natural quotient map 
$\pi:\C \to \HC$ in Thm.\ (\ref{ThmCSimCat})
sends a pair $(E \to X, H)$ to $[(E \to X),H]$ where
$[(E \to X, H)]$ denotes the Bunke-Schick isomorphism
class over $X$ 
(see Def.\ (\ref{DefCCat1} for a def of Bunke-Schick isomorphism)
of the pair $(E \to X, H).$

Given a pair $(E \to X, H)$ over $X,$ we claim the functor $F$ naturally
induces a functor $E: \HC \to \mathbf{Top}$
which sends the set of pairs over $X$
related to $(E \to X, H)$ by an isomorphism
of pairs over $X$ (see Ref.\ \cite{BunkeS1} and Def.\ (\ref{DefCCat1}))
to the underlying topological space $X.$ The functor
$E$ lifts any map between two isomorphism
classes of pairs say $\phi: [(E_X \to X, H_X)] \to [(E_Y \to Y, H_Y)]$ in
$\HC$ to a pullback map $\tilde{\phi}$ in $\C$ between $(E_X \to X, H_X)$
and $(E_Y \to Y, H_Y)$ covering $f:Y \to X$
as in Thm.\ (\ref{ThmCSimCat}) above and then defines
$E(\phi): \phi \mapsto (f:Y \to X).$
Thus, if the functor $E$ exists and is unique, we would
have that $F \simeq E \circ \pi.$

We now show that the functor
$E$ exists and is a unique.  We further identify
$E$ with the functor $\rho$ (see Sec.\ (5.1.1) of Ref.\ 
 \cite{BRichter}) which we construct naturally below
from $P$, $\mathbf{Top}$ and the values of $P.$
We will need the following definition from 
Sec.\ (5.1.1) of Ref.\  \cite{BRichter}:
Given a contravariant functor 
$F: \C^{op} \to \mathbf{Set},$ we can 
build a new category $F \backslash \C$ out of $F,\C$
and the values of $F$ by viewing $F$ as a presheaf
of sets on $\C^{op}$ and then using the Grothendieck
construction to construct a category with 
$\mathbf{Set}-$valued fibers
over $\C$ (see nCatLab page on Category of Elements
Ref.\ \cite{nLab-CatEl} and Sec.\ (2) of nCatLab page on
Grothendieck Construction Ref.\ \cite{nLab-GCons} 
for the Grothendieck Construction for a 
contravariant functor).
\begin{definition}
For a given contravariant functor 
$F:\C^{op} \to \mathbf{Set},$ 
let $F \backslash \C$ be the category whose objects
are pairs $(C,x)$ with $C$ an object of $\C$ and 
$x \in F(C).$ A morphism
from $(C,x)$ to $(C',x')$ is an $f \in \C(C,C')$ 
with $F(f)(x') = x.$
\end{definition}

The category $F \backslash \C$ is termed 
the category of elements
of $\C$ (also termed $\mbox{El}_F(\C)$ or $\int F$). It is a
functorial pullback (see nCatLab page on $2-$pullback,
Ref.\ \cite{nLab-2Pbck}, the nCatLab page on Category of
Elements, Sec.\ (2) of Ref.\ \cite{nLab-CatEl} and
the nCatLab page on Grothendieck Construction 
Ref.\ \cite{nLab-GCons} esp. Sec.\ (2) of that page 
on Grothendieck Construction for a contravariant functor) 
and fits into a pullback square 
\begin{equation}
\begin{CD}
\int F \simeq \mbox{El}_F(\C) @>>> \mathbf{Set_{\ast}}^{op} \\
@VV{\pi_F}V   @VV{U^{op}}V \\
\C @>>F^{op}> \mathbf{Set}^{op}. \label{FSlashCDef} \\
\end{CD}
\end{equation}
Here $F^{op}$ is the opposite of the functor
$F:\C \to \mathbf{Set}^{op}$ and 
$U: \mathbf{Set_{\ast}} \to \mathbf{Set}$
is the forgetful functor from the category of
pointed sets to the category of sets and is the
universal bundle with $\mathbf{Set}-$fibers over 
$\mathbf{Set}.$ 
By Ref.\ \cite{nLab-CatEl}, Sec.\ (1)
the category $\mbox{El}_F(\C) \simeq F\backslash \C$ 
is a category fibered over $\C$ by the natural 
functor $\pi_F:F \backslash \C \to \C$ whose
fiber over an object $c \in \C$ is the set $F(c).$ 

Note that the functor defined by Bunke-Schick in 
Ref.\ \cite{BunkeS1}
is the contravariant functor 
$P: \mathbf{Top}^{op} \to \mathbf{Set}.$ 
Thus we may construct the category 
$P \backslash \mathbf{Top}$ consisting of pairs
$(X, s)$ with $X$ a topological space and $s \in P(X).$

A morphism from $(Y, s')$ to $(X, s)$ consists of a map of
topological spaces $f:Y \to X$ which induces a map 
$P(f)(s) = s'$ (since $P$ is contravariant). 
Thus $P(f)$ must consist of a 
morphism of pairs from a pair in the isomorphism class
$s \in P(X)$ to another pair in the isomorphism class 
$s' \in P(Y).$
By Thm.\ (\ref{ThmCSimCat}) above we may assume 
that this morphism
is a pullback map from a pair in the isomorphism class $s$ to another
pair in the isomorphism class $s'.$ 

There is a forgetful functor (see Ref.\ \cite{BRichter}) 
$\rho: P \backslash \mathbf{Top} \to \mathbf{Top}$ that
sends the ordered pair $(X,s)$ --- with $s \in P(X)$ an isomorphism
class of pairs over $X$ to $X$ --- and sends a morphism from
$(X,s)$ to $(Y,s')$ to the map $f:Y \to X.$ 
This will fit into the following pullback square
\begin{equation}
\begin{CD}
\int P \simeq P \backslash \mathbf{Top} \simeq \mbox{El}_P(\mathbf{Top})  @>>> \mathbf{Set_{\ast}}^{op} \\
@VV{\rho}V   @VV{U}V \\
\mathbf{Top} @>>P^{op}> \mathbf{Set}^{op}. \label{PSlashTopDef} \\
\end{CD}
\end{equation}
Here, the functor $P^{op}:\mathbf{Top}^{op} \to \mathbf{Set}^{op}$ 
is the opposite functor defined from the functor 
$P:\mathbf{Top}^{op} \to \mathbf{Set}.$

Also, the map in the upper row of the commuative diagram
Eq.\ (\ref{PSlashTopDef}) is the assignment
$\pi(e) \mapsto (P(F(e)), \pi(e))$ for any pair $e$ in $\C.$

However, $(X, s)$ with $s \in P(X)$ is the set of all 
homotopy equivalence classes of pairs over $X$ which
are homotopy equivalent in the sense of 
Ref.\ \cite{BunkeS1} to the given pair $s \in P(X)$ and
hence the category $P \backslash \mathbf{Top}$ is
actually the category of homotopy equivalence classes of 
pairs $\HC$ defined above. 

We may also identify $\rho$ with the functor
$E:\HC \to \mathbf{Top}$ defined above in this subsection
since $E$ and $\rho$ agree on objects and on morphism
$\rho$ sends a morphism $\phi$ between $(X,s)$ to $(Y,s')$ to
the map $f:Y \to X$ which agrees with the action of $E(\phi)$
described just after Thm.\ (\ref{ThmCSimCat}) above.

Hence we have proved the following theorem:
\begin{theorem}
Let $\HC$ be the category of homotopy equivalence classes of pairs
defined above. Let $P: \mathbf{Top}^{op} \to \mathbf{Sets}$ be the
Bunke-Schick functor. Let $P \backslash \mathbf{Top}$ be the
category defined before this theorem.
\leavevmode
\begin{enumerate}
\item We have an equivalence of categories between
the categories $\HC$ and $P \backslash \mathbf{Top}.$
\item The functor 
$\rho: P \backslash \mathbf{Top} \to \mathbf{Top}$
defined above is exactly the functor $E: \HC \to \mathbf{Top}$
above.
\end{enumerate}
\label{ThmHCTEqv}
\end{theorem}
{\flushleft{\bf{Note:}}} 
Thm.\ (\ref{ThmHCTEqv}) also
shows that the fibration
$E:\HC \to \mathbf{Top}$ should be viewed as the 
`geometrization' of Bunke-Schick's functor $P$
by the category of elements construction.

As described above, the category
$P\backslash \mathbf{Top}$
is obtained from the presheaf of sets defined by $P$ 
on $\mathbf{Top}$ by the Grothendieck Construction
(see Grothendieck Construction for contravariant
functors Sec.\ (2) of Ref.\ \cite{nLab-GCons}) 
and is a category fibered in sets over $\mathbf{Top}$
by the functor $\rho$ in the previous paragraph.

We show in Thm.\ (\ref{ThmPiFXDef}) below that the
fiber of $E$ over any topological space $X$
is a subcategory of $\HC$ consisting of 
a set of objects with the identity automorphism
for each object as expected. We also show that this
set of objects is in bijection with the set $P(X).$
We also prove in Thm.\ (\ref{ThmMGFMM}) below that 
$E$ is a Grothendieck Fibration with this subcategory
as a fiber, as expected. 

We show in Thm.\ (\ref{ThmPiFXDef}) that
the essential fiber of $E$ over a topological space
$X$ is a small groupoid 
(namely the groupoid  $\pi(\F_X)$ defined in  
Lem.\ (\ref{LemPiFXGrpd})) whose set of objects
is the subcategory in the previous paragraph, hence $\HC$ 
may also be viewed as
a category fibered in (small) groupoids 
over $\mathbf{Top}$ (for a definition of a category 
fibered in groupoids, see Ref.\ \cite{nLab-CatFibG}).
When viewed like this, $E:\HC \to \mathbf{Top}$ is
a {\em Street Fibration} over $\mathbf{Top}$
(see Ref.\ \cite{nLab-StrFib}) and
we will use this property in SubSec.\ (\ref{SSecNMFE})
below. 

Since $\HC$ is the category of elements of
$P:\mathbf{Top}^{op} \to \mathbf{Sets}$ 
this suggests that objects in the category
$P \backslash \mathbf{Top}$ or equivalently in
the category $\HC$ should be viewed as generalizations
of topological spaces constructed using 
Bunke-Schick isomorphism of pairs. 

We suggest that the category 
$\HC$ extends the idea of a topological space by
representing a given topological space $X$ by
the category of homotopy equivalence classes of pairs and 
morphisms of pairs over $X$ i.e., the category
$\pi(\F_X)$ should
be viewed as a 'categorification' of $X.$ Since the fiber
of $E$ is $P(X),$ there are no more than $|P(X)|$
pairs over a given topological space $X.$ 

Since $X$ is usually the underlying
topological space of a Riemannian manifold, 
each pair is trivial when restricted to a sufficiently
small open set. Thus, the fibration 
$E: \HC \to \mathbf{Top}$ should be
viewed being constructed from the set of
covering spaces of topological spaces in $\mathbf{Top}.$

We have a unique natural functor 
$\C \to \mathbf{Top}$ given 
by composing the quotient map $\pi:\C \to \HC$ with the 
map $\rho$ above that is 
$$
\rho \circ \pi: \C \to \HC \simeq (P \backslash \mathbf{Top}) \to
\mathbf{Top}.
$$ 
This sends any pair $(E \to X, H)$ to its homotopy class
$[(E \to X, H)]$ and then sends this class to its underlying
topological space $X$ that is:
$$
\rho \circ \pi: (E \to X, H) \mapsto [(E \to X, H)] \mapsto X.
$$
It is clear that, by uniqueness,
$F \simeq  \rho \circ \pi \simeq E \circ \pi.$
When composed with the functor $P$ the map $E \circ \pi$
gives a unique functor $P\circ E \circ \pi:\C \to \mathbf{Set}$ 
which sends a pair $(E \to X, H)$ to the isomorphism class of pairs
over $X$ and hence must be the required unique lift of 
$P$ to $\C.$

\subsection{The subcategory $\C_X$ \label{SSecDefCX}}
Just before Def.\ (\ref{DefCCat1}) above we had defined the
subcategory $\C_X$ of $\C.$ The objects of $\C_X$ consists of
all pairs over $X$ and the arrows in $\C_X$ consist of
morphism of pairs over $X$ between two pairs over $X$ 
which cover some map (trivial or nontrivial) $f:X \to X.$

Note that the functor $F$ defined in 
SubSec.\ (\ref{SSecPCLift}) restricts to $\C_X$ and sends any pair
in $\C_X$ to $X$ and any morphism of pairs in $\C_X$ to a map from
$X$ to $X.$ Clearly this restriction of the functor $F$ 
is full and might not be faithful.

It is easy to show that $\C_X$ is a small category for every topological space
$X.$ This is because the isomorphism classes of pairs over $X$ form a 
set (by Ref.\ \cite{BunkeS1}) and by SubSec.\ (\ref{SSecMorC}) 
these isomorphism classes consist of homotopy classes of elements
of $Homeo(X)$ together with multiplication by the set of circle-valued 
maps on $X$ and both these are sets as well. 
In addition, by the discussion in SubSec.\ (\ref{SSecMorC}) above, 
the collection of mappings between pairs over $X$ 
(i.e. the $Hom-$set in $\C_X$) is also a set.

Let $\mathtt{Cat}$ be the category of all small categories with
functors between them as arrows (see Ref.\ \cite{BRichter}).
The following theorem shows that we may construct a 
functor $\G: \mathbf{Top}^{op} \to \mathtt{Cat}$ from $\C_X.$
This is presheaf of small categories on $\mathbf{Top}$ and we
will use it in the following sections.
\begin{theorem}
Let $\G$ be an assignment which sends 
a topological space $X$ to the small category $\C_X$
and which sends maps between topological spaces
$f: Y \to X$ to functors $G(f): \C_X \to \C_Y.$
The assignment $G(f)$ assigns to an object $(E \to X, H_X)$ its
pullback along $f$ as an element of $\C_Y.$ 
Then $\G$ is a functor $\G: \mathbf{Top}^{op} \to \mathtt{Cat}.$
\label{ThmGFunc}
\end{theorem}
\begin{proof}
First, for every map $f: Y \to X,$ the assignment $G(f)$ is
a functor as follows:
Suppose we had a morphism of pairs in $\C_X$ as follows
\begin{equation}
\begin{CD}
(E_2, H_2) @>\tilde{h}>>(E_1,H_1) \\
@VVV   @VVV \\
X @>h>> X. \nonumber \\
\end{CD}
\end{equation}
Then, as discussed in the paragraph after Def.\ (\ref{DefCCat2}) 
above, we may factor this morphism into a composition of 
a pullback morphism followed by an isomorphism of pairs over
$X.$ 

If we pick a circle valued function $w \in C(X,S^1),$ we may
view the isomorphism of pairs over $X$ as a gauge transformation
of the bundle $E_2$ over $X$ induced by multiplication 
by $w.$ 

Pulling this function back to $Y$ by composing with the morphism 
$f$ gives an isomorphism of pairs over $Y.$
In particular the identity morphism of pairs over $X$ pulls
back to the identity morphism of pairs over $Y.$

Pullback squares with base a map $h':X \to X$ may be pulled back
along $f$ as well. 

Composing both these results above we can see that the pullback of an
arbitrary morphism of pairs over $X$ along $f:Y \to X$
is a morphism of pairs over $Y.$

Also the pullback of two composable morphisms of pairs over $X$ 
(in the sense of Def.\ \ref{DefCCat2} above) are composable morphisms
of pairs over $Y$ and the entire diagram of pullbacks and morphisms
commutes. 

Hence $G(f)$ is a well-defined functor from $\C_X$ to $\C_Y.$

From the definition of pullback and the above two results we can show
that $G(f \circ g) \simeq G(g) \circ G(f).$
The remaining properties of a functor follow from this.
Hence we have a contravariant functor 
$\G: \mathbf{Top} \to \mathtt{Cat}.$
\end{proof}

We now argue that, in a very precise sense, $\C_X$ is determined
by the isomorphisms in it together with some extra data:
In SubSubSec.\ (\ref{SSecMorPXY}) we had shown 
that the pullback of pairs in $\C_X$ under an arbitrary map 
$f:Y \to X$ could be calculated using pullback along a
cofibration map $h: Y \to M_f$ where $M_f$ is the mapping
cone of the map $f$ (see Ref.\ \cite{May}).

We showed in SubSubSec.\ (\ref{SSecMorPXY})
that the behaviour of an arbitrary pair over 
$X$ up to isomorphism of pairs under an arbitrary map $f:Y \to X$ 
(up to homotopy of $f$) was determined by pulling back a pair on
$M_f$ induced by the given
pair on $X$ along a natural map $h:Y \to M_f.$ 
Thus, for any $f \in [Y,X]$
the isomorphisms in the category $\C_X$ 
together with the construction of the mapping cone of $f$ determine the
pullback of pairs up to isomorphism. Hence, we can study
$C_X$ together with its natural map to 
$C_{M_f}$ for every $f \in [Y,X]$ which is not
a homeomorphism and obtain the behaviour of Topological T-duality
pairs under pullback along the map $f$
(up to isomorphism of pairs over a fixed base).

We will use this argument in SubSec.\ (\ref{SSecGXCX}) with $Y$
replaced by $X$ to motivate the construction of the functor $\G.$

\subsection{The subcategories $\B_X,\F_X,\pi(\F_X)$ \label{SSecBXFX}}
After Thm.\ (\ref{ThmHCTEqv}) above we had argued that
the functor $E: \HC \to \mathbf{Top}$ was a category
fibered in groupoids over the site $\mathbf{Top}.$ 
In particular we claimed that the fiber of $E$ over
a topological space $X$ should, in a sense, be viewed
a generalization of $X.$

We will see in SubSec.\ (\ref{SSecNMFE})
below that the functor $E$ is equivalent to a
groupoid-valued presheaf $\M$
on $\mathbf{Top}.$ We will show that
the value of this groupoid
on a topological space is the essential fiber of
the functor $E,$ namely the small groupoid 
$\pi(\F_X)$ defined
in Lem.\ (\ref{LemPiFXGrpd}) below.

We will see in SubSec.\ (\ref{SSecHotG}) below
that the set of objects of the small groupoid $\pi(\F_X)$
is the set $P(X)$ where $P$ is Bunke-Schick's functor. 
Thus, the presheaf $\M$ is a generalization of Bunke-Schick's
functor $P(X).$ In addition we show in Sec.\ (\ref{SecP_i})
that we can use this presheaf to define higher topological 
T-duality functors.

In this subsection, for each topological space $X,$ 
we first naturally construct two subcategories 
$\B_X, \F_X$ of $\C_X$ from Bunke-Schick's functor 
$P$ and its lift to the category $\C$ 
defined in SubSec.\ (\ref{SSecPCLift}) above. We then 
define the category $\pi(\F_X)$ which is a subcategory of 
$\HC$ and prove that it is a groupoid.

The subcategory $\B_X$ of $\C_X$ has as objects all pairs over the topological space  $X$ and as morphisms all isomorphisms of any pair
over $X$ as in part (a) of Def.\ (\ref{DefCCat1}). 
Due to this, all the arrows in the category 
$\B_X$ cover the identity arrow $\id:X \to X.$ We had 
calculated the behaviour of morphisms in the subcategory $\B_X$ in SubSubSec.\ (\ref{SSecMorPXXid}) above.

The subcategory $\F_X$ of $\C_X$ has as objects all pairs over the topological space $X$ and as arrows all isomorphisms between pairs over $X$ covering self-homeomorphisms $f:X \to X.$ We had calculated
the behaviour of morphisms in the subcategory $\F_X$ in
SubSubSec.\ (\ref{SSecMorPXXf}) above.
We will discuss the subcategory $\pi(\F_X)$ of $\HC$ at the end
of this subsection. 

We now show that both the above subcategories of
$\C_X$ namely $\B_X$ and $\F_X$
may be naturally constructed from the functors 
$P\circ F$ and $F$ as the essential fiber of these 
functors over an object of their respective target categories. 
We will use these subcategories which we have constructed in this subsection throughout the paper.

We will need the definition of the essential fiber of a Functor over
an object in the target category of the functor
(from the nCatLab `Essential
Fiber' Webpage---see Ref.\ \cite{nLab-EssFib}). The 
essential fiber of a Functor $p: \E \to \B$ between two
categories $\E$ and $\B$ is an analogue of the homotopy
fiber of a map between two spaces (see Ref.\ \cite{May} Ch.\ (8) Sec.\ (6)
and Ref.\ \cite{nLab-EssFib} Remark (2.2)). 
By Ref.\ \cite{nLab-EssFib}, the essential fiber of a functor $p$ 
may be viewed as the category which represents the
possible ways by which we may obtain 
$b$ in $\B$ by applying $p$ to an object
of $\E.$
\begin{definition}
Let $p: \E \to \B$ be a functor and $b \in \B$ an object, the
essential fiber of $p$ over $b$ is the category whose:
\begin{enumerate}
\item objects are pairs $(e, \phi)$ where $e \in \E$ is an object
and $\phi: p(e) \simeq b$ is an isomorphism in $\B.$
\item morphisms $(e, \phi) \to (e', \phi')$ are morphisms
$f:e \to e'$ in $\E$ such that $\phi' \circ p(f) = \phi,$
\item the composition operation is the obvious one.
\end{enumerate}
\end{definition}

We now show that the essential fiber of the forgetful
functor $F: \C \to \mathbf{Top}$ over a topological 
space $X$ is the subcategory $\F_X$ of $\C_X$ defined
at the beginning of this subsection.
\begin{theorem}
The essential fiber of $F: \C \to \mathbf{Top}$ over a topological
space $X$ is the category $\F_X.$ This category is a groupoid.
\label{ThmFXDef}
\end{theorem}
\begin{proof}
By definition, $F$ sends the pair $(E \to X, H)$ to $X.$
The essential fiber of $F$ over $X$ is a category
whose objects are ordered pairs $(e, \phi)$ with
$e = (E \to Y, H_Y)$ a pair in $\C$ together
with a homeomorphism $\phi: Y = F(e) \to X $ in $\mathbf{Top}.$
Thus the objects of the essential fiber of $F$ may be taken as
all pairs over $X.$

Let $(e, \phi)$ be in the essential fiber with $e = (E \to Y, H_Y)$ with
$\phi:Y \to X$ a homemorphism.Let $(e',\phi')$ be another element
of the essential fiber with $e' = (E' \to W, H_W)$ with $\phi': W \to X.$
A morphism from $(e,\phi)$ to $(e', \phi')$ is a morphism of pairs 
$f:e \to e'$ in $\C$ such that 
$\phi' \circ F(f) = \phi.$ That is, morphisms between $(e,\phi)$ and
$(e', \phi')$ in the essential fiber are a morphism of pairs in $\C$ covering
a map $F(f): Y \to W$ commuting with the homeomorphisms
$\phi, \phi'$ above. Hence, $F(f)$ must be a homeomorphism in 
$\mathbf{Top}$ and $f$ must be a morphism of pairs in the
sense of Def.\ (\ref{DefCCat2}) covering
a self-homeomorphism of the base which we had argued in
the previous paragraph may be taken as $X.$ 
Thus any morphism in the essential fiber category
must be an isomorphism of pairs in $\C$ over $X$ covering a
self-homeomorphism of the base space $X$ (see discussion
after Def.\ (\ref{DefCCat2}) near the end of that subsubsection with
$Y \simeq X$).
This is just the category $\F_X$ defined above.

Since every arrow in $\F_X$ is invertible by definition, it is
a groupoid, in fact, since $\F_X$ is a small category it 
is a small groupoid (see Remark (2.4) and 
Def.\ (2.6) of nCatLab page on Groupoids, Ref.\ \cite{nLab-Grpd}).
\end{proof}

The essential fiber of the functor $F: \C \to \mathbf{Top}$ 
over a topological space $X$ may be viewed as the category of all
possible ways of obtaining $X$ by applying $F$ to a pair in $\C.$
Also, note that any self-homeomorphism of $X$ acts on 
the category $\F_X$ by pulling back each pair and each morphism
and hence yields a new category isomorphic to $\F_X$ i.e. 
every self-homeomorphism of $X$ induces 
an automorphism of $\F_X$ by pullback (see Remark (4.5) of
Ref.\ \cite{nLab-EssFib}).  

We may study the lift of Bunke-Schick's functor $P$ to $\C$ by
restricting it to $\C_X$ for every topological space $X.$
Consider the assignment $(E \to X, H) \mapsto [(E \to X, H)].$
This is the functor $P \circ F$ restricted to $\C_X.$ Its image
is the set $S = P(X).$ We restrict $P \circ F$ to $\C_X$ by making the
target category (denoted $\T$) a single set $S$ with morphisms only 
the identity $f: S \to S.$

We now show that the essential fiber of $P \circ F$ restricted to $\C_X$ is
the subcategory $\B_X$ of $\C_X$ defined at the beginning of this
subsection. 
\begin{theorem}
Let $\T$ be a set $S$ viewed as a
category consisting of a single set $S$ as an object with
morphism only the identity function on $S.$ 
The essential fiber of $P \circ F: \C_X \to \T$ is 
the subcategory $\B_X$ of $\C_X.$ This category is a groupoid.
\label{ThmBXDef}
\end{theorem}
\begin{proof}
By definition $P \circ F$ sends the pair $e = (E \to X, H)$ in $ \C $ to 
$P(X).$  The essential fiber of $P \circ F$ over $S$ in $\T$
is a category whose objects are ordered pairs $(e, \phi)$ where
$e =(E \to X, H)$ is a pair in $\C$ and 
$\phi: P \circ F(e) \simeq S$ is an isomorphism.

Let $e' = (E' \to X, H')$ be another pair in $\C_X$ such that we have a bijection of sets $\phi':P\circ F(e') \simeq S.$ Note that $F(e') = X$
as well. Hence $\phi': P(X) \simeq P(X)$ must be the identity.

A morphism between $(e, \phi)$ and $(e', \phi')$ is a morphism
of pairs in $\C$ namely $f: e \to e'$ such that 
$\phi' \circ P \circ F(f) = \phi.$ Thus, the underlying map of
topological spaces $F(f): X \to X$ must induce the identity map
on the set $P(X)$ by the above.
Thus, the essential fiber of $P \circ F$ is a subcategory of
$\C$ consisting of pairs over $X$ and morphisms between pairs over
$X$ covering the identity map $id:X \to X.$ This is exactly
the definition of Bunke-Schick isomorphism of pairs over $X$
in Def.\ (\ref{DefCCat1}) above and is the category
$\B_X.$

Every arrow in $\B_X$ is invertible and hence it is a groupoid,
since $\B_X$ is a small category, it is actually a small
groupoid (see Remark (2.4) and Def.\ (2.6) of Ref.\  \cite{nLab-Grpd}).
\end{proof}

The essential fiber of $P \circ F: \C_X  \to \T$ may be viewed as
the category of possible ways of obtaining $S$ by applying $P \circ F$ to 
pairs $e$ over $X.$ Note that
there are no automorphisms in $\T,$ so despite Remark (4.5) of
Ref.\ \cite{nLab-EssFib}, Thm.\ (\ref{ThmBXDef}) does not give us any 
natural automorphisms of $\B_X.$
However, this is only an artefact of the construction above, any
map $f:Y \to X$ induces a functor $\B_X \to \B_Y$ as we now
argue:
\begin{theorem}
Any map $f:Y \to X$ induces a functor $\B_X \to \B_Y.$
\end{theorem}
\begin{proof}
We may lift any isomorphism of pairs in $\B_X$ which covers
the identity on $X$ to a map $g:X \to S^1$ (see 
SubSubSec.\ (\ref{SSecMorPXXid})). This will pullback by composition
with $f$ to give $g \circ f: Y \to S^1$ and hence give an  isomorphism of the pullback pairs on $Y.$ 
Homotopic maps $g_1, g_2: X \to S^1$ will pullback
to homotopic maps. Hence we obtain an assignment sending objects
in $\B_X$ to objects in $\B_Y$ and similarly for arrows.
Also, the identity arrow on a pair in $\B_X$ lifts to the
constant map on $X$ with value $1 \in S^1.$ Under composition by
$f$ this yields the constant map on $Y$ with value $1 \in S^1$ and
hence the identity arrow on the pullback pair.

All the requirements for a functor are satisfied except that
we need to check that composition of automorphisms in $\B_X$ correspond to pointwise products of maps $g_1g_2$ on $X.$
However pointwise products of maps $g_1g_2$ will
pullback to pointwise products of maps on $Y$ namely 
$(g_1 \circ f) (g_2 \circ f).$ 

Due to this we obtain the required functor.
\end{proof}

Recall we have a functor $\pi: \C \to \HC$ which assigns to
any pair over $X$ the homotopy equivalence class of the pair
over $X.$ Clearly, the image of the subcategory 
$\F_X$ above under $\pi$ is a subcategory of $\HC.$ 

We now prove that $\pi(\F_X)$ is a groupoid.
\begin{lemma}
Let $\F_X$ be the subcategory of $\C$ defined above.
Let $\pi: \C \to \HC$ be the functor defined in 
SubSec.\ (\ref{SSecPCLift}) above.
\leavevmode
\begin{enumerate}
\item The subcategory $\pi(\F_X)$ of $\HC$ is a groupoid.
\item Let $f:e \to e'$ be any isomorphism of pairs 
in $\F_X,$ then there exists a pullback isomorphism of pairs
$p$ in $\F_X$ between two pairs homotopy equivalent to
$e,e'$ respectively such that $\pi(f) = \pi(p).$
\end{enumerate}
\label{LemPiFXGrpd}
\end{lemma}
\begin{proof}
\leavevmode
\begin{enumerate}
\item The objects of $\pi(\F_X)$ are the objects of $\HC$
and hence any object in $\pi(\F_X)$ 
is a collection of homotopy equivalent
pairs over $X.$ Any morphism $f:\pi(e) \to \pi(e')$ in
$\pi(\F_X)$ (with $e, e'$ pairs in $\C$) is of the form
$f = \pi(\tilde{f})$ where $\tilde{f}: e \to e'$ is a morphism
in $\F_X$ which is a subcategory of $\C.$

Hence $\tilde{f}$ must be an isomorphism of
pairs, by definition. Thus, there must be an inverse for
$\tilde{f}$ in $\F_X$ which, without loss of generality,
may be taken as $\tilde{g}:e' \to e$ such that 
$\pi(\tilde{g}) = g$ is a morphism in $\HC.$

Clearly, $g: \pi(e') \to \pi(e)$ and $g \circ f = 1_e$ and
$f \circ g = 1_{e'}.$ Thus $f$ is an iso arrow in $\pi(\F_X).$

Note that picking a different lift of $f$ in $\F_X$ 
will give the same morphism $g$ by uniqueness of $g.$
Since $f$ was arbitrary, $\pi(\F_X)$ must be a groupoid.
\item Note that by the discussion after 
Def.\ (\ref{DefCCat2}) near the
end of that subsection and by part (2) of
Thm.\ (\ref{ThmCSimCat}) if we are given an 
arbitrary isomorphism $f$ between two pairs in $\F_X,$
there is a pullback isomorphism $p$ between those two pairs
(which must be a morphism between those two pairs in 
$\F_X$ as well) such that $\pi(f) = \pi(p).$
Thus, the image in $\pi(\F_X)$ of an arbitrary 
isomorphism of pairs in $\F_X$ under the functor
$\pi$ is the same as the image of some pullback
isomorphism of pairs in $\F_X$ under the functor $\pi.$ 
\end{enumerate}
\end{proof}

We now calculate the essential fiber of the functor $E:\HC \to \mathbf{Top}$
defined in SubSec.\ (\ref{SSecPCLift}) over a topological
space $X.$ Recall that the functor
$F: \C \to  \mathbf{Top}$ above factorized naturally
as $F \simeq E \circ \pi$ where $E: \HC \to \mathbf{Top}$ was formally
defined in Thm.\ (\ref{ThmHCTEqv}) above and $\pi: \C \to \HC$ was the
functor which on acted on objects by sending
pairs in $\C$ to their Bunke-Schick isomorphisms of pairs.

We first show that the fiber of the functor $E$ over
a topological space $X$ is the subcategory of 
$\HC$ consisting of the set of Bunke-Schick 
isomorphism classes of pairs over
$X$ with only identity arrows as morphisms. This is
expected since, as we argued in the paragraph
after Thm.\ (\ref{ThmHCTEqv}) above, the functor 
$E$ was obtained by the Grothendieck construction from
the functor $P: \mathbf{Top}^{op} \to \mathbf{Set}$
(since $P$ is a contravariant functor from $\mathbf{Top}$
to $\mathbf{Set}$, we use the Grothendieck 
construction for a contravariant functor described
in Sec.\ (2) of Ref.\ (\cite{nLab-GCons})).

However, $E$ has more structure than the above
over any topological space $X.$
We now show that the essential fiber of $E$ over a topological space
$X$ is the groupoid $\pi(\F_X)$ discussed above.
We also prove a result about pullback isomorphisms
induced on $\pi(\F_X)$ by maps $f:Y \to X.$
\begin{theorem}
Let $E:\HC \to \mathbf{Top}$ be the functor
defined in SubSec.\ (\ref{SSecPCLift}), 
Thm.\ (\ref{ThmHCTEqv}) above. 
\leavevmode
\begin{enumerate}
\item The fiber of $E$ over a topological space $X$ is
a set of points naturally bijective with $P(X).$
\item The essential fiber of $E$ over a topological space $X$
is the groupoid $\pi(\F_X).$
\item Let $\pi(\F_X)$ be as in Thm.\ (\ref{ThmPiFXDef})
above. Then, any map $f: Y \to X$ induces a natural
functor $\pi(f): \pi(\F_X) \to \pi(\F_Y).$
Changing $f$ in its homotopy class does not change
$\pi(f).$
\end{enumerate}
\label{ThmPiFXDef}
\end{theorem}
\begin{proof}
\leavevmode
\begin{enumerate}
\item Consider the fiber of the functor
$E:\HC \to \mathbf{Top}.$ By definition
this consists of pairs $\pi(e), \pi(e')$ in $\HC$ with
$e,e'$ pairs in $\C$ such that $E(\pi(e)) \simeq X$
and $E(\pi(e')) \simeq X$ and all morphisms $f: \pi(e) \to \pi(e')$
such that $E(f) \simeq \id_X.$ This implies that
$E(\pi(e)) \simeq E(\pi(e')) \simeq X.$ By 
Thm.\ (\ref{ThmCSimCat}) above,
$f = \pi (\phi)$ where $\phi: e \to e'$ is a morphism in $\C.$
We have that 
$E(f)  \simeq \id_X,$ $E \circ \pi \circ \phi \simeq \id_X$
hence $F(\phi) \simeq \id_X.$ Thus, $\phi$ must cover the
identity map on $X$ and hence must be a Bunke-Schick
isomorphism over $X$ between $e$ and $e'.$ 
Thus the image of
$\phi$ under $\pi$ must be the identity map in $\HC,$
by definition and $\pi(e) \simeq \pi(e')$ in $\HC.$
 
Due to this the fiber of the functor $E$ is the 
subcategory of $\HC$
whose objects are Bunke-Schick
isomorphism classes of pairs over $X$
and which has no morphisms except the identity morphism
on each object. 

Thus, the fiber of the functor $E$ is
a small subcategory of $\HC$ and its underlying
set is naturally bijective with $P(X).$
\item Let $\pi(a)$ be the image of a pair $a = (E_X \to X, H_X) \in \C$ 
in the category $\HC.$ The essential fiber of the functor $E$ above
over a topological space $X$ 
is a category whose objects are ordered pairs $(\pi(a), \phi)$ with
$a = (E_X \to X, H_X)$ a pair in $\C$ and $\phi:Y  = E(\pi(a)) = F(a)  \to X$
a homeomorphism in $\mathbf{Top}.$ Thus the objects in the
essential fiber category over a topological space $X$
may be taken as homotopy equivalence
classes of pairs over $X.$

Let $(\pi(e), \phi)$ be an object in the essential fiber of the functor 
$E$ above over a topological space $X$
with $e = (E_Y \to Y, H_Y)$ a pair in $\C$
and with $\phi: Y \to X$ a homeomorphism. Let $(\pi(e'), \phi')$ be another
object in the essential fiber of the functor $E$
over $X$ with $e' = (E_W \to W, H_W)$ another
pair in $\C$ and with $\phi': W \to X$ a homeomorphism. 
A morphism from $(\pi(e),\phi)$ to $(\pi(e'), \phi')$
is a morphism of homotopy classes of 
pairs $\tilde{f}: \pi(e) \to \pi(e')$ such that
$\phi' \circ E(\tilde{f}) = \phi.$ 
By Thm.\ (\ref{ThmCSimCat}) above,
we may lift $\tilde{f}: \pi(e) \to \pi(e')$ to a morphism of
pairs $f: e \to e'$ in $\C$ which covers a map of spaces
$F(f): F(e) \simeq Y \to F(e') \simeq W$ such that the map of pairs
$f$ in $\C$ is a {\em pullback} of the pair $e'$ along $F(f).$
Note that $F(f) = E(\tilde{f}) = E(\pi(f))$ since $E \circ \pi \simeq F.$

Thus morphisms between two objects
$(\pi(e), \phi)$ and $(\pi(e'), \phi')$ in the 
essential fiber of the functor $E:\HC \to \mathbf{Top}$ 
over a topological space $X$
may be identified with a morphism of pairs in $\C$ covering
a map of topological spaces $F(f):Y \to W$ such that 
$\phi' \circ F(f) = \phi$ with $\phi', \phi$ homeomorphisms in 
$\mathbf{Top}.$ As a result, $F(f)$ must be a homeomorphism in 
$\mathbf{Top}$ and $f$ above must be a morphism of pairs in $\C$ which
is also a pullback in $\C$ of a pair over $X$
along a self-homeomorphism of the base 
(which we had argued above can be taken as $X$). In particular $f$ must
be an isomorphism of pairs.

Thus, any morphism in the essential fiber category of the functor $E$
over a topological space $X$ may be taken as
an image under $\pi: \C \to \HC$ of an isomorphism of pairs in 
$\C$ over $X$ covering pairs in $\HC$ --- this isomorphism of pairs
in $\C$ must be a {\em pullback} of a pair over $X$ 
along a self-homeomorphism of $X.$ 
This is the category $\pi(\F_X)$ by Thm.\ (\ref{ThmCSimCat}) above.

\item As in the previous part, any morphism between two 
elements of $\pi(\F_X)$ may be lifted to a morphism of pairs
in $\C$ which is a {\em pullback} morphism over a nontrivial 
map $g:X \to X.$ Each of these pairs and the pullback morphisms
between them may then be pulled back along $f:Y \to X$ in 
$\C.$ It is clear that the image of the pullbacks under $\pi$ is still
$\pi(\F_X).$ Further, the pullback of the composite of two maps
must be the composite of their pullbacks since each map lifts, the lifts
must be pullbacks and the composite of the lifts maps to the
composite of the maps under $\pi$ by definition.
Thus, any map $f: Y \to X$ induces a natural functor
$\pi(f): \pi(\F_X) \to \pi(\F_Y).$ 

Changing $f$ in its homotopy class induces a Bunke-Schick
isomorphism on pullback of pairs in $\C$ 
by Lemma (\ref{LemHotIso}) above.
Thus, the image of pairs in $\HC$ does not change.  
Also $\pi(f)$ does not change when $f$ changes
in its homotopy class.
\end{enumerate}
\end{proof}

\subsection{Construction of the functors $\N$ and $\M$ \label{SSecNMFE}}
In this subsection we first naturally construct a presheaf
of groupoids
$\N: \mathbf{Top}^{op} \to \mathbf{Gpd}$ 
from the Bunke-Schick functor
$P$ and the functor $F$ defined above --- see SubSec.\ \ref{SSecPCLift}).
The functor $\N$ assigns to each topological space $X$ 
the category $\F_X$ defined in SubSec.\ (\ref{SSecBXFX}) above.
That is, this functor assigns the groupoid $\N(X)$ 
(see Def.\ (2.6) of Ref.\ \cite{nLab-Grpd}) to any
topological space $X.$ 

We then use the construction of the
functor $\N$ above to construct a natural
presheaf of groupoids
$\M: \mathbf{Top}^{op} \to \mathbf{Gpd}$ using
the Bunke-Schick functor $P$ and the functor
$E$ defined above in SubSec.\ (\ref{SSecPCLift}). 
This functor $\M$ assigns to each topological space
$X$ the small groupoid $\pi(\F_X)$ defined in 
Lem.\ (\ref{LemPiFXGrpd}) above. Later in this subsection
we construct a space $G(X)$ from $\M(X)$ by passing
to the classifying space of $\M(X).$

We argue about the relation of the above with string field
theory in SubSec.\ (\ref{SSecSFT-TTD}) below.
In the later sections of this paper (see for example
SubSec.\ (\ref{SSecHotG}) or Sec.\ (\ref{SecP_i}) below)
we will use the groupoid $\M(X)$ and the space
$G(X)$ to define a higher Topological T-duality functor.

In SubSec.\ (\ref{SSecDefCX}) we had defined the subcategory
$\C_X$ of $\C$ and discussed some of its properties. In this
subsection we will construct a groupoid
from $\C_X$ by ignoring all maps in $\C_X$ 
which are not isomorphisms. This groupoid is the category
$\F_X$ defined in SubSec.\ (\ref{SSecBXFX}) above.

It is not clear why this is natural, since $\C_X$ has
many more maps in it than just isomorphisms of pairs. 
To make this construction plausible we argue as follows:
We had shown in Sec.\ (\ref{SecIntro}) just after 
Lemma (\ref{LemHomTf}) that pullback of pairs in 
$\C_X$ under an arbitrary map $f:Y \to X$ could be calculated 
using pullback along a cofibration map $h: Y \to M_f$ 
(where $M_f$ is the mapping cone of $f$ constructed
from $X, Y$ and $f$ --- see Ref. \cite{May}).

We showed in that Section that the behaviour of an arbitrary pair over 
$X$ up to isomorphism of pairs under an arbitrary map $f:Y \to X$ 
was determined by pulling back a pair on $M_f$ induced by the given
pair on $X$ along a natural map $h:Y \to M_f.$ 

Thus, given $f:Y \to X$ the isomorphisms in the category $\C_X$ 
together with the construction of the mapping cone of $f$ determine the
pullback of pairs up to isomorphism.
This will apply if we replace $Y$ by $X$ and as a result we can
study the behaviour of morphisms of pairs over $X$ covering
an arbitrary map $f:X \to X$ which is not the identity or a homeomorphism
(we had studied these last two cases in SubSec.\ (\ref{SSecMorC}) above).

Motivated by this, for any topological space $X$ 
we first construct a groupoid (see Ref.\ \cite{nLab-Grpd})
from $\C_X$ (this groupoid is termed the core groupoid
of $\C_X$ and there is a well defined functor $\Core$
which yields a groupoid denoted $\Core(\C_X)$ when
it acts on $\C_X$ --- we will discuss this in detail below) by 
ignoring all maps $f:X \to X$ which are not homeomorphisms. 
Pullback of pairs along an arbitrary map $f:Y \to X$
induces well-defined morphisms on this groupoid since isomorphic
pairs pullback to isomorphic pairs by Lemma (\ref{LemPairIso}).

It is clear that this will give us the groupoid
$\F_X$ of $\C_X$ which we had defined above in 
SubSec.\ (\ref{SSecBXFX}). If we repeat the construction of
the contravariant functor $\G$ in Thm.\ (\ref{ThmGFunc}) above
but with $\F_X$ instead of $\C_X$ we would obtain a contravariant
functor $\N:\mathbf{Top}^{op} \to \mathbf{Gpd}.$ This functor
would send each topological space $X$ to the groupoid $\F_X$
and send a map between spaces $f:Y \to X$ to a groupoid 
morphism $\F_X \to \F_Y.$

However, we prefer another way to construct the functor $\N.$
We now construct the above contravariant functor $\N$ naturally 
from $\C$ and the functor $F$ in SubSec.\ (\ref{SSecPCLift}) using
the idea of a Grothendieck Fibration (see nCatLab page 
Ref.\ \cite{nLab-GFib}):
\begin{theorem}
Let $F: \C \to \mathbf{Top}$ be the functor defined in
(\ref{SSecPCLift}) above. Let $\M:\mathbf{Top}^{op} \to \mathbf{Gpd}$ 
be the functor defined in the previous paragraph.
\begin{enumerate}
\item The functor $F$ is a Grothendieck Fibration in which the
cartesian arrows are pullback maps in $\C$ and hence a Street Fibration. 
\item The functor $F$ is categorically equivalent to a functor 
$\NN:\mathbf{Top}^{op} \to \mathbf{Gpd}$
whose value at a topological space $X$ is the essential
fiber of $F$ over that topological space.
\item The functor $\NN:\mathbf{Top}^{op} \to \mathbf{Gpd}$ is
exactly the functor $\N: \mathbf{Top}^{op} \to \mathbf{Gpd}$
described in the paragraph before this theorem.
\end{enumerate}
\label{ThmNGFNN}
\end{theorem}
\begin{proof}
\leavevmode
\begin{enumerate}
\item We argue that $F: \C \to \mathbf{Top}$ is a Grothendieck
Fibration (see Ref.\ \cite{BRichter} Def.\ (11.6.8) or see nCatLab page
on Grothendieck Fibrations Ref.\ \cite{nLab-GFib}). The cartesian 
arrows (see also Ref.\ \cite{LR} Def.\ (2.21) and Thm.\ (2.2.3))
in this Grothendieck Fibration are the pullback maps in $\C$ 
covering an arbitrary map $f:Y \to X.$ 

Suppose $f: u \to v$ was a pullback map in $\C$ and 
$ g: x \to v$ was an arbitrary morphism of pairs in $\C.$
Suppose we had any map $\hat{h}: F(u) \to F(x)$ such
that $F(g) \circ \hat{h} = F(f).$ 

Then, we could factor $g$ as in the discussion after Def.\ (\ref{DefCCat2})
into a composite of an isomorphism of pairs over $F(x)$ 
namely $\iota:x \to x$ and 
a pullback map $p: x \to v,$ that is $g = p \circ \iota.$ By the
discussion after Def.\ (\ref{DefCCat2}) the arrow $\iota$ is obtained
from the universal property of pullback and hence must be unique.

We define $\tilde{h}: u \to x$ by pulling back the pair $p^{\ast}(v)$
along $\hat{h}: F(u) \to F(x).$ This is clearly unique.
We define the lift $w$ of $\hat{h}$ to $\C$ by picking
$h = (\iota)^{-1} \circ \tilde{h}.$ By the discussion above $h$
is unique.
Then, 
$$
g \circ h = (p \circ \iota) \circ ( (\iota)^{-1} \circ \tilde{h}) = p \circ \tilde{h}.
$$
However, $p \circ \tilde{h}$ is a composition of pullback squares
along the composite map $F(g) \circ \hat{h}: F(u) \to F(v)$ and hence
must equal the pullback along $F(f)$ since we have
$F(g) \circ \hat{h} = F(f).$ 

Thus $F: \C \to \mathbf{Top}$ is a Grothendieck Fibration.
Hence, $F$ is also a Street Fibration since every
Grothendieck Fibration is automatically a Street Fibration
(see nCatLab page on Street Fibrations 
Ref.\ \cite{nLab-StrFib} Items (1) and (4)).
\item Since $F$ is a Street Fibration, the value of the functor 
$\mathbf{Top}^{op} \to \mathtt{Cat}$ obtained
from $F$ using the Grothendieck Construction when
evaluated at a topological space $X$ 
is the essential fiber of $F$ at $X$ i.e., $\F_X$ (see Item (4) of
nCatLab page on Street Fibrations Ref.\ \cite{nLab-StrFib}). 
Hence, the Grothendieck Construction gives us a presheaf on 
$\mathbf{Top},$ i.e. a functor $(\mathbf{Top})^{op} \to \F$ sending 
$X \mapsto \F_X.$ By Item (4) of the nCatLab page on Street Fibrations
the functor sends $f:Y \to X$ to a pullback functor $\F_X \to \F_Y.$
\item We use the Grothendieck Construction for Street Fibrations
(see nCatLab page on Grothendieck Construction Item (1) 
under discussion of categorical equivalence between fibrations on $\B$ and
presheaves on $\B^{op}$ in Ref.\ \cite{nLab-GCons} and also the
nCatLab page on Street Fibrations Ref.\ \cite{nLab-StrFib} Item (4)) and 
the equivalence between the Grothendieck Fibration
$F: \C \to \mathbf{Top}$ and a functor 
$\mathbf{Top}^{op} \to \mathtt{Cat}$ where $\mathtt{Cat}$ is the
category of small categories with functors between them as arrows.
This functor, by the argument in the previous item, must assign
$X \to \F_X$ and send $f:Y \to X$ to the pullback functor
$\F_X \to \F_Y.$ As a result it must be the functor 
$\M:\mathbf{Top}^{op} \to \mathbf{Gpd}$ described in Part (2) above.
\end{enumerate}
\end{proof}

In Thm.\ (\ref{ThmNGFNN}) we had defined the functor 
$\N$ from the functor $F$ above (see SubSec.\ (\ref{SSecPCLift}) 
by categorical equivalence. In Thm.\ (\ref{ThmCSimCat}) above
we had constructed a functor $E: \HC \to \mathbf{Top}$ from 
$F: \C \to \mathbf{Top}$ which satisfied 
$E \circ \pi \simeq F.$ We now 
construct a functor $\M: \mathbf{Top}^{op} \to \mathbf{Gpd}$
from $E$ using categorical equivalence. 

We would expect that
$\M(X) \simeq \pi(\F_X)$ where $\pi: \C \to \HC$ is the
quotient functor from the category of pairs to the category
of pairs up to homotopy. By Thm.\ (\ref{ThmPiFXDef}) this is
the essential fiber of the functor $E:\HC \to \mathbf{Top}$
defined in SubSec.\ (\ref{SSecPCLift}) above.
\begin{theorem}
Let $E: \HC \to \mathbf{Top}$ be the functor defined after Thm.\ 
(\ref{ThmCSimCat}) above and let 
$\M: \mathbf{Top}^{op} \to \mathbf{Gpd}$ be the functor defined
in the previous paragraph.
\leavevmode
\begin{enumerate}
\item The functor $E$ is a Grothendieck Fibration in which every 
arrow is cartesian. Hence it is a Street Fibration.
\item The functor $E$ is categorically equivalent to a functor
$\MM: \mathbf{Top}^{op} \to \mathbf{Gpd}$ whose value at
a topological space $X$ is the essential fiber of $E$ over that topological
space.
\item The functor $\MM: \mathbf{Top}^{op} \to \mathbf{Gpd}$
is exactly the functor $\M: \mathbf{Top}^{op} \to \mathbf{Gpd}$
described in the paragraph before this theorem.
\label{ThmMGFMM}
\end{enumerate}
\end{theorem}
\begin{proof}
\leavevmode
\begin{enumerate}
\item We argue that any arrow in $\HC$ is a cartesian arrow
for $\E.$ 
Suppose $u,v, w \in \C$ and we were given an arrow $f:\pi(u) \to \pi(v)$
in $\HC$ and an arbitrary arrow $g:\pi(w) \to \pi(v).$  Suppose we
had a map $\tilde{h}:E(u) \to E(w)$ such that $E(g) \circ \tilde{h} = E(f).$

Note that using Thm.\ (\ref{ThmCSimCat}) above we can assume
that the lifts of $f,g$ to $\C$ induce pullback maps
$f': u \to v$ and $g': w \to v$ perhaps after changing $u,v,w$ in
their Bunke-Schick isomorphism classes.

Here $\tilde{h}, E(u), E(w)$ are maps of topological spaces.
Also, the pullback of pairs $\tilde{h}^{\ast} (E(g))^{\ast} v$ must be 
Bunke-Schick isomorphic to the pair $u$ since both are pairs over $E(u)$
and $E(g) \circ \tilde{h} = E(f)$ and the lift $E(g)^{\ast}$ of $E(g)$ to $\C$
may be taken to be $g$ itself and similarly for $f.$

Thus, the pullback map $h$ in $\C$
which consists of pullback along $\tilde{h}$ in $\C$  is such
that $\pi(h)  = \tilde{h}.$ 

This is the required cartesian lift of $\tilde{h}.$

Thus $E: \HC \to \mathbf{Top}$ is a Grothendieck Fibration. Hence,
$E$ is also a Street Fibration.
\item This follows from the previous part and the 
corresponding parts of the proof of 
Thm.\ (\ref{ThmNGFNN}) with $\NN$ there replaced by
$\MM$ above.
\item This follows from the previous two parts and the
corresponding parts of the proof of 
Thm.\ (\ref{ThmNGFNN}) with $\N$ there replaced by
$\M$ above.
\end{enumerate}
\end{proof}

\subsection{Construction of $G(X)$ from $\C_X$ \label{SSecGXCX}}
Recall that we defined a natural functor
$E: \HC \to \mathbf{Top}$ in Thm.\ (\ref{ThmHCTEqv}) above.
We proved in Thm.\ (\ref{ThmMGFMM}) above
that the functor $E$ endowed the category $\HC$ with the
structure of a Grothendieck fibration over $\mathbf{Top}.$ 
As a result, we argued in that subsection
that $E$ was also a Street fibration over $\mathbf{Top}.$
We also proved that the fiber of this fibration was the
small groupoid $\pi(\F_X).$ In SubSec.\ (\ref{SSecNMFE}) 
we constructed a functor
$\M: \mathbf{Top}^{op} \to \mathbf{Gpd}$ 
categorically equivalent (see Ref.\ \cite{nLab-GFib}) to 
the functor $E$ under categorical equivalence for Street fibrations.
We proved in that subsection that
the functor $\M$ assigned to any topological space $X$
the small groupoid $\pi(\F_X).$

In SubSec.\ (\ref{SSecPCLift}) after 
Thm.\ (\ref{ThmHCTEqv})
we argued that objects of the category $\HC$
should be viewed as generalizations of topological
spaces. In particular the value of the functor $\M$ on
a topological space $X,$ namely, the small groupoid 
$\M(X) \simeq \pi(\F_X),$ should be viewed 
as a `generalized space' corresponding to $X$
(see discussion after Thm.\ (\ref{ThmHCTEqv}) above ---
recall that $\pi(\F_X)$ is a subcategory of $\HC$
by Thm.\ (\ref{ThmMGFMM}) and 
Thm.\ (\ref{LemPiFXGrpd}) above) whose objects
are all pairs over $X.$ Thus using the groupoid $\M(X)$
instead of the space $X$ corresponds to studying $X$
using all pairs over $X.$

We show in Sec.\ (\ref{SecP_i}) that, for any topological space $X,$ 
the groupoid $\M(X)$ has two main types of invariants, the fundamental
group of its object space $(\M(X))_0$ and the automorphism
group of any object in $(\M(X))_0.$ We
discuss the fundamental group of $(\M(X))_0$
later in this section and show that Bunke-Schick's
functor $P(X)$ maps onto this fundamental group.
In Sec.\ (\ref{SecP_i}) we define a higher version of Bunke-Schick's
functor $P$ by calculating the automorphism group of an 
object in $(\M(X))_0$ using the fundamental
groupoid of the classifying space $B\M(X)$
of the small category $\M(X)$ (see 
Ref.\ \cite{BRichter}).  

More precisely, we construct
a new functor which we term
$G: \mathbf{Top}^{op} \to \mathbf{Top}$ which
assigns a topological space $G(X) \simeq B\M(X)$ to any
topological space $X.$ This space $G(X)$ is the
geometric realization of the simplicial space associated
to the groupoid $\M(X).$ We identify $G(X)$
with the moduli space of SFT vacua over 
$X$ whose points are equivalence classes of
SFT backgrounds which
are pairs over $X$ quotiented out by SFT gauge transformations
which are induced from the bundle gauge
transformations of the underlying bundle of
each pair. We see in Sec.\ (\ref{SecP_i})
that we may define a higher topological
T-duality functor from the homotopy groups of 
$G(X)$ for any topological space $X.$

We begin by defining the $\Core$ functor for any
category $\A.$ We then discuss the construction of
the functor $G$ from the functor
$\M$ above. We show that $G(X)$ is the
classifying space of the category 
$\pi(\F_X) \simeq \pi(\Core(\C_X)).$
At the end of this subsection we
define a natural choice of basepoint for each
connected component of $G(X)$ and also prove that
$G(X)$ has at most countably many components.
In Sec.\ (\ref{SecP_i}) below
we show that we may naturally define higher 
Topological T-duality functors using the functor $G.$

For any category $\A$ there is a functor 
$\Core: \A \to \Grpd$
which sends any category to the category with the same
objects but with only iso arrows between them (recall that
a category with only iso arrows between its objects is
termed a groupoid --- see Ref.\ \cite{BRichter}). 
This category is termed the `core groupoid' of the 
category $\A$ (see Ref.\ \cite{nLab-CoreGpd}). 

We had already defined the functor $\M$ in
SubSec.\ (\ref{SSecNMFE}) above. 
Recall that $\M: \mathbf{Top}^{op} \to \mathtt{Cat}$
is a functor which takes values in the category of
all small categories. We may compose this with
the classifying space functor $B$ of Ref.\ \cite{BRichter}
$B: \mathtt{Cat} \to \mathbf{Top}$ to obtain a new functor
denoted $G \simeq B \circ \M.$ 

It is clear that for any topological space $X,$
the value $G(X)$ is actually the classifying space 
associated to the small category $\pi (\Core(\C_X)) 
\simeq \pi(\F_X)$
that is, for a given topological space $X,$ we have that
$B(\pi \circ\Core(\C_X)) \simeq B(\pi (\F_X))$ 
(where the nerve or simplicial realization of a category
and the classifying space of a category are defined in 
Sec.\ (2) of Ref.\ \cite{Segal1}, 
see also Ref.\ \cite{BRichter}).

Given any map $f:Y \to X$ the corresponding 
morphism $G(f)$ in $\mathbf{Top}$ is the map 
$G(f) \simeq B\M(f):B\M(X) \to B\M(Y)$ where 
$B\M(f)$ is the map on classifying spaces induced
by the map $\M(f): \M(X) \to \M(Y)$ 
associated to the functor $\M$ above. This map
exists since $\M(X)$ is always
a small category for any topological space $X$
(see Ref.\ \cite{BRichter} for details
of the classifying space of a small category).
This map is simplicial since, by definition,
$B\M(f)$ is $|\Simp \circ \Core \circ \G(f)|$ 
where $|.|$ is the geometric realization functor of
Refs.\ \cite{Segal1, BRichter}, $\Simp$ is the simpicial 
realization functor of Refs.\ \cite{Segal1, BRichter}
and hence the composite map 
$\Simp \circ \Core \circ \G(f)$ is simplicial by 
definition.

It is also clear from the definition 
of the nerve that $B\Core(\C_X) \simeq B \F_X$ maps to 
$B \C_X$ since $\C_X$ has arrows in it which are not 
iso-arrows and hence $B\C_X$ has more relations and 
hence more cells than $B \Core(\C_X).$

For any groupoid $\G$ with set of objects $Y$ we declare
$x,y \in Y$ to be equivalent if there is an arrow in $\G$
from $x$ to $y.$ Due to the definition of a groupoid, 
this is an equivalence relation on the set $Y,$ and the
set of equivalence classes are termed the set of 
connnected components of the groupoid denoted $\pi_0(\G)$
(see Ref.\ \cite{nLab-Grpd}, Def.\ (2.15) or 
Ref.\ \cite{BRichter}). 

Using the construction of the groupoid $\F_X$ above, 
we had argued at the end of 
SubSubSec.\ (\ref{SSecMorPXY}) that 
there might be pullback squares in $\F_X$ in 
which correspond to maps between the underlying
bundles which change the characteristic class of the
bundles. If there are such maps, then 
the set of connected components 
of the classifying space of the
groupoid $\pi(\F_X)$ i.e., the set $\pi_0(B\M(X))$ 
(see Ref.\ \cite{nLab-Grpd}, Def.\ (2.15)) 
might have $0-$simplices corresponding to 
pairs of differing characteristic class.

For every $CW-$complex $X,$ we now prove that 
$\pi_0(B\M(X))$ is at most countable by using
the space $G(X)$ defined above. To do this we construct
a subset $A(X) \subset G(X)$ using the observation in the
previous paragraph and estimate the cardinality of $A(X).$

For any $X$ define the basepoint $p_X$ of $G(X)$ as
the image under $\pi$ of the trivial pair 
$(X \times S^1 \to X, 0)$ in $\C$
corresponding to the trivial bundle over $X$ with trivial 
$H-$flux. In addition for any fixed $CW-$complex $X$ 
and for each element $\alpha \in H^2(X,\KZ)$ we pick
a fixed principal circle bundle $E_{\alpha} \to X,$ a 
$H-$flux $[H_{\alpha}] \in H^3(E_{\alpha}, \KZ)$ and an 
associated pair $(E_{\alpha}, [H_{\alpha}]).$ 

We consider the orbit of a pair $(E_{\alpha}, [H_{\alpha}])$
above under all possible isomorphisms of pairs. 
This is the same as considering
the orbit of a fixed pair under the action of all morphisms of
pairs induced by composing pullback maps of pairs along an arbitrary
homeomorphism $h_{\alpha}: X \to X$ and all Bunke-Schick isomorphisms
of pairs over $X.$ (These were defined in Def.\ (\ref{DefCCat1}) above --- 
also see Lem.\ (\ref{LemPairHom}) and Lem.\ (\ref{LemHotIso}).) 

All arrows in $\F_X$ are isomorphisms of pairs 
over $X$ covering homeomorphisms $f:X \to X,$ and 
every isomorphism of pairs
may be uniquely factored as a composition
of the two types of morphism of pairs as
discussed after Def.\ \ref{DefCCat2}.
Hence the above orbit exhausts one
connected component of the groupoid $\F_X$ 
(in the sense of Ref.\ \cite{nLab-Grpd}, 
Def.\ (2.15)). Thus, it exhausts one connected component
of the groupoid $\pi(\F_X)$ (recall this is the
subcategory of $\HC$ defined in Thm.\ (\ref{ThmPiFXDef}) above).

From the above argument there could be morphisms in 
$\F_X$ which could join $(E_{\alpha}, [H_{\alpha}])$ 
and $(E_{\beta}, [H_{\beta}])$
with $\alpha \neq \beta.$ As a result the orbits in 
$\F_X$ might have pairs whose underlying bundles
are of {\em different} characteristic class. Hence, the
same would be true of orbits in $\pi(\F_X)$ and hence
in $G(X).$

We define a set $A(X) \subseteq G(X)$ by picking one element
from each orbit above. 
Clearly $|A(X)| \simeq |\pi_0(B\M(X))| \simeq |\pi_0(G(X))|.$
We now show that the set $\pi_0(G(X)) \simeq \pi_0(B\M(X))$ is at 
most countable using the definition of $A(X).$
\begin{lemma}
For any $CW-$complex $X,$
$A(X)$ is at most countable and $G(X)$ has at most
countably many components. Hence the set
$\pi_0(\pi(\F_X)) \simeq \pi_0(G(X))$ is
at most countable.
\label{LemPi0AG}
\end{lemma}
\begin{proof}
Note that there are only at most countably many pairs over $X$
since the total number of pairs is bounded above by
$|H^2(X,\KZ)| \times \Pi_{\alpha \in A(X)} |H^3(E_{\alpha},\KZ)|$
and $X$ is a finite-dimensional $CW-$complex and so has at most
a countable number of cells. 

Hence by the discussion before this theorem
$A(X)$ is a discrete countable subset of $G(X)$ with
exactly one element per component of $G(X).$ 
Also by the above there is exactly one element
of $A(X)$ in each connected component of $G(X)$ and hence 
$\pi_0(G(X))$ is at most countable as well.

By construction $A(X)$ is in bijection with $\pi_0(\F_X))$ and the last
statement in the lemma follows from the above.
\end{proof}
We show in Lem. \ (\ref{LemPP0phi}) below that there
is a natural surjective map
$\phi: P(X) \twoheadrightarrow \pi_0(G(X)).$
Clearly the above construction of the space $A(X)$ 
yields a natural section map
$a:\pi_0(G(X)) \hookrightarrow P(X)$ of the
map $\phi$ such that the image of $a$ is $A(X).$
We will use this map in Sec.\ (\ref{SecP_i}) below.

\subsection{Properties of the Homotopy groups of $G(X)$ \label{SSecHotG}}

By Lem.\ (\ref{LemPi0AG}) $G(X)$ is a countable union of 
connected components which we denote
$Z_{\alpha}$ one for each element of
$A(X)$ and we write  
\begin{gather}
G(X) \simeq \coprod_{\alpha \in A(X)} Z_{\alpha}. \label{EqGXZAlph}\\
\end{gather}

Note that the spaces $A(X), Z_{\alpha}, G(X)$ are geometric realizations of
simplicial spaces and hence always $CW-$complexes.

\begin{lemma}
\label{LemHotAG}
With the definitions above we have for any $k > 0$
$$\pi_k(G(X), p_X) \simeq \pi_k(G(X),A(X))$$ 
and 
$$\pi_k(G(X)) \simeq \pi_k(G(X), p_X)$$
and
$\pi_0(G(X)) \simeq \pi_0(A(X))$
\end{lemma}
\begin{proof}
Note that we have $\{ p_X \} \subseteq A(X) \subseteq G(X).$ 
Since $G(X)$ is a simplicial set and $\{ p_X \} , A(X)$ are 
simplicial subsets of $G(X),$ the inclusions above are cofibrations 
(since they are automatically inclusions of $CW-$complexes) 
and we have that $\pi_q(G(X), p_X) \simeq \pi_q(G(X), A(X)).$ 

The result follows from the long exact sequence of the pair $(G(X),A(X))$
\begin{gather*}
\pi_3(G(X),A(X)) \to \pi_2(A(X)) \to \pi_2(G(X)) \to \pi_2(G(X),A(X)) \to  \\ 
\ldots \to \pi_1(A(X)) \to \pi_1(G(X)) \to \pi_1(G(X), A(X)) \to \pi_0(A(X))  \to \pi_0(G(X)) 
\end{gather*}

Now, since $A(X)$ is discrete, $\pi_k(A(X)) \simeq 0, k\geq 1$ and
$\pi_0(A(X)) \simeq \pi_0(G(X))$ by definition of $A(X)$ so
$\pi_k(G(X)) \simeq \pi_k(G(X),A(X))$ for $k \geq 2$ and hence
$\pi_k(G(X)) \simeq \pi_k(G(X),p_X)$ for $k \geq 2.$

Also, $|\pi_0(G(X))|$ is at most countable by Lem.\ (\ref{LemPi0AG}) above
and the induced map on $\pi_0$ by the inclusion
$A(X) \hookrightarrow G(X)$ should be an isomorphism on 
$\pi_0$ just by the definition of $A(X).$

Also, we have the following short exact sequence
$$
0 \to \pi_1(G(X)) \to \pi_1(G(X),A(X)) \to \pi_0(A(X)) \simeq \pi_0(G(X))
$$

Hence $\pi_1(G(X)) \simeq \pi_1(G(X),A(X))$ as well so 
$\pi_1(G(X),A(X)) \simeq \pi_1(G(X), p_X)$
since $\pi_0(G(X)) \simeq \pi_0(A(X)).$
\end{proof}

In SubSec.\ (\ref{SSecNMFE}) we had naturally
constructed the functor 
$\M:\mathbf{Top}^{op} \to \mathbf{Gpd}$ 
from the functor $E: \HC \to \mathbf{Top}$ which
assigned to any topological space $X$
the small groupoid $\M(X) \simeq \pi(\F_X)$ 
(see Remark (2.4) and Def.\ (2.6)
of Ref.\ \cite{nLab-Grpd}).
We had then constructed the functor
$G: \mathbf{Top}^{op} \to \mathbf{Geom} \circ \mathbf{Simp} \circ \M$
from $\M$ which assigned to any topological space 
$X$ the geometric realization $G(X)$ of the nerve of 
the groupoid $\M(X)$ (this is also termed the 
classifying space of the groupoid $\M(X)$
and denoted $B\M(X) \simeq B \pi(\F_X)$, see
Refs.\ \cite{Segal1, BRichter}).

In Sec.\ (\ref{SecP_i}) we show that for any 
topological space $X$, the space $G(X)$ naturally encodes 
the set of values of the Bunke-Schick's functor $P(X):$
In that section we show that
the values of the Bunke-Schick functor on a 
space $X$ that is, the set of isomorphism class of pairs
over $X$ which is the set $P(X),$ may be identified with
the set of objects of the groupoid $\M(X)$ namely
$P(X) = (\M(X))_0.$ Thus the simplex $G(X)$ constructed
above has as vertices ($0-$simplices) the elements of the set $P(X).$

Note that morphisms in $\M(X)$ may join
pairs with underlying circle bundles of different 
Chern characters (see discussion at the end of
SubSubSec.\ (\ref{SSecMorPXXf})).
Hence they may join pairs
in different Bunke-Schick isomorphism classes.
Thus a given connected component of $G(X)$ may
contain bundles with different characteristic classes.
As a result $\pi_0(G(X))$ is not bijective with
$P(X)$ but as we show in Sec.\ (\ref{SecP_i})
we have a surjective map 
$P(X) \twoheadrightarrow \pi_0(G(X)).$

Since we wish to define higher Topological 
T-duality functors, it is interesting to ask whether the
groupoid $\M(X)$ or the space $G(X)$ encode any other
information apart from the set $P(X).$
We had noted above that the set 
$P(X)$ may be naturally identified
with the object set $(\M(X))_0$ of the small groupoid
$\M(X).$ It is natural then to examine the arrow space of 
$\M(X)$ or $(\M(X))_1$ which corresponds to the 
morphisms in $\pi(\F_X)$ or alternatively, 
the $1-$simplices of $G(X).$
We define higher functors $P_i(X)$ as
$P_i(X) = \pi_i(G(X)), i >0.$ We explain this 
construction in more detail in that section. 
In the same Section
we will also show that $P_i = 0, i \geq 2.$

Note that for principal bundles with circle fiber and any 
$CW-$complex $X$ above, the space $G(X)$ possesses a $\KZ_2-$action
since the isomorphism class of a pair is mapped to the isomorphism
class of the dual pair under Topological T-duality and this 
extends to a symmetry of the category of pairs over 
$X$ (see Ref.\ \cite{BunkeS1} and Ref.\ \cite{BunkeS2}).
This leads to a $\KZ_2-$action on $\pi_i(G(X))$ which, for $i=0$ is the
Bunke-Schick Topological T-duality mapping 
$T: P_0(X) \to P_0(X).$ For $i \geq 1,$ this action by $\KZ_2$ on $G(X)$
leads to a T-duality map $T:P_1(X) \to P_1(X)$ . 

It is interesting to connect the above argument to the $C^{\ast}-$algebraic 
formalism of Topological T-duality of Mathai and Rosenberg in 
Ref.\ \cite{MR1}.  By Ref.\ \cite{Pan2}, Cor.\ (2.1) (1), 
there is a natural isomorphism from the set of all continuous-trace $C^{\ast}-$dynamical systems with spectrum a principal circle bundle 
$p:E_p \to X$ with a $\KR-$action on the $C^{\ast}-$bundle
lifting the circle action on $E_p$ with the cohomology group 
$H^3(E_p,\KZ).$ Note that this gives us a natural pair from each
such continuous-trace $C^{\ast}-$dynamical system.

Thus, the space of pairs over $X$ is naturally equivalent to the 
category of continuous-trace $C^{\ast}-$dynamical systems with
spectrum a principal circle bundle $E_p \to X$ with a $\KR-$action 
on each $C^{\ast}-$dynamical system lifting the circle action on its 
spectrum.

Hence, the argument in this section applied to the category of 
$C^{\ast}-$dynamical systems over $X$ with spectrum a 
principal circle bundle $E_p \to X$ gives a simplicial space 
$G(X)$ just as well as the category of pairs does. 
From now on, we will only study to the category of pairs over $X$
since it is clear that the same space $G(X)$ will be obtained by us if we
consider $C^{\ast}-$dynamical systems for backgrounds which are principal
circle bundles with $H-$flux over $X.$

However, if the argument in this section could be extended to 
backgrounds which are principal $\KT^n-$bundles over $X$ with 
$H-$flux, the above construction would have to be modified, as 
the Topological T-dual could be a noncommutative space.

\subsection{Justification from SFT \label{SSecSFT-TTD}}

In this subsection we argue that the category
$\C$ in Sec.\ (\ref{SecIntro}) may be constructed naturally 
from a projection of closed bosonic SFT onto its massless 
modes.on backgrounds which are principal circle bundles 

String Field Theory has been used to study Topological
T-duality. In Ref.\ \cite{BMRS}, Sec.\ (3.1) 
Brodzki and coworkers described an application of bivariant $K-$theory
to study $D-$branes in the $C^{\ast}-$algebraic formalism of
Topological T-duality using arguments from Open String Boundary
Conformal Field Theory. As is well known, Open String
Boundary Conformal Field Theory naturally appears in
SFT descriptions of string theory
backgrounds (see Ref.\ \cite{Horowitz} for
a good introduction to SFT using
bosonic open SFT and Ref.\ \cite{KZw}
for an introduction to closed SFT on
$n-$dimensional Torus backgrounds, also see 
Ref.\ \cite{D-Hull} for SFT on backgrounds
which are principal torus bundles).
Also, as discussed in Sec.\ (\ref{SecIntro})
SFT may be described by an $L_{\infty}-$algebra.
Fiorenza, Sati and Schreiber in Ref.\ \cite{FSS} studied
the effect of Topological T-duality for principal 
circle bundles on natural $L_{\infty}-$algebras associated
to those bundles.

In this subsection we cite the 
String Theory literature to prove that closed bosonic SFT 
may be reduced to an effective SFT-like field
theory by projecting onto its massless modes 
(the $B-$field and graviton) on {\em every} SFT background
\cite{SFTDiffeo}. In addition there is a reduced description
of closed bosonic SFT for {\em $n-$torus} vacua in
the String Theory literature which projects SFT on 
$n-$tori onto its massless modes namely
the $B-$field and the graviton \cite{KZw}. In addition
there is a treatment of SFT for vacua which
principal $n-$torus bundle 
backgrounds \cite{D-Hull}.

We also argue that the characteristic class of the
$B-$field which yields the $H-$flux and the characteristic
class of the principal circle bundle which may be obtained
from the graviton are natural variables in SFT to parametrize such
vacua.  We then use this to justify the construction of the
category $\C$ in Sec.\ (\ref{SecIntro}).

We conclude the section with some remarks on other
works which seem connected to the material in this
subsection.

In Ref.\ \cite{SFTDiffeo} Sec.\ (2.2), Mazel et al show
that closed bosonic String Field
Theory in a `weak-field' approximation where the string
field does not vary too much on the scale of an individual
string may be projected onto an effective description for the
massless modes of the SFT (the graviton and 
the $H-$flux). Thus, such
backgrounds are well described by an effective
SFT-like description defined by keeping only 
the massless modes, i.e. the graviton
and the $B-$field of the full SFT. 

SFT on $n-$torus backgrounds was studied in detail 
by Kugo and Zwiebach in Ref.\ (\cite{KZw}).
The reduction of superstring SFT to an effective SFT-liked description
of low energy modes was proposed in generality for 
in the work of Sen and coworkers in Ref.\ \cite{Sen1}.
Arvanitakis and Hull  calculated 
the corresponding reduced description of closed bosonic SFT 
for $n-$torus backgrounds in Ref.\ (\cite{AHull, AHull2}).
Also, Dabholkar and Hull Ref.\ (\cite{D-Hull}) derived a 
description of closed bosonic SFT for backgrounds
which are $n-$torus bundles over a base spacetime.

Arvanitakis and Hull in Refs.\ (\cite{AHull, AHull2}) demonstrate
that there is a consistent projection of SFT on 
$n-$torus backgrounds onto its massless modes,
the $B-$field and the graviton. Thus, restricting 
attention to the effective SFT involving the
$B-$field and the graviton fields
on $n-$torus backgrounds will not introduce
any inconsistency in the argument.

Kugo and Zwiebach in Ref.\ \cite{KZw}
demonstrate that a SFT background 
which is a $n-$torus is described by the 
quantity $g+B$ where $g_{ij}$ is the metric of the
background and $B_{ij}$ is the $B-$field on it 
(see Ref.\ \cite{KZw} Sec.\ (2.1) for example). 
The fields $g$ and $B$ are the massless modes of 
SFT on the background $E.$ 

For $n-$torus backgrounds Ref.\ \cite{KZw} 
Sec.\ (2.1) show that all the 
quantities of interest in SFT on the
background $E$ such as the 
position and momentum operator and the String Field $\Psi$ 
may be expanded in terms of oscillator expansions in 
creation and annihilation operators $\alpha^j_n(g + B)$
which depend on the massless fields of the background $E.$ 

Changing the background to another background $E'$ 
will give rise to a different set of oscillator expansions 
which can be related to the oscillator expansion in the
original background $E$ as in Ref.\ \cite{KZw} 
Eq.\ (2.19-2.22).

A study of SFT on backgrounds which are
principal $n-$torus bundles over a base spacetime is
in the work of Dabholkar and Hull (see Sec.\ (9) of Ref.\ \cite{D-Hull})
where a closed bosonic SFT description of such backgrounds
is given.

Thus it seems reasonable to conclude that
closed bosonic SFT on $n-$torus bundle backgrounds may
be projected in a consistent fashion onto a reduced
field theory which keeps only the massless modes in the theory.
We now argue that we may construct the data of a pair $([p], [H])$
from these massless modes. This physically justifies the 
construction of the category $\C$ as an approximation of
certain backgrounds of closed bosonic SFT.

We note that each of the two data making a Topological T-duality pair in the previous paragraph may be naturally obtained from massless modes of 
the SFT around the background $E$:
The class $[p]$ is the exterior derivative of the connection $1-$form $\Theta$ on $\E$ which is a $\KR-$valued $1-$form
defined (up to gauge transformations) by the metric $g_{ij}.$ 
More precisely, this exterior derivative is (see Ref.\ \cite{Minasian} Sec.\ (2.1))
$p^{\ast}F=d\Theta$ on $E$ and is invariant under the pullback induced by the circle action on $E.$ Where
$F \in \Omega^2_{\KZ}(M)$ is an integral two-form on $M,$ the curvature of the bundle $p:E \to M.$
 
Since $g_{ij}$ is one of the massless modes of the
SFT on the background $E,$ we may view $p^{\ast}F=d\Theta$ and hence $[p]$ as obtained from the SFT around the background $E.$  
Also, the class $H$ is the curvature of the $B-$field on $E,$ and hence is also
obtained from one of the degree zero modes of the String Field on the background $E.$ 

Hence the data $([p], [H])$ associated to a pair $(p:E \to X, [H])$ in
Sec.\ (\ref{SecIntro}) may be obtained from the massless
modes of the reduced SFT on the background $E$ and this
justifies the construction of the category $\C$ in Sec.\ (\ref{SecIntro}).

We now discuss the physical interpretation of the space $G(X)$ which
we will use in SubSec.\ (\ref{SSecGXCX}) and in
the following Sections to construct higher Topological T-duality
functors. Recall that we had
defined $G(X)$ as the simplicial classifying space of the
groupoid $\M(X)$ which we had constructed naturally from
the forgetful functor $F:\C \to \mathbf{Top}$ in SubSec.\ (\ref{SSecPCLift})
and following subsections.

It is clear that we may view the space $G(X)$ defined 
in SubSec.\ (\ref{SSecGXCX}) as the moduli space of all possible gauge
equivalence classes of SFT vacua
which are principal circle bundles over $X$ with $H-$flux
and the morphisms between pairs in $G(X)$ should then be a
special type of `Change of Background'
maps.

It might be argued that the interpretation of the space $G(X)$
given in the previous paragraph is {\em unphysical}, since $G(X)$ corresponds
to the set of toroidally compactified vacua over $X$ with 
`change of background' maps between them however
there are {\em many} other backgrounds on which the SFT is defined.

This is {\em false} since there is an
approximation of SFT describing only the class of toroidally 
compactified vacua over $X:$
It is known that String Field Theories are described by an
$L_{\infty}-$algebra (see Refs.\ \cite{Zwiebach}, \cite{AHull}).
The work of Arvanitakis and coworkers in Ref.\ \cite{AHull, AHull2} 
and Dabholkar and Hull in Ref.\ \cite{D-Hull} 
shows that for $n-$torus backgrounds
closed bosonic SFT may be projected onto another 
$L_{\infty}-$algebra which consists of variables 
describing the low energy behaviour of the original SFT on toroidally
compactified backgrounds only and the resulting 
projected theory is a  consistent stringy quantum field theory.
Thus for $X$ the $n-$torus 
the space $G(X)$ may be looked on as the moduli space of possible
backgrounds of this reduction of SFT. 

For backgrounds consisting of principal $n-$torus bundles over a 
base spacetime a description in terms of closed bosonic SFT is
derived in the work of Dabholkar and Hull in Sec.\ (9) of
Ref.\ (\cite{D-Hull}). It is clear that the above argument
will generalize to principal $n-$torus backgrounds.

It is well known that SFT may be defined in terms of an
$L_{\infty}-$algebra. Is it possible to construct the space 
$G(X)$ from  $L_{\infty}-$algebra? 
There is a natural Kan complex which
can be constructed from an $L_{\infty}-$algebra, the Maurer-Cartan simplicial 
space of the $L_{\infty}-$algebra (see Ref.\ \cite{Getzler}).
It would be interesting to see if $G(X)$ is related to the Maurer-Cartan space of the associated SFT $L_{\infty}-$algebra perhaps as a fibration.

It is interesting to note that Fiorenza, Sati and Schreiber in
Ref.\ \cite{FSS} demonstrated a form of Topological T-duality by associating an 
$L_{\infty}-$algebra to a principal circle bundle $E$ with a circle action and
showing that Topological T-duality
for $E \to X$ gives a Fourier-Mukai transform of the associated $L_{\infty}-$algebra. However, its not clear that the $L_{\infty}-$algebra they
associate to $E \to X$ is the same as the $L_{\infty}-$algebra induced by 
SFT on $E.$

We end this subsection by mentioning
the work of Ref.\ \cite{BeniSch}, Benini and Schenkel study 
quantum field theories (in particular gauge theories) 
on categories fibered in groupoids over a base category
of spacetimes with extra structure using the theory of 
Kan extensions of functors. In this paper we study 
SFT using categories fibered in 
groupoids over $\mathbf{Top}$ and the theory of 
categorical equivalence of functors.
The theory of Kan extensions of functors used in 
Ref.\ \cite{BeniSch} is an example of categorical  
equivalence, in particular the Grothendieck Construction for 
categories fibered in Sets (see Thm.\ (\ref{ThmHCTEqv})
above) is an example of a Kan extension. However, since 
we are using the Grothendieck
Construction for Street Fibrations in this paper, this is
more general than the construction for categories
fibered in groupoids. Also, since SFT is a 
third-quantized theory, and (see the low energy effective
approximation to SFT
described in Sec.\ (2.2) of Ref.\ \cite{SFTDiffeo})
its restriction to the backgrounds considered in this 
paper is still a third-quantized theory involving ghost
creation and destruction operators, it is not possible to 
compare the results in this paper directly
with the result of Ref.\ \cite{BeniSch}. However it is
interesting that the two works use somewhat similar ideas.

\section{Calculation of $P_i(X)$ \label{SecP_i}}
\subsubsection{Introduction}
In Thm.\ (\ref{ThmMGFMM}) above, we had defined a
functor $\M: \mathbf{Top}^{op} \to \mathbf{Grpd}$ 
which assigned to any topological space $X$ a small
groupoid $\M(X)$ which was
the small category $\pi(\F_X)$ above
(see SubSec.\ (\ref{SSecBXFX}) above
for a definition of $\pi(\F_X)$). 
As discussed in Subsec.\ (\ref{SSecGXCX})
$\pi(\F_X)$ is a small 
groupoid since it is $\pi \circ \Core(\C_X)$ and 
$\C_X$ is a small category. 

In this section we use this functor $\M$ to construct two
topological invariants from any topological space $X,$
the first invariant will be the $\mathbf{Set}-$valued
functor $P$ defined in Ref.\ \cite{BunkeS1}
restricted to $\mathbf{CW}^{op}$
the opposite of the category of $CW-$complexes and
cellular maps which we denote $P_0(X)$
and the second will be a 
groupoid-valued functor on $\mathbf{CW}^{op}$ 
which we denote $P_1(X).$ 

More precisely, we construct
the two invariants $P_0(X), P_1(X)$
from two invariants naturally associated to 
the groupoid $\M(X),$ namely the set of connected
components of the groupoid $\M(X)$ denoted 
$\pi_0(\M(X))$ and the automorphism group of 
any object in $(\M(X))_0$ denoted 
$\Aut_x(\M), x \in (\M(X))_0.$ 
We can only define these invariants for
the category of $CW-$complexes and 
cellular maps due to Lemma.\ (\ref{LemPi0AG}) above.
We show below that we may calculate
$\Aut_x(\M), x \in (\M(X))_0$ 
from $\Pi_1(B\M(X))$ the fundamental 
groupoid of the space $B \M(X)$ (see Ref.\ \cite{BRichter}).

We will also show below that these four invariants of
$X$ fit into a $2-$commutative diagram of groupoids
\begin{equation}
\begin{CD}
P_1(X) \simeq \Pi_{x \in A(X)} \Aut_x(\M(X)) @>>> \Pi_1(G(X)) \simeq \M(X) \\
@VVV   @VV{p}V \\
P_0(X) \simeq ((\M)_0)(X)) @>>{\phi}> \pi_0(G(X)) \simeq \pi_0(\M(X))  \label{GrpdP0P1} \\
\end{CD}
\end{equation}
and this defines $P_1(X)$ as a functorial pullback of $\Pi_1(G(X))$
over the natural map $\phi:P_0(X) \to \pi_0(G(X)).$

\subsubsection{Construction of $P_0(X)$}
We now derive each of the four functors in the
diagram Eq.\ (\ref{GrpdP0P1}) above naturally
from the groupoid-valued functor $\M.$
We had mentioned above that a groupoid was a 
category with every arrow invertible. A morphism of
groupoids is a functor between two such categories.
A homotopy of two such morphisms of groupoids is
a natural transformation between two such functors 
(see Def.\ (2.2) of Ref.\ \cite{nLab-Grpd}). Using
this it is easy to define homotopy equivalence and
weak homotopy equivalence of
groupoids (see Ref.\ \cite{nLab-Grpd} Defs.\ (2.13-2.17)). 
By Ref.\ \cite{nLab-Grpd} Def.\ (2.16 - 2.17)
there are two invariants which characterize
groupoids up to weak homotopy equivalence of groupoids, 
one is $\pi_0(\G)$ which we discussed at the end of
SubSec.\ (\ref{SSecGXCX}) and another is 
$\Aut_x(\G)$ for $x$ an element of $\G_0.$

Hence given any topological space $X$ we may associate
the set $\pi_0(\M(X))$ to the groupoid $\M(X).$ 
By definition this is the set of connected components of the
groupoid $\M(X)$ which we discussed at the end of
SubSec.\ (\ref{SSecGXCX}) above --- see the paragraphs
before Lem.\ (\ref{LemPi0AG}).
In addition by Ref.\ \cite{BRichter} there is a natural 
bijection between the connected components of the 
category $\M(X)$ and the connected components of 
the space $G(X) \simeq B \M(X)$ so that 
$\pi_0(G(X)) \simeq \pi_0(\M(X)).$

To any small groupoid $\G$ we may associate its set of
objects $\G_0$ and this gives us a forgetful functor 
$\Phi:\mathbf{Gpd} \to \mathbf{Set}$
which acts on objects in $\mathbf{Gpd}$ as 
$\G \mapsto \G_0$ and acts in the obvious way on
morphisms of groupoids. Composing the functor $\Phi$
with the groupoid-valued functor 
$\M$ above we have a forgetful functor 
$\Phi \circ \M: \mathbf{Top}^{op} \to \mathbf{Set}$ which
sends any topological space $X$ to the {\em set} of 
isomorphism classes of pairs over $X.$  Due to this the
functor $\Phi \circ \M$ 
may be identified with the Bunke-Schick
functor $P$ of Ref.\ \cite{BunkeS1} since $\Phi \circ \M(X)
\simeq (\M_0)(X) \simeq P(X).$ It is clear that the groupoid 
$\M(X)$ encodes Bunke-Schick's functor $P(X)$ as its set 
of objects and we may view it as a `categorification' of
$P(X).$ We had seen this phenomenon in the previous
section when we had studied the space $G(X)$ obtained
from $\M(X)$ by the nerve construction. 

Instead of using $\M(X)$ it is natural to
use $G(X) \simeq B\M(X)$ in the above argument: That is,
given a topological space $X,$ 
we construct the assignment $X  \mapsto \pi_0(G(X))$
and a $\mathbf{Set}-$valued functor $\pi_0(G(X))$ on 
$\mathbf{Top}^{op}.$ By Ref.\ \cite{BRichter}, since
$\M(X)$ is a small category, we have  that
$\pi_0(G(X)) \simeq \pi_0(B\M(X)) \simeq \pi_0(\M(X)),$ the set of
connected components of the small category $\M(X).$
We also show in Lem.\ (\ref{LemPP0phi})
below there is a natural surjective map
$\phi: P_0(X) \twoheadrightarrow \pi_0(G(X)).$ 

We may view the functor $P_0$ as valued in
groupoids since it is $\mathbf{Set}-$valued
and a set may be naturally viewed as a trivial groupoid
with only identity morphisms for each object.
Similarly, we may view $\pi_0(G(X))$ as groupoid
valued since $\pi_0$ is always a set.
Hence the map $\phi$ in Eq.\ (\ref{GrpdP0P1}) 
can be viewed as a groupoid morphism between
groupoids $\phi: P_0(X) \to \pi_0(\M(X)).$

\subsubsection{Definition of $P_1(X)$}
We use the fundamental groupoid  
$\Pi_1(G(X))$of the moduli space $G(X)$ 
(see Ref.\ \cite{May}) to define a higher Topological
T-duality functor. By Ref.\ \cite{BRichter}, 
$\Pi_1(G(X)) \simeq \Pi_1(B\M(X)) \simeq \M(X)$
since $\M(X)$ is always a groupoid. There is 
a natural map $p: \Pi_1(G(X)) \to \pi_0(G(X))$
which sends any object or arrow in $\Pi_1(G(X))$ to
the connected component it is in, in particular, the
map sends any loop in $G(X)$ to its connected
component. 

We now have a natural definition of three of the four 
invariants of $X$ in Eq.\ (\ref{GrpdP0P1}) above.
We now define $P_1(X)$ as a $2-$pullback using these
four invariants. We will show that $P_k(X), k \geq 2$ are trivial
after Cor.\ (\ref{CorPi1GX}) below.

First, we need some preliminary results on the structure
of the simplicial space associated to the groupoid
$\M(X)$ for any topological space $X.$


It is easy to show that the simplicial space $\F(X)$
which is the simplicial realization of $\M(X)$ i.e., 
$\F(X) \simeq \Simp(\M(X)) \simeq
\Simp \circ \pi \circ \Core(\G(X))$
is always a Kan complex:
\begin{lemma}
The simplical space $\F(X)$ is always a Kan complex.
\label{LemFKan}
\end{lemma}
\begin{proof}
%
 
Note that the small category
$\F(X) \simeq \Simp(\M(X))$ 
(see SubSec.\ (\ref{SSecNMFE}) especially  
Thm.\ (\ref{ThmMGFMM}))
is the nerve of the groupoid
$\M(X) \simeq \pi(\F_X)$ and hence the
nerve of a groupoid by construction---see 
Thm.\ (\ref{ThmMGFMM}) above.
Hence $\F(X)$ must be a Kan complex by Props.\ 
(1.3.3.1) and (1.3.5.2) of Ref.\ (\cite{Kerodon-NKan}). 
\end{proof}
 
We now show that there is a surjective map
$\phi:P(X) \twoheadrightarrow \pi_0(G(X))$
using the Kan property of the simplicial space
$\F(X).$
\begin{lemma}
Let $P(X), G(X)$ be as defined above in this subsection.
\label{LemPP0phi}
\leavevmode
\begin{enumerate}
\item There is a surjective map 
$\phi:P(X) \twoheadrightarrow \pi_0(G(X)).$ 
\item The set $P_0(X) \simeq \pi_0(G(X))$ depends only
on the homotopy type of $X.$
\end{enumerate}
\end{lemma}
\begin{proof}
\leavevmode
\begin{enumerate}
\item We know that the simplicial space
$\F(X)$ is isomorphic to 
$\Simp(\M(X)) \simeq \Simp(\pi(\F_X))$ 
where $\F_X$ is the subcategory of $\C$ defined in 
SubSec.\ (\ref{SSecBXFX}) above.
By Lem.\ (\ref{LemFKan}) above $\F(X)$ is Kan,
hence the homotopy groups 
of the geometric realization $G(X)$ of $\F(X)$ are
isomorphic to the simplicial homotopy groups of $\F(X).$

From the definition of simplicial homotopy
$\pi_0(G(X))$ is isomorphic the the set of equivalence
class of objects of $\F(X)$ with two objects $x, y$ in
$\F(X)$ considered equivalent iff there is an arrow
$\gamma: x \to y$ in $\F(X).$ 
Also, by definition of $\F(X),$ each object in
$\F(X) \simeq \Simp(\M(X)) \simeq \Simp(\pi(\F_X))$ is a
homotopy equivalence class of pairs over a fixed 
topological space $X$ in $\C$ and the set of objects in 
$\F(X)$ is in bijection with the set $P(X).$ 

It is clear that the equivalence relation on 
$\F(X)$ above which defines $\pi_0(G(X))$ gives
a natural map $\phi: P(X) \twoheadrightarrow \pi_0(G(X)).$
\item It can be shown (see Def.\ (4.2.2) and
Lem.\ (4.2.1) of Ref.\ \cite{UMDThesis})
that the set of isomorphism classes of pairs over $X$
depends only on the homotopy type
of $X.$ It follows that $P_0(X)$ and $\pi_0(G(X))$ 
depend only on the homotopy type of $X.$
\end{enumerate}
\end{proof}

As discussed at the beginning of this section,
for any topological space $X$
there should exist a second natural invariant 
of $X$ which we had denoted $P_1(X)$
(this is a different invariant of $X$
from the set $P_0(X) \simeq P(X)$). 
There is an equivalence of
categories between $\M(X)$ and the
subcategory $\pi(\F_X)$ of
$\HC.$ Hence any element of $\HC$ is the
image of the  Bunke-Schick equivalence class
of a pair $e$ in $\C$ under $\pi.$

We require that we should be able to naturally
construct the invariant $P_1(X)$ 
from the functor $\M$ by using
the automorphisms of an arbitrary element 
$\pi(e)$ in $\M(X)$ where $e$ is an 
arbitrary pair over $X$ i.e., an object of
$\C$ such that $F(e) \simeq X$

It is not clear how to define this invariant
functorially in $X$ since for any groupoid $\G$ 
the assignment $\G \mapsto
\Aut_x(\G)$ depends on the element $x$ of $\G.$
We now argue that the correct functorial definition of
this invariant may be obtained from the fundamental
groupoid $\Pi_1(G(X))$ or equivalently
$\Pi_1(B\M(X)).$

We note that we have a well-defined functor 
$p:\Pi_1(B\M(X)) \to \pi_0(B\M(X))$ which
sends any object and arrow in $B\M(X)$ to its 
connected component. We can also assume that 
this functor is groupoid-valued by giving the set 
$\pi_0(B\M(X))$ the structure of a trivial groupoid 
whose objects are all elements of $\pi_0(B\M(X))$ and
whose arrows are only identity maps. It is clear that
with this the function $p$ induces a natural
morphism of groupoids also denoted $p.$

Now we have a diagram of groupoid morphisms
$$
P_0(X) \overset{\phi}{\to} \pi_0(G(X)) \overset{p}
{\leftarrow} \Pi_1(G(X))
$$ 
termed a cospan of groupoid morphisms.
It is natural to define $P_1(X)$ as a
$2-$pullback in the groupoid category constructed
from this cospan. There are many types of
$2-$pullbacks but in this paper we use the simplest type of 
$2-$pullback namely the comma object 
$\phi \downarrow p$ constructed from $\phi$ and $p.$
Here $\phi \downarrow p$ is the comma
object associated to the groupoid morphisms
$p$ and $\phi$ in the groupoid category 
(see nCatLab pages on comma objects
Ref.\ \cite{nLab-Comma} and 
$2-$pullbacks Ref.\ \cite{nLab-2Pbck}). 

Hence we naturally
obtain the following $2-$commutative
diagram
\begin{equation}
\begin{CD}
P_1(X) \simeq (\phi \downarrow p) @>>> 
\Pi_1(G(X)) \simeq \M(X) \\
@VV{\eta}V   @VV{p}V \\
P_0(X) \simeq ((\M)_0)(X)) @>>{\phi}> \pi_0(G(X)) \simeq \pi_0(\M(X))  \label{GrpdP1Pbck} \\
\end{CD}
\end{equation}
where $\eta:P_1(X) \to P_0(X)$ is a connecting map.
We will describe $\eta$ in terms of the algebraic topology
of pairs over $X$ in Thm.\ (\ref{ThmPairMCG}) below.

We now determine $P_1(X)$ in terms of automorphisms
of the elements of $\M(X).$
\begin{theorem}
The $2-$pullback of $p: \Pi_1(G(X)) \to \pi_0(G(X))$ 
along the functor $\phi: P_0(X) \to \pi_0(G(X))$
is a small groupoid whose objects are $P_0(X)$
and whose arrows at any $x \in P_0(X)$ are
$\Aut_x(\M(X)).$
\label{ThmP1Pbck}
\end{theorem}
\begin{proof}
We use the nCatLab pages on comma objects
Ref.\ \cite{nLab-Comma} and $2-$pullbacks Ref.\ \cite{nLab-2Pbck}
in this proof.

From the definition of a comma object, it is clear that the
comma object $(\phi \downarrow p)$ consists of all the loops
in $G(X)$ based at each $0-$simplex in $G(X).$ Since two
$0-$simplices in $G(X)$ in the same connected component
have  bijective sets of loops (by definition of $\Pi_1(G(X))$), 
we obtain the result from the definition of $A(X).$
\end{proof}

The above result shows that to calculate $P_1(X),$ 
we need to calculate $\Aut_x(\M(X)).$
We calculate $\Aut_x(\M(X))$
using a spectral sequence argument in 
Cor.\ (\ref{CorPi1GX}) below. 
We then determine the structure of $P_1(X)$ in
Thm.\ (\ref{ThmPairMCG})
.
Before this we examine some properties of $G(X)$ which will
help us determine the structure of $\Aut_x(\M(X))$ in 
Cor.\ (\ref{CorPi1GX}) below.

We define higher Topological T-duality functors
$P_i(X), i \geq 2$ as a functorial pullback using a higher analogue
of the fundamental groupoid construction
after Cor.\ (\ref{CorPi1GX}) below and show
that $P_i(X) = 0, i \geq 2.$

It is interesting to determine the value of $P_1(X)$ 
a few test spaces $X.$ This is possible in some cases,
for example, using Thm.\ (\ref{ThmPi1H2}) below, 
it is clear that $\pi_1(G(X)) \simeq P_1(X)$ will
be nonzero if there exists a $E_{\alpha} \to X$ such that
$H^2(E_{\alpha},\KZ)$ is
nonzero. As a result $P_1(S^m) \simeq 0, m \geq 3.$

\subsubsection{Calculation of $P_1(X)$}
We now argue that the automorphisms of a given 
pair described in Thm.\ (\ref{ThmPikG})
determine the structure of $G(X)$ and let us calculate
$P_1(X)$ for any $X.$ We first show that a quotient of
$\pi_1(G(X))$ is a direct product of abelian groups.
By the Hurewicz theorem, there is a map 
$$
f^k:\pi_k(G(X)) \simeq \pi_k(G(X),A(X))  \to H_k(G(X),A(X)) \to \tilde{H}_k(G(X)/A(X))
$$

\begin{theorem}
The group $\im(f^1)$ is a direct sum of groups each one isomorphic to $\KZ_m$ or $\KZ$ for each
$\alpha \in A(X).$
\end{theorem}
\begin{proof}
As argued in Eq.\ (\ref{EqGXZAlph}),
the space $G(X)$ is a disjoint union of connected
components $Z_{\alpha}.$ The group $\im(f^1)$ must be
isomorphic to a direct product of abelian groups, since
$G(X)/A(X)$ is a wedge of loops-one for each 
connected component $Z_{\alpha}$
of $G(X).$ Some of these homotopy classes have $2-$ 
and higher cells glued on in the nerve construction. 
By Thm.\ (\ref{ThmPikG}) above 
$\pi_k(G(X)) = 0, k \geq 2.$ Hence
any nontrivial $2-$ or higher-cells glued onto the above 
loops are all nullhomotopic. As a result the only relations
that can exist in $\pi_1(G(X))$ are relations of the form
$a^m = 1$ for some $m >0$ and hence we have 
$$
\im(f^1) \simeq \Sigma_{\alpha} \tilde{H}_1( Z_{\alpha}) \simeq 
\Sigma_{\tiny L \in \mbox{( Loops in } Z_{\alpha} \mbox{) } } G_L.
$$
where each group $G_L$ is either $\KZ_m$ or $\KZ.$
\end{proof}

Thus, a quotient of $P_1(X)$ is a direct
sum of abelian groups. We would like to 
calculate $P_1(X)$ directly. This can be done using
using the homotopy spectral sequence for a tower 
of fibrations as described in the notes by B. Guillou
(see Ref.\ \cite{Guil} Sec.\ (3.1)).

We first filter $G(X)$ by simplicial degree as
$\dots \to G_2(X) \to G_1(X) \to G_0(X)$ so that
$G(X) = \varprojlim G_i(X).$ 

Let $F_i$ be the fiber of the map $p_i: G_{i+1} (X) \to G_i(X)$ so we have
$F_s \stackrel{i_s}{\to} G_s \stackrel{p_{s-1}}{\to} G_{s-1}.$

Then (by Ref.\ \cite{Guil}) there is a `fringed' spectral sequence--the homotopy 
spectral sequence for the tower of fibrations $G_i-$ with $E_1$ term
$$E^{s,t}_1 \simeq \pi_{t-s}(F_s) \Rightarrow \pi_{t-s}(G)$$ 
converging to $\pi_{\ast}(G).$ This spectral sequence-like
construction 
has sets along the $(0,m)$ and $(n,0)$ entries in the
tableaux and groups in the rest of the entries ---hence the
name `fringed'. However, calculations may be done
as for an ordinary spectral sequence (all of whose entries
are groups) except that the fringe maps must now be
functions. For details see Ref.\ \cite{Guil}.

\begin{theorem}
Let $X, G(X), F_i(X)$ be as above then $\pi_1(G(X)) \simeq \pi_1(F_1(X)).$
\label{ThmSSeqPi1FG}
\end{theorem}
\begin{proof}
Consider the $E^1-$page of the homotopy spectral sequence for the tower of fibrations $G_i.$
\\
\begin{sseqpage}[classes={draw=none}, right clip padding=20pt, top clip padding=20pt, left clip padding=30pt, xscale=3.0 ]
\class["\pi_0F_0"](0,0)
\class["\pi_0F_1"](1,1)
\class["\pi_0F_2"](2,2)
\d2(0,0)(1,1)
\d2(1,1)(2,2)
\class["\pi_1F_0"](1,0)
\class["\pi_1F_1"](2,1)
\d2(1,0)(2,1)
\end{sseqpage}

We also have the following identifications in the $E^1-$page above:
\begin{gather*}
\pi_0(F_0) \simeq P_0(X) \\
\pi_i(F_0) \simeq 0, i \geq 1\\
\end{gather*}

Compare this with the $E^{\infty}-$page of the same spectral sequence: \\
\begin{sseqpage}[classes={draw=none}, right clip padding=20pt, top clip padding=20pt, left clip padding=30pt, xscale=3.0 ]
\class["\pi_0G"](0,0)
\class["\pi_0G"](1,1)
\class["\pi_0G"](2,2)
\class["\pi_1G"](1,0)
\class["\pi_1G"](2,1)
\end{sseqpage}

Since we know that
the sequence converges to $\pi_{\ast}(G),$ the terms in the $E^{\infty}-$page are known at
least in low degrees.

It is clear that the term $\pi_1(F_1)$ survives to the $E^{\infty}-$page.

Hence, $\pi_1(G) \simeq \pi_1(F_1)$ and $\pi_i(G) \simeq 0, i \geq 2.$
\end{proof}

We can now calculate $\pi_1(G(X)):$

\begin{corollary}
We have that
$$\pi_1(G(X)) \simeq \pi_1(F_1(X)) \simeq \protect{ \underset{\tiny L \in (\mbox{Loops in } F_1) }{ \Asterisk } } G_L
\simeq \protect{ \underset{ \tiny \beta \in (\mbox{Automorphisms of } (p_0, H_0))}{ \Asterisk }} G_{\beta}$$
where the $G_{\beta}$ are a collection of groups each isomorphic to some $\KZ_m$ or $\KZ$ and so are the
$G_L.$
\label{CorPi1GX}
\end{corollary}
\begin{proof}
By the proof of Lem.\ (\ref{ThmSSeqPi1FG}) above 
$\pi_1(G) \simeq \pi_1(F_1)$ is generated by the simplicial loops in $F_1$ which are the
simplicial loops in $G(X)$ beginning and ending at a pair $(p_0,H_0)$ in $G_0(X).$

By Thm.\ (\ref{ThmPikG}) above and the construction of the groupoid $\F(X)$
in Subsec.\ (\ref{SSecGXCX}) above loops in $\pi_1(G(X))$ correspond exactly
to automorphisms of a pair. By Lem.\ (\ref{LemHotIso}) above homotopic maps 
give rise to isomorphic pairs. Hence the collection of automorphisms of a pair
over $X$ is in bijection with $Homeo(X).$ If we pick a unique representative $f_{\alpha}$
for each nontrivial homotopy class in $Homeo(X)$ as in Sec.\ (\ref{SecIntro}) after
Lem.\ (\ref{LemPairHom} ,then the collection of automorphisms of a pair
over $X$ is given by pulling the pair back along $f_{\alpha}$ and then composing
by an automorphism of the pair covering the identity map $id:X \to X.$ 

Recall that $G(X)$ is a groupoid obtained by the nerve construction from $\Core(\C_X)$ as described
in Subsec.\ (\ref{SSecGXCX}) above and we also have from
Thm.\ (\ref{ThmPikG}) above that $\pi_k(G(X)) = 0, k \geq 2.$ Hence
any nontrivial $2-$ or higher-cells glued onto the above loops are all nullhomotopic.
As a result the only relations that can exist in $\pi_1(G(X))$ are relations of the form
$a^m = 1$ for some $m >0.$ Thus, $\pi_1(G(X))$ can only be a free product 
of abelian groups one for each loop with each group isomorphic to 
$\KZ$ or to $\KZ_m$ for some $m.$
\end{proof}

\subsubsection{Isotopies of $X$ and $P_1(X)$}
We had noted after Thm.\ (\ref{ThmHCTEqv}) 
that Thm.\ (\ref{ThmHCTEqv}) geometrizes
Bunke-Schick's functor $P$ by showing that
it is equivalent to the fibration $E: \HC \to \mathbf{Top}$
using the category of elements construction.

We now prove a result which geometrizes $P_1(X)$
by identifying elements of $P_1(X)$ in terms of the
action of the mapping class group $MCG(X)$ of
$X$ on isomorphism classes of pairs over $X.$

To do this, we first show that elements of $\Aut_x(\M(X))$
correspond to a pullback isomorphism of $x$ covering
a nontrivial self-homemorphism of the base 
$X \simeq F(x)$ of the pair $x.$

The space $G(X)$ has as vertices all pairs over $X.$
In Thm.\ (\ref{ThmGXP1}) below we
first show that every $1-$simplex in $G(X)$ lifts
to a pullback isomorphism of pairs over $X$ which
cover a nontrivial self-homeomorphism $f:X \to X.$
Then we show that every loop in $\pi_1(G(X))$ is homotopy
equivalent to a simplicial loop in the $1-$skeleton of $G(X).$
Then we prove that elements of $P_1(X)$ are
automorphisms of pairs in $A(X)$ which cover a nontrivial 
self-homeomorphism of $X.$

In Thm.\ (\ref{ThmPairMCG}) below we show that
elements of $P_1(X)$ correpond to isotopy classes
of maps from $X$ to $X.$
\begin{theorem}
Let $X$ be a topological space and let $\M$ be the presheaf
of groupoids defined in SubSec.\ (\ref{SSecNMFE}) above.
Let $\F(X) \simeq \Simp(\M(X))$ and let 
$G(X) \simeq | \F(X) |.$ 
Then we have the following:
\label{ThmGXP1}
\leavevmode
\begin{enumerate}
\item Every $1-$simplex in $G(X)$ is in one-to-one 
correspondence with a pullback isomorphism of pairs $f$
in the subcategory $\F_X$ of $\C_X$
such that $F(f): X \to X$ is a nontrivial self-homeomorphism.
\item Any loop in $\pi_1(G(X))$ is homotopic to a 
simplicial loop in the $1-$skeleton of $G(X).$
\item Any simplicial loop in $\F(X)$ 
corresponds to a pullback isomorphism of the pair 
$x$ covering a nontrivial self-homeomorphism of the 
base $X$ of the pair $x.$
\end{enumerate}
\end{theorem}
\begin{proof}
\leavevmode
\begin{enumerate}
\item By Thm.\ (\ref{ThmCSimCat}),
the space $G(X)$ associated to
$\F(X) \simeq \Simp(\M(X)) \simeq \Simp(\pi(\F_X))$
has $1-$simplices joining two different equivalence
classes $\pi(e), \pi(e')$ in $\HC$
where $e,e' \in \C$ are pairs over $X$
whenever there is a pullback isomorphism
$\phi: e \to e'$ in $\C$
which covers a map $\pi(\phi): \pi(e) \to
\pi(e')$ in $\HC.$

By the discussion after Def.\ (\ref{DefCCat2})
every morphism in $\C$ may be uniquely factored as a 
composition of a pullback map with Bunke-Schick
isomorphisms of pairs. 
By definition of the functor $\pi: \C \to \HC,$ there can be
no $1-$simplices in $G(X)$ corresponding to Bunke-Schick
isomorphisms of pairs over $X$ since these map to the
identity morphism in $\HC$ and hence give rise to
contractible loops in $G(X) \simeq |\F(X)|.$

Hence every $1-$simplex in $G(X)$ corresponds exactly 
to a nontrivial pullback map $\phi: e \to e'$ in $\C.$ 
Such maps are unique sine they are
determined by the pullback
map on the circle bundles underlying the pairs $e, e'$ 
which in turn are fully determined by the map of 
topological spaces $f: F(e) \to F(e')$ where 
$F: \C \to \mathbf{Top}$ is the
forgetful functor defined in SubSec.\ (\ref{SSecPCLift})
above. 

Hence, for the map $\phi$ to be nontrivial, the
induced map on topological spaces $f:F(e) \to F(e')$ 
must be nontrivial. Since both $e,e'$ are pairs 
in $\M(X),$ the induced map $F(e) \simeq F(e') \simeq X.$
Thus $f:X \to X$ must be a nontrivial self-map and since 
$\phi$ is an isomorphism of pairs (since it is a morphism in 
$\M(X)$), $f:X \to X$ must be a self-homeomorphism of $X.$

\item Note that the simplicial space $\F(X)$ is Kan and
hence the group $\pi_1(G(X))$ which
consists of homotopy classes of loops in the
topological space $G(X)$ which
end at the basepoint $p_X$ of $G(X)$ (defined
at the end of SubSec.\ (\ref{SSecGXCX})) is actually
isomorphic to the simplicial 
homotopy classes of all simplicial loops in the simplicial
space $\F(X)$ which
end in the corresponding basepoint (also $p_X$)
of $\F(X)$ which is the simplicial homotopy group 
$\pi_1(\F(X)).$

Thus, since $G(X)$ is the geometric realization 
of the simplicial space $\F(X)$, given any loop $[\gamma]$
in $\pi_1(G(X))$ we may pick a simplicial loop $[s]$ in
$\F(X)$ and this gives us a simplicial loop 
in the $1-$skeleton of  $G(X)$ which is in the same
homotopy class as $[\gamma].$

\item This follows from the previous part. A simplicial loop
in $\F(X)$ corresponds to a pullback isomorphism of the
pair $x$ covering a nontrivial
self-homeomorphism of the base $X \simeq F(x)$ of 
$x.$
\end{enumerate}
\end{proof}

The simplicial category $\F(X) \simeq \Simp(\pi(\F_X))$ 
is a Kan complex by Thm.\ (\ref{LemFKan}) above, hence
the `horn-filling' property of a Kan complex 
(see Ref.\ \cite{nLab-Kan}) implies that a simplicial
homotopy of a simplicial loop in $\F(X)$ 
based at the pair $x \simeq (E, H) \in A(X)$ exists.

By definition (see Ref.\ \cite{nLab-SimpHot}),
for any $x \in \HC,$ a simplicial homotopy from
the arrow $f:x \to x$ to the arrow $g:x \to x$ in
the simplicial category $\F(X)$ is 
a simplicial morphism in $\F(X)$ of the form
$\Phi: \pi(x \otimes I) \to x$ satisfying the conditions of 
Ref.\ \cite{nLab-SimpHot} Sec.\ (2).
Here $x \otimes I$ is the cylinder object in the category
$\C$ associated to $x$ and $\pi:\C \to \HC$ is the projection
functor defined in SubSec.\ (\ref{SSecPCLift}) above.

We note that such a cylinder object was constructed
for the category $\C$ in Ref.\ \cite{BunkeS1} item (2.1.4):
The authors show that there is a
natural pair $(\tilde{E}, \tilde{h})$
over $X \times I$ in item (2.1.4) of 
Ref.\ \cite{BunkeS1} defined by the pullback maps
$f_i^{\ast}(\tilde{E}, \tilde{h}) = (E,h), i =0,1$ 
(here the inclusions $f_i, i=0,1$
are defined by $f_i: B \to I \times B$ with $f_i(b) = (i,b)).$ 
It is clear that this defines a cylinder object in $\C$
associated to the pair $(E,H)$ over $X.$

We use this cylinder object in the category $\C$
to define the cylinder object of any object 
$\gamma = [(E,h)]$ in $\HC$ as follows:  We 
define the cylinder object
of any element $\gamma$ as above to be
the object $\gamma \otimes I = [(\tilde{E}, \tilde{h})]$
in $\HC.$  It is easy to show that this object is independent
of the choice of $(E,h)$ and is a cylinder object 
in $\HC$ in the sense of Ref.\ (\cite{nLab-SimpHot}).

We define a simplicial homotopy of the object 
$\gamma \in \HC$ in the simplicial 
category $\F(X) \simeq \pi(\F_X)$ as a simplicial map 
$\Phi: \gamma \otimes I \to \gamma$ as in 
Ref.\ (\cite{nLab-SimpHot}).

We now show that elements of $P_1(X)$ correspond
to isotopy classes of maps from $X$ to $X.$

\begin{theorem}
Let $X$ be a topological space and let
$G$ be the functor defined at the beginning
of Sec.\ (\ref{SecP_i}). Let $P_0(X), P_1(X)$ be as above. 
Let $\eta: P_1(X) \to P_0(X)$ be the map defined
before Thm.\ (\ref{ThmGXP1}) which assigns
to each element $y \in P_1(X)$ the isomorphism
class of the pair $\eta(y)$ over $X.$

\leavevmode
\begin{enumerate}
\item Based simplicial homotopy classes of based simplicial loops in 
$G(X)$ with basepoint $x \in A(X)$
correspond to unbased isotopy classes of maps from $X$ to
$X$ in the mapping class group $MCG(X)$ of unbased isotopies of
the topological space $X \simeq F(x).$
\item An element $y \in P_1(X)$ corresponds to an unbased
isotopy class of maps $f:X \to X.$ 
\item The group $\eta^{-1}(x)$ of
any $x \in P_0(X)$ (see Thm.\ (\ref{ThmP1Pbck}) above)
is a quotient of $MCG(X),$ the mapping class group of
unbased isotopies from $X$ to $X,$
where $X \simeq F(x)$ is the base space
of the pair $x \in P_0(X).$
\end{enumerate}

Hence each element of $P_1(X)$ is an (unbased)
isotopy class in $MCG(X)$ acting on the
isomorphism class of the corresponding pair over $X$ 
by pullback.
\label{ThmPairMCG}
\end{theorem}
\begin{proof}
\leavevmode
\begin{enumerate}
\item Recall that $G(X) \simeq \Geom(\F(X))$ where
$\F(X)$ is the simplicial space associated to
the subcategory $\pi(\F_X)$ of $\HC.$

By Item (3) of Thm.\ (\ref{ThmGXP1}) above, any
simplicial loop $[s]$ based at a pair $x \in A(X)$ 
corresponds to an element of $\Aut_x(\M(X))$ and
hence corresponds to a pullback map from the pair
$x$ to the pair $x$ which covers a 
self-homeomorphism $\phi$
of the base space $X \simeq F(x).$ (Here 
$F:\C \to \mathbf{Top}$ is the forgetful functor defined
in SubSec.\ (\ref{SSecPCLift}).) We now want
to extend this statement to a homotopy of loops. 

We define a homotopy of the object $\gamma \in \HC$ 
in the simplicial category $\pi(\F_X)$ as a simplicial map 
$\Phi: \gamma \otimes I \to \gamma.$

This map is equivalent to a homotopy of morphisms of pairs
as follows: Let $\eta, \psi$ be two morphisms of pairs from
$x$ to $x.$ Then the homotopy $\Phi$ above is equivalent
to a map $\Phi_t: x \to x$ such that $\Phi_0 = \eta$
and $\Phi_1 = \psi$ and, by definition of the morphisms
of the category $\F(X),$ $\Phi$ must be a homotopy of
bundle morphisms i.e., a family of equivariant maps 
$\Phi_t:E \to E, t \in [0,1]$ which depend continuously
on $t$ such that each $\Phi_t$ preserve the $H-$flux under
pullback.

It is clear by quotienting
by the circle action on $E,$ that the 
homotopy of bundle morphisms $\Phi_t$ must 
cover a nontrivial homotopy $I \times X \to X$ where 
$X \simeq F(x)$ as above.  

Also, the set of maps
$\Phi_t$ must be morphisms of pairs in $\F(X)$
for every $t.$ Since the simplicial loop must stay in 
$\F(X)$ for every $t,$ $\Phi(t)$ must be an isomorphism 
of pairs for every $t \in [0,1].$

The bundle morphism $\Phi: I \times E \to E$
above is an equivariant homotopy from $E$ to
itself. It is clear that this equivariant 
homotopy induces a homotopy
of self-homeomorphisms of the base $X.$

Since the simplicial loop must always stay in the simplicial
space $\pi(\F_X)$ at every stage of the simplicial
homotopy, the induced map on the base space $X$ must
always be a homeomorphism at every stage of its
homotopy.

Thus, the simplicial homotopy class
of the loop should correspond to {\em isotopy classes}
of the induced self-homeomorphism of the base space
$X.$

\item Given an element $y \in P_1(X),$ we can obtain
an isomorphism class of pairs $\eta(y) \in P_0(X)$
using Thm.\ (\ref{ThmGXP1}) above and also recover
$X$ as the base space of
the isomorphism class of pairs $\eta(y) \in P_0(X)$ 
associated to $y.$

An element $y$ of $P_1(X)$ corresponds to
isotopy classes of maps in the mapping class group of $X$
denoted $MCG(X)$ acting on the isomorphism class
$\eta(y) \in P_0(X)$ of the pair underlying $y.$

\item By the previous parts, the action of an element
of $MCG(X)$ on $X$ will induce a well-defined action on 
each isomorphism class 
$x \in P_0(X)$ of pairs over $X.$ Thus the group
$\eta^{-1}(x)$ must be a quotient of $MCG(X).$
\end{enumerate}
\end{proof}

The proof of Thm.\ (\ref{ThmPairMCG}) does not show 
that the lift of the Mapping Class Group action on $X$ to
the isomorphism class of a pair $[x]$ over $X$ is
nontrivial. For a nontrivial lift of the Mapping
Class Group action on $X$ must lift to an action
on the principal circle bundle underlying the pair
$x.$

For example by Ref.\ \cite{ChenTs}, the lift of the 
Mapping Class Group action on any oriented nontrivial circle
bundle over $S^2$ will always be trivial. Hence, there
cannot be a nontrivial lift of the mapping class group
action to any nontrivial pair over $S^2.$ An
example where the lift is nontrivial is given in
Ref.\ \cite{ChenTs}.

{\flushleft \bf{Example:}} As an example of the above
theorem, we consider pairs over the two-torus $\KT^2.$
Without loss of generality we may assume that
$\KT^2$ has a smooth structure and all
our pairs are can be represented by smooth principal
circle bundles over $\KT^2.$
We assume that the underlying circle bundles of
these pairs are {\em oriented}.
It is well known that the Mapping Class Group 
$MCG(\KT^2)$ of orientation-preserving isotopy 
classes of maps from $\KT^2$ to itself is $SL(2,\KZ).$

The cohomology of $\KT^2$ is generated by two classes
$a,b$ in degree $1$ and their cup product $a \cup b$
generates $H^2(\KT^2,\KZ).$ A homeomorphism in
a given class of $MCG(\KT^2)$ corresponding to
a matrix $\gamma \in SL(2,\KZ)$ acts on $a,b$ as a
linear map. The induced map on
$H^2(\KT^2, \KZ)$ must be the determinant map and 
hence must the the identity.

Hence, every principal circle bundle $p_k: E_k \to \KT^2$
of characteristic class $k \in H^2(\KT^2, \KZ)$ is mapped
to itself (up to gauge transformations of bundle) 
by the action of $MCG(\KT^2).$
Also Ref.\ \cite{ChenTs} proves that each
oriented isotopy class in $MCG(\KT^2)$ lifts
to an oriented isotopy in $MCG(E_k).$ Hence every
oriented pair over $\KT^2$ has a nontrivial action by 
orientation preserving isotopies of $\KT^2.$

It is possible to extend $MCG(\KT^2)$ to a 
bigger group by including orientation reversing
homeomorphisms of $\KT^2.$ However,
acting on a pair $y \in P_0(\KT^2)$ 
by pulling back the pair by an
orientation-reversing homeomorphism
of $\KT^2$ will yield a pair with reversed orientation.
If the original pair is nontrivial, the pulled back pair
cannot be Bunke-Schick equivalent to
$y$ and hence must correspond to a different element
$w \in P_0(\KT^2).$ Hence, the pullback of a nontrivial
pair by an orientation-reversing homeomorphism of
$\KT^2$ cannot be lifted to an automorphism of the 
Bunke-Schick isomorphism class of the pair.
Hence it cannot yield an element of $P_1(\KT^2).$  

Hence the orientation-preserving isotopies of $\KT^2$ 
which form the group $SL(2,\KZ)$ lift to each oriented
pair over $\KT^2$ but
none of the orientation-reversing isotopies of $\KT^2$
will lift. \qed

{\flushleft {\bf Example:}} Let $K$ be a knot tamely 
embedded in $S^3$ and consider $X \simeq S^3 -K.$
Without loss of generality we may assume that $X$
has the structure of a manifold.
Then the symmetry group of the knot $K$ namely
$MCG(S^3 -K)$ acts on isomorphism classes
of pairs over $X.$ If the action lifts we would obtain
an action of the symmetry group of the knot $K$
on isomorphism classes of pairs over $X.$ 

As an example let $K$ be the $(p,q)-$torus knot in
$S^3.$ Then, $X$ is a Seifert fibered space over
an orbifold which is a disk with two orbifold points
one of order $p$ and one of order $q.$ 
If $p=q,$ then $MCG(X) \simeq \KZ/2.$ It
would be interesting to see if the action lifts. \qed

We suspect the reason for the appearance of 
$MCG(X)$ is the following:
SFT projected to its massless modes induces a 
SFT-like gauge theory on each spacetime background.
If we restrict to SFT gauge transformations which
do not distort the spacetime background too much
(see Ref.\ \cite{SFTDiffeo} and the discussion
in SubSec.\ (\ref{SSecSFT-TTD}) above) this theory
reduces to a theory with gauge transformations which are
diffeomorphisms of the base and bundle
gauge transformations (assuming the chosen background
has a free circle isometry).

Gauge equivalence classes of pairs in $\HC$ correspond to 
moduli spaces of this low-energy theory. It is well-known
(see Ref.\ \cite{BeniSch} for a discussion for gauge
theories) that such moduli spaces can be used to study
the action of the mapping class group of a space. 

\subsubsection{Triviality of the $P_k(X), k \geq 2$}

So far, we have defined $P_0(X), P_1(X)$ using
$G(X).$ We now give a definition of $P_k(X), k \geq 2$
as a functorial pullback of a generalization of the 
fundamental groupoids $\Pi_k(G(X)).$

For the definition of the higher T-duality functors
$P_k(X), k \geq 2$ we use the generalization of the
fundamental groupoid $\Pi_1(X)$ of a space $X$
to the higher homotopy
groupoids $\Pi_k(X)$ of Ref.\ \cite{GraVit}. We 
replace $\Pi_1(G(X))$ in Eq.\ (\ref{GrpdP1Pbck}) by 
higher homotopy groupoids $\Pi_k(G(X))$ of
Ref.\ \cite{GraVit} defined as
$\Pi_k(G(X)) \simeq \Pi_1(\Omega^{k-1}G(X)), k\geq 2$
where $\Omega^i G(X)$ is the $i-$th loop space of
$G(X).$ (Since $G(X)$ is the geometric realization of
the simplicial space $\pi(\F_X)$ the space $\Omega G(X)$ 
always has a simplicial model and thus $\Pi_k(G(X))$ is 
always the fundamental groupoid of the geometric
realization of a simplicial space for every $k.$)

We define $P_i(X)$ as 
a $2-$pullback of $\Pi_k(G(X))$ which fills
in a square similar to the $2-$pullback square in 
Eq.\ (\ref{GrpdP1Pbck}) above
\begin{equation}
\begin{CD}
P_k(X) \simeq (\phi \downarrow p_k) @>>> 
\Pi_k(G(X)) \simeq \Pi_1(\Omega^kG(X)) \\
@VVV   @VV{p_k}V \\
P_0(X) \simeq ((\M)_0)(X)) @>>{\phi}> \pi_0(G(X)) \simeq \pi_0(\M(X)).  \label{GrpdPkPbck} \\
\end{CD}
\end{equation}

We now show that the functors $P_k(X)$ must be zero
if $k \geq 2.$ This might seem redundant, but we will
use this definition of higher Topological T-duality
functors in Sec.\ (\ref{SecPropP0P1}) below 
and hence are introducing them here.

\begin{theorem}
Let $G(X)$ be the nerve of the groupoid $\M(X)$
as described above. Let $P_k(X)$ be the higher
Topological T-duality functors defined by $2-$pullback
in Eq.\ (\ref{GrpdPkPbck}) above.
\label{ThmPikG}
\leavevmode
\begin{enumerate}
\item All the homotopy groups
$\pi_i(G(X)) = 0$ whenever $i \geq 2.$
\item For all topological spaces $X,$ the higher 
Topological T-duality functors 
$P_k(X)$ for all $k \geq 2$ are the trivial groupoid
whose set of objects is $P_0(X)$ and whose only arrows
are the identity map.
\end{enumerate}
\end{theorem}
\begin{proof}
\leavevmode
\begin{enumerate}
\item By construction the category $\M(X)$
is a groupoid for any $X.$
Also $G(X)$ is actually the geometric realization
of the nerve of this groupoid. 

The simplicial set
associated to the nerve of a groupoid $\G$ 
only has nontrivial 
homotopy groups $\pi_k$ when $k \leq 1$ with $\pi_1(x)$
for any object $x$ in $\G$
the automorphism group of a given pair
(see Item (4) Lem.\ (4.1) in Ref.\ \cite{nLab-SimpHotGrp}).
  
By Lem.\ (\ref{LemHotAG}), $\pi_k(G(X),p_X) \simeq \pi_k(G(X))$ 
which is zero if $k \geq 2.$
\item For any connected space $W,$ we have that
$[S^1, \Omega^{k-1} W] \simeq [\Sigma^{k-1} S^1, W]
\simeq [S^k, W] \simeq \pi_k(W).$ 
So if $\pi_k(W) \simeq 0,$ then $[S^1, \Omega^{k-1} W] \simeq 0.$

Due to this argument and the previous item, 
the fundamental groupoid
$\Pi_1(\Omega^{k-1} G(X))$ must be trivial. Also, the pullback 
in the diagram Eq.\ (\ref{GrpdPkPbck}) above will be trivial
for every connected component of $\Omega^{k-1} G(X).$ 
Thus by the previous part the value of $P_k(X)$ over
any $x \in P_0(X)$ is the trivial group for 
$k \geq 2.$
\end{enumerate}
\end{proof}

The above results let us determine the geometry of $G(X):$
\begin{corollary}
Let $X, G(X)$ be as in the previous theorem. Then, the universal cover of
$G(X)$ is contractible. Also, 
$G(X) \simeq BP_1(X)$ where $P_1(X)$ is defined after Lem.\  (\ref{CorPi1GX}) above.
\end{corollary}
\begin{proof}
Now, if $\tilde{G}(X) \to G(X)$ is the universal cover of $G(X),$ 
by Thm.\ (\ref{ThmPikG}), $\pi_i(\tilde{G}(X)) =0, i \geq 1.$
Hence, $\tilde{G}(X)$ must be contractible by Hurewicz
theorem. Now $P_1(X) \simeq \pi_1(G(X))$ is a group 
acting freely on $\tilde{G}(X)$ so,
this implies that $\tilde{G}(X) \simeq EP_1(X)$ and $G(X) \simeq BP_1(X).$
\end{proof}

\subsubsection{Physical Significance of $P_k(X)$}
Due to background independence, 
SFT possesses gauge theory-like symmetries 
which correspond to background preserving and background
changing transformations of the underlying spacetime
background. As we had discussed in
SubSec.\ (\ref{SSecSFT-TTD}) above the morphisms
in the category $\C_X$ consist of both background
preserving and background changing transformations
of SFT. 

However, in this paper we have restricted
our attention to morphisms in the category
$\pi(\F_X)$ and used only these morphisms
to construct $P_0(X), P_1(X).$
The morphisms in 
$\pi(\F_X)$ are equivalent to pullback
isomorphisms in $\C_X$ covering a nontrivial
self-homeomorphism of the base $X$
and hence correspond
to background preserving and background changing 
gauge transformations of pairs over
$X$ in SFT. 
In SubSec.\ (\ref{SecPropP0P1}) we argue that it
should be possible to construct more complex invariants
which correspond to arbitrary transformations in 
SFT 
(these correspond to arbitrary morphisms in $\C_X$).

From the discussion in this section we see that the 
Bunke-Schick Functor $P_0$ and its extension $P_1$
defined here have the same value on
the set of pairs isomorphic to a given pair $([p],H):$
\begin{enumerate}
\item \label{Enum:P0} The isomorphism class of the pair under isomorphism
of pairs (not the isomorphism class of the circle
bundle structure $p:E \to X$ which is the class 
$[p] \in H^2(X,\KZ)$).
The set $P_0(X)$ is exactly the set of
these isomorphism classes.
\item \label{Enum:P1} All possible automorphisms of a 
fixed pair which cover any self-homeomorphism of the 
base $X$ of the pair. As described in 
Cor.\ (\ref{CorPi1GX}) above
$P_1(X)$ is constructed from these automorphisms. 
\end{enumerate}

It is strange that these two invariants of a pair are not
exactly $[p]$ and $H,$ but the two quantities which
are connected to $P_0(X), P_1(X)$ described above. 
This is because isomorphisms of pairs can change $[p]$  
and $H$ as discussed in Sec.\ (\ref{SecIntro}) between 
Lem.\ (\ref{LemPairHom}) and Lem.\ (\ref{LemHotIso}). 

By the arguments in SubSec.\ (\ref{SSecSFT-TTD})
the isomorphism class of a 
pair $(E,[H])$ in $\C$ described in 
Item (\ref{Enum:P0}) above 
corresponds in SFT 
to equivalence classes 
of data of the form $G,B$ associated to 
backgrounds $E$ over $ X$ 
under the background preserving gauge transformations 
of SFT $\delta E$
(described in that subsection)
which preserve the principal 
bundle structure with base $X.$ 
The functor $P_0$ takes the same value on each 
equivalence class of $G, B$ above and hence its 
value labels the equivalence classes of these backgrounds.

The functor $P_1$ takes the same value on
equivalence classes of automorphisms of a fixed pair 
$(E \to X, [H])$ which cover a non-trivial 
self-homeomorphism of the base $X$ of the pair
as mentioned in Item (\ref{Enum:P1}) above.
In SFT, these correspond to 
background-changing gauge transformations of the
spacetime background associated to the pair since they 
correspond to diffeomorphisms of the spacetime background 
associated to the pair which are not homotopic to the 
identity and hence might change the characteristic
class of the underlying principal circle bundle or its 
$H-$flux---see discussion in Sec.\ (2.1) of Ref.\ \cite{KZw} 
especially around Eq.\ (2.22). 

\section{Properties of $P_0$ and $P_1$ \label{SecPropP0P1}}

In this section we discuss some elementary properties of the functors $P_i.$ We first discuss a way of calculating pairs
over $X$ using the relative functor $P_i(X,A)$ defined above. 
Then we discuss two possible natural 
extensions of the functors $P_i(X).$


\subsection{Exact Sequence for $P_i$}
If $G(X)$ is connected, then the long exact sequence of
homotopy groups for a space gives us a long exact
sequence for the functors $P_0, P_1$ and also relative
versions of these functors: For $j:A \hookrightarrow X$
and each $i \geq 0,$ we have an inclusion
$P_i(X) = \pi_i(G(X))$ and a relative functor
$P_i(X,A) = \pi_i(G(X),G(A)).$
Note that for $i >0 ,$ each of the functors
$P_i(X), P_i(X,A)$ are valued in  groups. 

Let $A \hookrightarrow X$ be an inclusion. We have an induced map by pullback 
$C(X) \to C(A).$ Restriction of pairs induces a map of simplicial sets $\phi:G(X) \to G(A).$
By changing $G(X), G(A)$ up to homotopy we can view $\phi$ as a fibration
with homotopy fiber $K\phi$ such that $K\phi \hookrightarrow G(X) \to G(A)$
is a short exact sequence of spaces.
We have a long exact sequence of higher homotopy groupoids
(see Ref.\ \cite{GraVit} Sec.\ (2), Cor.\ (4))
\begin{equation}
\ldots \to \Pi_n(K\phi) \to \Pi_n(G(X)) \to \Pi_n(G(A)) \to \ldots.
\label{EqSESPik1}
\end{equation}
Due to Thm.\ (\ref{ThmPikG}) above Eq.\ (\ref{EqSESPik1}) reduces to
\begin{equation}
0 \to \Pi_1(K\phi) \to \Pi_1(G(X)) \to \Pi_1(G(A)).
\label{EqSESPik2}
\end{equation}

Under the $2-$pullback construction we will obtain another exact
sequence for $P_i$ since $2-$pullback is left exact.
Hence we should obtain
\begin{equation}
0 \to P_1(X,A) \to P_1(X) \to P_1(A)
\label{EqSESPik}
\end{equation}
where $P_1(X,A)$ is defined as the pullback of $\Pi_1(K\phi)$ along $\phi$ as in 
Thm.\ (\ref{ThmPikG}) above. Thus, there is a short exact sequence relating
$P_1(X)$ to $P_1(A).$

If $G(X)$ is connected we can show that $P_i(X) \simeq \pi_i(G(X)).$
For such spaces $X$ we can prove the following.
Pick a basepoint $ \{ \ast \} \subseteq A \subseteq X.$
This gives a map $$G(X) \overset{\phi}{\to} G(A) \to G(\ast).$$
The unique trivial pair over the basepoint
gives a natural basepoint in $G(A) \subset G(X).$
 
Define the fiber $F$ of the above map to be the fiber over $G(\ast).$ 
This is the collection of pairs in $X$ which restrict to the trivial pair over $A.$
We have a fibration $$F \hookrightarrow G(X) \to G(A).$$ 
Note that the above property is true of $G(A)$ as well as $G(X).$
Hence, we obtain from the fiber homotopy sequence and the above that
$$
\pi_n(F) \simeq  0, n \geq 2
$$
and 
$$
0 \to \pi_1(F) \to P_1(X) \to P_1(A) \to \pi_0(F) \to P_0(X) \to \{ \ast \}
$$
\begin{lemma}
Suppose $A$ is such that every pair over $X$ restricts to the
trivial pair over $A.$ Then, $\pi_i(F) \simeq P_i(X), i=0,1.$
\end{lemma}
\begin{proof}
Suppose $A$ satisfies the hypotheses above. Then, every pair in $X$ is a trivial pair over $A$ 
and hence belongs to $F.$ Thus, $F$ is homotopy equivalent
to $G(X)$ and the map $G(X) \to G(A)$ is nullhomotopic.  Hence, $\pi_1(F) \simeq P_1(X)$
and $\pi_0(F) \simeq P_0(X).$
\end{proof}

For any space $X$ and for $A = \ast,$ there can only be one pair over $A$ the trivial pair over $A.$
Hence, by the above theorem $\pi_i(F) \simeq P_i(X), i=0,1.$ However the 
fiber $F$ above is difficult to caculate for most examples.

\subsection{Generalizations}
We now propose two extensions of the formalism presented
in this paper which would help model SFT
backgrounds better. 

Firstly we argue below that changing the
category $\pi(\F_X)$ 
which we used to study Topological T-duality
functors would give a better insight into the structure 
of the set of SFT gauge transformations
acting on pairs over $X.$  We suggest that doing this
would let us define new functors $P'_k(X)$ so that
$P'_0(X) \simeq P_0(X) , P_1'(X) \simeq P_1(X)$ but 
$P'_k(X)$ are not automatically trivial for $k \geq 2.$

Secondly, we argue that it is important to describe 
spacetimes which carry extra data such as Ramond-Ramond
fields or derived functor data on them. We suggest
a natural extension of the formalism in this paper using
Thomason Cohomology of small categories.

\subsubsection{Nontrivial $P_k, k \geq 2$ \label{SSecNonTrivP_k}}
Recall that in SubSec.\ (\ref{SSecSFT-TTD})
we had defined $\pi(\F_X)$ as the 
collection of gauge transformations which act
on pairs over $X.$ We had also defined
$G(X) \simeq B \pi(\F_X).$
We had defined the higher functors
$P_i$ in the above by
using the fundamental groupoid $P_1(X) = \Pi_1(G(X)).$ 
We showed in Sec.\ (\ref{SecP_i}) above that
if we define higher functors $P_i(X)$ from the
higher groupoids $\Pi_k(G(X))$ all $P_k(X)$ are
trivial.

The reason for this is that $G(X)$ has no structure in
degrees higher than one. This is due to the structure
of the subcategory $\pi(\F_X)$ of $\HC$ since
the classifying space $B \pi(\F_X)$ 
is the $1-$skeleton of $B \pi(\C_X)$
and the inclusion $B \pi(\F_X) \hookrightarrow B \pi(\C_X)$ induces an
isomorphism on homotopy groups up to degree $1.$

In addtion in SubSec.\ (\ref{SSecSFT-TTD}) we had argued
that $\pi(\F_X)$ should be viewed as a collection of
SFT gauge transformations which
act on isomorphism classes of 
pairs over $X$ and yield new pairs over $X.$
Each of these transformations lifts to a
pullback isomorphism of pairs which
cover a nontrivial self-homeomorphism of the base $X.$
However, these are not all the gauge transformations
that are possible in SFT. 

In order to obtain nontrivial higher Topological T-duality
functors it is natural to consider a larger set of gauge
transformations. Since our category $\C_X$ has as objects
the collection of all pairs over $X,$ the largest set we 
can consider is the set of gauge transformations
which could change the isomorphism class of a pair but
leave $X$ alone. In particular, we would need to consider
pullback morphisms of pairs over $X$ which cover arbitrary
morphisms of $f:X \to X$ i.e. morphisms which are not
self-homeomorphisms of $X.$

As a result, we consider the subcategory
$\pi(\C_X)$ of $\HC$ since as we had argued above
the subcategory
$\pi(\F_X)$ may be viewed as an approximation of 
$\pi(\C_X).$ This implies we use
the classifying space $B\pi(\C_X)$ in defining
the functors $P'_k(X)$ above
instead of the classifying space $B\pi(\F_X)$ 
since the space $B\pi(\C_X)$ includes $B\pi(\F_X)$ 
as its $1-$skeleton. Hence, the resulting functors
$P'_k(X) \simeq \pi_k(B(\C_X))$ should agree with 
$P_0(X)$ and $P_1(X)$
when $k =1,2,$ and should not be automatically 
trivial when $k \geq 2.$

It is important to study the above extension 
for physical reasons since the collection of all gauge
transformations of SFT which leaves $X$
invariant must be $\C_X$ and hence the effects of
these gauge transformations on pairs up to isomorphism
of pairs must be encoded in the category $\pi(\C_X).$

It is easy to argue that the classifying space
$B\C_X$ must have nontrivial higher homotopy
groups. For example, maps $f:X \to X$ 
which are not homeomorphisms must contribute a
$2-$cell to the classifying space $B\C_X$ since, by
the discussion in SubSec.\ (\ref{SSecMorPXY}), any
such map may be factorized uniquely up to homotopy
equivalence as $f = h \circ g$ with $h:X \to M_f$ a
cofibration and $g:M_f \to X$ a homotopy equivalence
as in SubSec.\ (\ref{SSecMorPXY}). 

Note that if $f$ is a homeomorphism $f:X \to X,$ this
construction collapses since then $h$ is also a homotopy
equivalence and $f$ is homotopy equivalent to $h$ --- in
this case the glued $2-$cell must be homotopic to a $1-$cell
by the definition of simplicial realization of a category. Thus
we will recover the result of this paper if we restrict 
ourselves to the subcategory $\pi(\F_X)$ of $\pi(C_X).$

Recall we had defined a functor $\G: \mathbf{CW}^{op}
\to \mathtt{Cat}$ in Thm.\ (\ref{ThmGFunc})
which assigned to any topological
space $X$ the small category $\C_X$ and assigned
pullback functors to morphisms $f:Y \to X.$
If we could repeat the constructions in
this paper with the functor $\G$
above we might be able to define higher Topological
T-duality functors $P_k(X)$ which are nontrivial when 
$k \geq 2.$ We had proved in the previous sections
that $P_1(X)$ was an isotopy
invariant. It would be interesting to study the
extension described above further.

\subsubsection{Flux backgrounds}
We now examine another matter which is interesting
physically. Many physical spacetimes carry extra data
on them apart from the $H-$flux and graviton fields
for example Ramond-Ramond fields or $D-$brane
configurations. These backgrounds have been studied
in string field theory.

It is interesting to ask if these types of backgrounds can 
also be modelled in the formalism developed in this paper.
It is clear that a first step in including such 
backgrounds into the class of backgrounds
studied in this paper would be to require more data 
to be assigned to a pair. As an example, 
to the pair $(E,H)$ we could assign a derived
functor naturally associated to that pair together with
the data in the pair $(E,H).$

Thus, we would need to make a new category similar
to the category of pairs $\C$ which would consist
of spacetime backgrounds with the required extra data.
It is not clear how to do this naturally, so that the new
backgrounds map naturally under morphisms of pairs.

We suggest the following partial solution to this matter:
A natural way to assign more data to a pair is to use
Thomason Cohomology for small categories in particular
for the small category $\C_X$ or an extension of
$\C_X$ (we had shown in SubSec.\ (\ref{SSecDefCX}) 
above that $\C_X$ is  small category for every $X$).
An exposition of the homology and cohomology of small
categories can be found in Chap.\ (16) of 
Ref.\ \cite{BRichter} especially Def.\ (16.1.1)).

A simple example of Thomason Cohomology for 
small categories is the following extension of 
the higher Topological T-dualiy functors to 
Topological T-duality functors with coefficients and 
Topological T-duality functors associated to a spectrum.

By analogy with the usual definition of homology
with coefficients we may define
Topological T-duality functors with `coefficients in a group 
$K$' as
$Q_i^K(X) = \pi_i(G(X)\wedge \K)) = \tilde{H}_i(G(X), K)$ 
where $\tilde{H}_i(X,K)$ is the {\em reduced homology} of $X$
with coefficient group $K,$ the group $\pi_i$ is 
the homotopy group 
of a spectrum and $G(X) \wedge \K$ denotes the 
smash product of 
$G(X)$ with the Eilenberg-Maclane Spectrum $\K$ 
associated to $K.$ 

In addition if $S$ is an arbitrary spectrum we can define 
$Q_i^S(X) = \pi_i(G(X)\wedge S)$ to be the
Toplogical T-duality functors twisted by a spectrum
$S$ and this would be a generalization of the above.
All these functors should be invariant under the action
of Topological T-duality on $G(X)$ just as the functors
$\pi_i(G(X)), i=0,1$ are.

The assignment described in the previous two paragraphs 
above is clearly a Thomason natural system on $\C_X$ 
with values in a category $\E$ or a contravariant
Thomason natural system on $\C_X.$ 
Here, we could pick $E$ to be a suitable chain complex
giving us the required cohomology
and the resulting cohomology theory on $\C_X$ would
carry an action of Topological T-duality. 
It would be interesting to study this further.

\section{Isomorphism Classes of Triples and $P_1(X)$ \label{SecP32P1}}
We know that $P_1(X)$ is generated by closed loops in the space $G(X).$
We now show that these loops are in one-to-one correspondence with
classes in $H^2(E_{\alpha},\KZ)$ for some principal circle bundle 
$p_{\alpha}: E_{\alpha} \to X:$ 
 
\begin{theorem}
There is a one-to-one correspondence between homotopy classes of loops in $\pi_1(G(X))$ and classes in 
$H^2(E_{\alpha},\KZ)$ for $\alpha \in A(X).$
\label{ThmPi1H2}
\end{theorem}
\begin{proof}
Suppose we are given a homotopy class of a loop $[\alpha] \in \pi_1(G(X)).$  Since 
$\pi_1(G(X)) \simeq \pi_1(F_1),$ (as in the proof of Thm.\ (\ref{ThmSSeqPi1FG}) above),  
we naturally obtain a loop in the fiber of the natural map $F_1 \to G_1(X) \to G_0(X).$
This corresponds to a loop in the $1-$skeleton of $G(X).$ By Lem.\  (\ref{LemHotAG}) and
excision, this loop corresponds to  a path in $G(X)/A(X)$ 
between the base point pair $(E, 0)$ and another pair (say $v \simeq (E_{\alpha},H)$) together with a loop corresponding to an 
automorphism of the pair $v.$ 

The pair $v$ is isomorphic to a gerbe with band $U(1)-$with
Dixmier-Douady invariant $H$. There is a one-to-one correspondence between automorphisms
of this gerbe and automorphisms of the associated pair (see Ref.\ \cite{Pan2}).
Hence automorphisms of such gerbes are in one-to-one correspondence with loops in $G(X)$
based at $v.$ Two loops are homotopic if and only if they induce equivalent
automorphisms of $v$ (see discussion before Lem.\ (\ref{LemHomTf}) ).

It is well-known that the equivalence classes of these automorphisms are classified by
elements of $H^2(E_{\alpha},\KZ)$ (see for example Ref.\ \cite{BunkeS1}, Sec.\ (3.1) and also
Ref.\ \cite{Pan2}). 
\end{proof}

In Ref.\ \cite{Pan2} it is shown that the isomorphism classes of
such triples form a set $P_{3,2}(X).$ Are the functors $P_0(X), P_1(X)$ 
connected to the space of triples of Ref.\ \cite{Pan2}
\begin{gather}
X \to P_{3,2}(X) \simeq \mbox{ Equiv. Classes of  triples of the form } \nonumber \\
\mbox{ ( Principal Bundle $E \to X$, Class in $H^2(E,\KZ)$, Class in $H^3(E,\KZ)$ ) } 
\label{EqTriple}
\end{gather}
over $X?$ 

We can also consider assignments  of the form
\begin{gather}
X \to (P_0(X), P_1(X))
\label{EqP0P1Triple}
\end{gather}

It is clear from the arguments in the previous sections
that Eq.\ (\ref{EqTriple}) and Eq.\ (\ref{EqP0P1Triple}) 
are rearrangements of the same data.
It is also clear from the above that Topological T-duality acts on the data in
Eq.\ (\ref{EqP0P1Triple}).
Thus, we should expect that Topological T-duality should act on the data in Eq.\ (\ref{EqTriple}).
We suggest that this action of Topological T-duality on the data in
Eq.\ (EqTriple) above is the Topological T-duality for triples investigated in Ref.\ \cite{Pan2}.

Note that there is a forgetful natural transformation from $P_{3,2} \to (P_0,P_1)$
which sends $\mbox{ Equivalence class of } ([p], H, b)$ to
\begin{gather}
( \mbox{ Equivalence class of } ([p], H), \phi_b(\mbox{ Equivalence class of } ([p], H) ) ) 
\label{EqP32P0P1}
\end{gather}
where $\phi_b$ is an element of $\pi_1(G(X))$ corresponding to $b \in H^2(E,\KZ).$  
It would be interesting to investigate the relationship
between the functor $P_{3,2}$ of Ref.\ \cite{Pan2} and
the above pair $(P_1,P_2).$ Note that for any compact $CW-$complex $X,$
$P_{3,2}(X)$ is {\bf Set}-valued, unlike
$P_1(X)$ above which is a {\bf group}.

Note that $P_{3,2}(X)$ was studied by directly defining a category of T-duality triples over
a base space $X.$ The above functor $(P_0(X), P_1(X))$ is more geometric and depends on
the construction of the space associated to the category of pairs over $X.$ 

{\em \bf  Conjecture: The image of $P_{3,2}(X)$ via the above natural transformation
generates $P_1(X)$ as a group.}

%

Since there cannot be any more functors apart from $P_0, P_1$ in this picture,
it is reasonable to ask whether Topological T-duality and Topological T-duality for Triples
are the only T-duality-like transformations possible beginning with the category of pairs over
a base space $X.$

As shown in Ref.\ \cite{Pan2}, if $X$ is not simple, $P_{3,2}(X)$ is related in
a nontrivial manner with the action $\pi_1(X)$ on the higher homotopy groups of
$X.$ It is interesting to ask if $P_1(X)$ displays similar behaviour.

As discussed above the functor $P_{3,2}(X)$ can be recovered from $P_0(X), P_1(X).$ 
If $P_{3,2}(X)$  can be related to a suitably defined Thomason cohomology group of $\C_X$
for example the cohomology with local coefficients of a suitable local system on $G(X)$
(see Remarks at the end of Sec.\ (\ref{SecPropP0P1})),
it might be possible to geometrize $P_{3,2}(X)$ just as $P_0, P_1$ are geometrized by
$\pi_k(G(X))$ and gain more insight into
the functors $P_{3,2}(X), P_0(X)$ and $P_1(X).$ 
For this, it is natural to try to construct a
contravariant Thomason system on $\C_X$ using
Thm.\ (\ref{ThmPi1H2}) above. It would be interesting to examine this further.

\section{Doubled Geometries and T-Folds \label{SecTFold}}
In SubSec.\ (\ref{SSecTFI}) we briefly review doubled
geometries and T-folds. In SubSec.\ (\ref{SSecTFDGC})
we argue that doubled geometries appear naturally
in the formalism of this paper. In SubSec.\ (\ref{SSecTFTF}) 
we describe how the above formalism may be naturally
extended to include T-fold backgrounds. 

\subsection{Introduction \label{SSecTFI}}
The idea of a doubled geometry orginated from 
SFT.
In Ref.\ \cite{KZw}, Kugo and Zwiebach showed
that the SFT on a toroidal background
$E$ could be naturally 
written using the coordinates 
on the original background and on the T-dual background.
The authors proved that 
SFT defined on
a given background $E$ with a free circle action
can be described in terms of the coordinate variable 
$x$ on $E,$  its conjugate momentum $p,$ 
the string winding $w$ and a variable 
$q$ conjugate to $w.$
Here, $q$ may be identified with the coordinate on the 
T-dual principal circle bundle.

In Ref.\ \cite{HullZ} Hull and Zwiebach showed
that closed SFT on a torus
background naturally defines a 
double field theory on a doubled space constructed 
from the original torus background and its T-dual
torus background. The authors showed
that the doubled space
has as natural coordinates the variables
$x$ and $q$ described in the previous
paragraph. Hull and Zwiebach 
conjectured that a doubled field theory 
propagated on the doubled space. 
This doubled field theory was naturally obtained
from closed SFT on the original 
background and was described using fields which were
functions of the coordinates on the original
background and the coordinates on the T-dual torus
background simultaneously.

In Ref.\ \cite{HullT} Hull showed that doubled
geometries could be constructed from string theories
propagating on spacetime backgrounds which were 
arbitrary principal torus bundles. 
He showed that when the T-dual principal
torus bundle was nongeometric, the doubled geometry
was no longer a classical manifold but a generalized
space termed a T-fold. A T-fold behaves like a 
principal torus bundle locally except
that some of the transition functions are replaced
by T-dualities.

Further Arvanitakis and coworkers demonstrated in 
Ref.\ \cite{AHull, AHull2} that the $L_{\infty}-$algebra
describing the infinitesimal symmetries 
(gauge transformations) of
SFT could be projected onto variables
describing the low energy behaviour of SFT
on toroidially compactified backgrounds only and
excluding all other backgrounds. 
This projected theory naturally yielded
the double field theory description of
such backgrounds described in the previous paragraph.
Thus the doubled field theory description on the doubled
geometry is actually a reduction of SFT
on the doubled geometry. 

\subsection{ Doubled Geometry and $\C$ 
\label{SSecTFDGC}}
We have seen that both the doubled geometry and 
the field theory on it are reductions of SFT.
In this paper we suggest that Topological T-duality
should also be viewed as a reduction of String Field
Theory. Hence, the doubled space should appear
in the Topological T-duality description.

In Refs.\ \cite{MR1, BunkeS1} the
doubled geometry appears in a commutative
diagram termed the Diamond Diagram (see 
Ref.\ \cite{BEM} for an explanation in terms
of string theory) connecting the original pair 
$(E,H)$ over a topological space $X,$ the dual 
pair $(E^{\#}, H^{\#})$
also over $X$ and a third space the 
correspondence space identified
as the fiber product $E \times_{X} E^{\#}$ which
is a $\KT^2-$bundle over $X$ 

\begin{equation}
\begin{CD}
E @<<{p^\ast(p^{\#})}< (E \times_{X} E^{\#}) \\
@VV{p}V   @VV{p^{\#}(p)}V \\
X @<<{p^{\#}}< (E^{\#}). \label{CSp1} \\
\end{CD}
\end{equation}

Thus for principal circle bundles with $H-$flux the
correspondence space may be identified with
the doubled geometry of Zwiebach, Kugo and Hull
in Refs.\ \cite{KZw, HullZ}.
All the known formalisms of Topological T-duality
have a natural construction of the correspondence space.

In this paper we have argued that Topological
T-duality may be naturally described using
closed SFT. We argue in this section
that the correspondence space can also
be naturally constructed using the formalism of this paper. 
We also argue at the end of this section that extending
the approach in this paper to principal 
torus bundles and their associated
nongeometric T-duals will permit the description of 
T-fold spaces.

For a $CW-$complex $X,$ let $G$ be as in 
Subsec.\ (\ref{SSecGXCX}) above.
Recall that we have the category $\C$ of pairs and
a natural forgetful functor $F:\C \to \mathbf{Top}.$
We also have the category of homotopy equivalence
classes of pairs with pullback morphisms between them
$\HC$ and a natural projection functor $\pi:\C \to \HC.$
In addition we have a functor $E: \HC \to \mathbf{Top}$
such that $F \simeq  E \circ \pi.$

Topological T-duality sends isomorphism classes
of pairs over $X$ to isomorphism classes of pairs over
$X$. For pairs whose underlying bundle is a principal
circle bundle, the Topological T-duality symmetry group
is $\KZ_2.$ Hence, $\HC$ is a category with a 
$\KZ_2-$action and we have a natural $\KZ_2-$symmetry
of the essential fiber $\pi(\F_X)$ of the functor 
$E: \HC \to \mathbf{Top}$ induced by Topological T-duality. 
Now the Topological
T-dual of the isomorphism class of a pair over any
topological space $X$ is another such isomorphism
class over the {\em same} space $X,$ hence
Topological T-duality acts 
as a fiber-preserving automorphism of the 
functor $E$ which also preserves the essential fiber of
$E$ over any topological space $X.$

For any pair $(E,H)$ in $\C$ we define
$\mathcal{T}([(E,H)]) \simeq E \times_{X} E^{\#}.$
We argue that $\mathcal{T}$ defines a functor
$\mathcal{T}: \HC \to \mathbf{Top}$ constructed
in the next paragraph. We show that this functor
possesses the property that
the Topological T-duality automorphism of $\HC$
restricted to the essential fiber $\pi(\F_X)$ of $E$
maps to a nontrivial self-homeomorphism of 
${\mathcal{T}}({\mathcal C}_X).$ 

To construct the functor $\mathcal{T},$ 
we first argue that such a functor must
exist using the `Topological T-duality Diamond diagram' of
Ref.\ \cite{BEM}.  By Ref.\ \cite{BEM}
we have a commutative diagram
of principal circle bundles whose top right element
is a fiber product of $E$ and $E^{\#}$
\begin{equation}
\begin{CD}
E @<<{p^\ast(p^{\#})}< (E \times_{X} E^{\#}) \\
@VV{p}V   @VV{p^{\#}(p)}V \\
X @<<{p^{\#}}< (E^{\#}). \label{EqPairIso} \\
\end{CD}
\end{equation}

This gives two natural maps 
$(E,H) \mapsto (E \times_{X} E^{\#})$
and $(E^{\#},H^{\#}) \mapsto (E \times_{X} E^{\#}).$
By definition of fiber product of bundles, 
changing $(E,H)$ by an isomorphism of pairs 
yields an isomorphic fiber product
and the same for $(E^{\#}, H^{\#}).$

Thus, if we could construct $(E \times_{X} E^{\#})$
from the pair $(E,H)$ alone we would 
get a natural assignment $[(E,H)] \to E \times_{X} E^{\#}$
which should extend to a functor 
$\mathcal{T}:\HC \to \mathbf{Top}.$ 
Then the action of Topological T-duality for 
circle bundles would interchange $[(E,H)]$ and 
$[(E^\#,H^{\#})]$ since it acts
on the category $\HC$ and also
on the subcategory $\pi(\F_X)$
of $\HC$ via a $\KZ_2-$action, 
Hence the Topological T-duality autmorphism acting on 
$\pi(\F_X)$ would induce the self-homeomorphism
of $E \times_{X} E^{\#}$ which interchanges $E$ 
and $E^{\#}.$

It is not clear that such a functor exists because the space 
$E \times_{X} E^{\#}$ involves the T-dual bundle 
$E^{\#}$ and it is difficult
to construct $E^{\#}$ directly from the pair $(E,H).$ 
In Ref.\ \cite{MaWu} Mathai and Wu construct the space
$(E \times_{X} E^{\#})$ directly from the pair $(E,H)$
using equivariant cohomology. Their method is very
general since they were studying spaces with an 
arbitrary circle action, and the case of principal
circle bundles studied in this paper follows trivially.
The argument in the previous paragraph then
shows that the functor $\mathcal{T}: \HC \to \mathbf{Top}$ 
described in the paragraph before
Eq.\ (\ref{EqPairIso}) actually exists. 

In Refs.\ \cite{HullZ} 
it is argued that the doubled space $\mathcal{T}([(E,H)])$
is a 'generalized' space
associated to $(E \to X,H).$ In the above, this space is 
just the topological space which is 
the image of $\pi(E \to X,H)$ above in $\HC$
under ${\mathcal T}.$

When $E$ is not a principal circle bundle but a principal $\KT^n-$bundle, however, with
$n \geq 2,$ it is possible for the T-dual to be a nonocommutative space. In this case, 
the correspondence space cannot be a topological space
and must be defined differently, for example as a
non-type-I $C^{\ast}-$algebra.  
\subsection{T-folds and $\C$ \label{SSecTFTF}}
We had mentioned T-folds at the beginning of
this section. As we had described there, T-folds
are generalized spaces which
replace the doubled geometry when the T-dual
is nongeometric.

In Ref.\ \cite{HullT} the author
showed that this doubled space gives a
T-fold when the original string theory background
is nongeometric. The author showed in this case
the doubled space was obtained by gluing the 
toroidially compactified background to its T-dual after a 
T-duality using `T-duality transition functions'. 
Thus the T-fold associated to a spacetime background
with a nongeometric T-dual
and the doubled field theory propagating
on it is a natural reduction of closed String Field
Theory on the given background.

We argue that the argument in this section can
be easily extended to string theory backgrounds which
are smooth principal $\KT^n-$bundles with a smooth
Riemannian metric. 
We restrict to backgrounds which were 
Riemannian manifolds with a free smooth
circle action. The results of Belov, Hull and Minasian in
Ref.\ \cite{Minasian} calculate the Topological T-dual
of such backgrounds.
The results in Ref.\ \cite{Minasian} show
that an analogue of the correspondence space may be constructed
for all these backgrounds even if the
$H-$fluxes on these backgrounds are non-T-dualizable.
In this case the analogue of the correspondence space is
{\em not} a principal circle bundle but only a non-principal
circle fibration and there is no
T-dual principal circle bundle. We could define the
value of $\T$ on spaces with non-dualizable $H-$flux
to be this non-principal circle fibration. Examining the
argument in Ref.\ \cite{Minasian}, it is clear that the functor
$\T$ extends to these spaces. In addition, the topological
T-duality group acts on the circle fibration $\T(e)$ with
$e \in \HC$ by interchanging circle fibers even if 
$\T(e)$ is a circle fibration and not a principal circle
bundle.

It is clear that it is possible to
construct the functor $\T$ above restricted to
pairs whose underlying topological spaces
were the topological spaces obtained from
the Riemannian manifolds with a free smooth
circle action after forgetting the smooth structure.
In addition the functor $\T$ extends to such spaces
which possess non-dualizable-$H-$flux.

As a result we may view the functor $\T$ as defining
a T-fold geometry for spaces with non-dualizable $H-$flux.
\section{Conclusion \label{SecLast}}
In SubSec.\ (\ref{SSecLI}) we briefly
overview and comment on results obtained in each
Section of this paper. In SubSec.\ (\ref{SSecLC})
we describe our conclusions.
\subsubsection{Introduction \label{SSecLI}}
In this paper we have described a new approach to study Topological T-duality for circle bundles using SFT. We have constructed the category
of pairs $\C$ using arguments from SFT. We have used the theory
of categorical equivalence to construct a functor which assigns
to each topological space $X$ the moduli space of SFT gauge
equivalence classes of pairs over $X.$ We have used this moduli
space to construct groupoid valued functors
$P_k: \mathbf{CW}^{op} \to \mathbf{Grpd}.$
In addition, we have demonstrated that $P_0$ may be identified with the Topological T-duality functor $P$ of Bunke et al and $P_1$ is a new invariant.
Also we have shown that $P_k(X), k \geq 2$ are trivial.
We have also demonstrated that the arrows in $P_1(X)$ are
naturally induced from
the action of the mapping class group $MCG(X)$ on isomorphism
classes of pairs over $X.$ We have also argued that as a consequence
of this that the symmetry
group of a knot tamely embedded in $S^3$ acts on isomorphism 
classes of pairs over the knot complement.

We have argued above that the category of pairs over $X$ namely $\C_X$ 
gives a natural description of toroidally compactified backgrounds over $X$
which are principal circle bundles over $X$ and related this category to 
the SFT description (Ref.\ \cite{Horowitz}) of these backgrounds. 
The functors $P_i(X)$ label these backgrounds and their automorphisms.

We had argued above that the functors $P_i(X)$ are a geometrization of the Topological T-duality
for Triples described in Ref.\ \cite{Pan2}. We argue that there is a natural map from 
$P_0 \times P_1$ to $P_{3,2}$ given by an automorphism of the  natural gerbe associated to a pair.

We had also given a definition of a T-fold within the formalism of Topological T-duality. We had argued
that for circle bundles this is just the assignment of the correspondence space to a pair together with its Topological
T-dual.
\subsubsection{Conclusions \label{SSecLC}}
We now outline some natural extensions of the
formalism of this paper.
\flushleft{\bf{Non-Circle Fiber:}}
Is it possible to extend the construction of 
$\C$ in this paper to other spaces including 
circle fibrations, $T^n-$bundles,
fiber bundles with other fibers apart from the circle,
and also to spaces with a circle action but which might
not be principal circle bundles?
It is clear that owing to its flexibility 
SFT would be helpful for defining such
extensions. 

Functors analogous to Bunke-Schick's functor $P$ have been
defined for each of these spaces. Can these
functors be obtained from the 
homotopy theory of various natural moduli spaces in these
problems in a manner similar to the functors $P_k(X)$ 
above?

For example in Ref.\ \cite{BEM} Topological T-duality 
has been investigated for bundles
with spheres as fibers and it might be possible to apply the
methods of Ref.\ \cite{BEM} to this situation.

\flushleft{\bf{Category of Elements:}}
We had argued in SubSec.\ (\ref{SSecPCLift}) above
(see discussion after Thm.\ (\ref{ThmHCTEqv} )
that the fibration $E: \HC \to \mathbf{Top}$ is
the category of elements of Bunke-Schick's functor 
$P:\mathbf{Top}^{op} \to \mathbf{Set}$ and
hence geometrizes $P.$

It would be interesting to see if some analogue of this
construction exists for the spaces which are fibered
over a base with fibers which are not $S^1$ as
described in the previous paragraphs.
This is important since functors analogous to 
$P$ have been defined for most of these spaces
in the Topological T-duality literature. It would
be interesting to see if these functors could be
geometrized in the sense of this paper.

If such a geometrization is possible, 
It would also be interesting to see if
the construction of the small category 
$\pi(\F_X)$ could be generalized to 
such situations since the underlying set of the small 
category $\pi(\F_X)$ is the value of Bunke-Schick's 
functor $P(X)$ and as mentioned above analogues
of $P(X)$ have already been defined for these spaces.
Is there a general pattern in this construction
as the fiber space changes? 

In addtion the analogue of the construction of the
Higher Topological T-duality functors $P_k$ should be
studied for such spaces and it should be checked
which of these are always trivial as in this paper.
It should be possible to
see if these higher functors can be geometrized in the
sense of this paper. This is of obvious interest, since, for example,
many important three-manifolds are Riemann surface
fibrations over a circle, but such spaces cannot be viewed
as pairs in the sense of this paper. 
\flushleft{\bf{Arbitrary $S^1-$spaces:}}
String Theory backgrounds with a circle
isometry with arbitrary fixed sets of the circle isometry
are exactly all the possible string backgrounds which are 
compactified on a single circle direction whose
noncompactified directions have $X$ as an underlying
topological space. Ideas 
from SFT will help in analyzing the collection of all
such backgrounds.

Also if the construction of the moduli space of pairs
$G(X)$ mentioned above in Sub-Sec.\ (\ref{SSecGXCX})
could be generalized to the category of {\em topological
stacks} then such a theory would 
include all String Theory backgrounds with a circle isometry
with arbitrary fixed sets of the circle action. There
is already a well-defined theory of Topological T-duality
for topological stacks (\cite{BunkeS2,Pan3}).

It is also likely that this theory could then be extended to all 
$S^1-$fibrations over $X$ since a theory of Topological 
T-duality is already known for backgrounds which are 
non-principal circle fibrations over $X$  (see 
Ref.\ \cite{MR2}).

\flushleft{\bf {Noncommutative T-duals:}}
In Ref.\ \cite{BunkeS2, MR1} show that for principal bundles
with $\KT^n-$fibers under certain
conditions there might be no T-dual space and the
$C^{\ast}-$algebraic T-dual might be a
noncommutative space.  If the construction of the moduli
space$G(X)$ described in this paper were generalized
to pairs consisting of principal $\KT^n-$bundles over 
$X$ with $H-$flux described in Ref.\ \cite{BunkeS2, MR1} it
might be possible to study noncommutative 
T-duals using a formalism similar to the one
described in this paper.

\flushleft{\bf{Maps which are not isomorphisms on base:}}
It would be interesting to see if the category 
$\pi(\F_X)$ could
be generalized to a larger category such that nontrivial
higher topological T-duality functors $P_k(X), k \geq 2$ 
may be obtained. This might be possible if we use the
construction described near the end of 
SubSec.\ (\ref{SecPropP0P1}) above
involving mapping cones of maps $f:X \to X.$

It would be interesting to see if the Thomason cohomology
groups could be calculated for the moduli space $G(X)$ 
as outlined in SubSec.\ (\ref{SecPropP0P1}) above.

\flushleft{\bf{Non-Type II backgrounds:}}
All the above was for backgrounds for closed bosonic SFT. These
correspond to backgrounds on which the closed bosonic string
propagates. In this paper we identify these backgrounds with
backgrounds on which Type II Strings propagate since their
massless modes are the same (since we are ignoring
supersymmetry). The type II Superstring field theory
induces supersymmetric Type II
theories on all its backgrounds and the formalism 
of this paper should apply to this SFT as well.

However there are variants of SFT which induce the other 
known string theories on all their backgrounds (for example
Type I String Theory or Heterotic String
Theory). It is interesting to ask if the formalism in this 
paper may be reworked for backgrounds on which other 
types of String Theory are
propagating. If so, it might be possible to construct 
formalisms similar
to Topological T-duality for other dualities in the duality web.

\flushleft{\bf{Mapping Class Group of Base:}}
In addition, we showed in this paper that elements of 
$P_1(X)$ correspond
to elements of the mapping class group of
a space which lift to automorphisms of the pair. When the 
space is a two-manifold we used a set of results by Chen 
and Tshishiku \cite{ChenTs} to study the lift of isotopy 
classes of self homeomorphisms of the base to the pair. 

We have also argued that this might happen when the space 
was a knot complement. In this case sometimes 
we might be able to obtain a pair
over the knot complement which had the symmetries of the
knot as its automorphisms covering isotopy classes of 
self-homeomorphisms the base. 
These isotopy classes of self-homeomorphisms of the base
would then be the symmetries of the knot. These pairs would
correspond to certain elements of $P_1(X).$ It
is interesting to ask when such a pair would exist
and also which isotopies of the base
would lift and which would not. This
calculation would need  results from the theory 
of three-manifolds and knot theory.
\section{Acknowledgements}
I thank the School of Arts and Sciences, Ahmedabad University for providing a congenial
working environment while the bulk of
this work was being done.
I thank the School of Mathematics, NISER, Bhubaneshwar for providing a congenial working environment for the first draft of this paper.
I thank Professor Jonathan Rosenberg, University of Maryland at College Park for his advice and encouragement during the writing of this paper.

\providecommand{\href}[2]{#2}

\end{document}